\documentclass[11pt, a4paper, twoside]{article}
\usepackage[utf8]{inputenc}
\usepackage[english]{babel}
\usepackage{amsmath}
\usepackage{amsthm}
\usepackage{amsfonts}
\usepackage{amssymb} 
\usepackage{amsthm}
\usepackage{mathtools}
\usepackage{pgfplots}
\usepackage{caption}
\usepackage{subcaption}
\usepackage[margin=1in]{geometry}
\usepgfplotslibrary{fillbetween}
\usetikzlibrary{fadings}
\usetikzlibrary {decorations.markings}
\usepackage{bbm}
\usepackage[shortlabels]{enumitem}
\usepackage{float}
\usepackage{stackengine}
\usepackage{accents}
\usepackage{appendix}
\usepackage{xcolor}
\definecolor{darkred}{rgb}{0.5,0,0}
\definecolor{lightblue}{rgb}{0,0.4,0.8}
\definecolor{darkgreen}{rgb}{0,0.5,0}
\usepackage{listings} 
\usepackage{xfrac}
\usepackage{comment}
\usepackage[export]{adjustbox}

\newcounter{sideremark}

\definecolor{dkgreen}{rgb}{0,0.6,0}
\definecolor{gray}{rgb}{0.5,0.5,0.5}
\definecolor{mauve}{rgb}{0.58,0,0.82}

\usepackage[citecolor=blue]{hyperref}
\hypersetup{
	colorlinks=true,
	linkcolor=blue,
	urlcolor=blue,
}
\usepackage{graphicx}
\hypersetup{colorlinks=true, linktoc=all, linkcolor=blue}
\usetikzlibrary{arrows.meta}

\pgfplotsset{compat=newest}
\newcommand\smallO{
	\mathchoice
	{{\scriptstyle\mathcal{O}}}
	{{\scriptstyle\mathcal{O}}}
	{{\scriptscriptstyle\mathcal{O}}}
	{\scalebox{.7}{$\scriptscriptstyle\mathcal{O}$}}
}
\DeclareRobustCommand{\bigO}{%
	\text{\usefont{OMS}{cmsy}{m}{n}O}%
}
\newcommand*{\defeq}{\mathrel{\vcenter{\baselineskip0.5ex \lineskiplimit0pt
			\hbox{\scriptsize.}\hbox{\scriptsize.}}}=}
\newcommand*{\eqdef}{=\mathrel{\vcenter{\baselineskip0.5ex\lineskiplimit0pt
			\hbox{\scriptsize.}\hbox{\scriptsize.} }}}
\newcommand{\ssup}[1]{{\scriptscriptstyle{({#1})}}}
\newcommand{\eps}{\varepsilon}
\newcommand{\NN}{\mathbb{N}}
\newcommand{\NNN}{\overline{\mathbb{N}}}
\newcommand{\RR}{\mathbb{R}}
\newcommand{\EE}{\mathbb{E}}
\newcommand{\PP}{\mathbb{P}}
\newcommand{\TT}{\mathcal{T}}
\newcommand{\BT}{\overline{\mathcal{T}}}
\newcommand{\FF}{\mathcal{F}}
\newcommand{\BF}{\overline{\mathcal{F}}}
\newcommand{\HH}{\mathcal{H}}
\newcommand{\DD}{\mathcal{D}}

\newcommand{\PGF}{\mathcal{G}}
\newcommand{\PGFd}{\overline{\mathcal{G}}}
\newcommand{\Val}{\mathcal{V}}
\newcommand{\YY}{\mathbf{Y}}
\newcommand{\XX}{\mathbf{X}}
\newcommand{\VV}{\mathbf{V}}
\newcommand{\DZ}{\mathbf{Z}}
\newcommand{\DF}{\mathcal{B}}

\newcommand{\II}{\mathbbm{1}}

\newcommand{\as}{\:\text{a.s.}}

\DeclareMathOperator*{\esup}{ess\,sup} 
\DeclareMathOperator{\supp}{supp}
\DeclareMathOperator{\Geom}{Geom}
\DeclareMathOperator{\Unif}{Unif}
\DeclareMathOperator{\Ber}{Ber}

\DeclareMathOperator{\Beta}{B}
\DeclareMathOperator{\IHR}{IHR}
\DeclareMathOperator{\DHR}{DHR}
\DeclareMathOperator{\IHRE}{IHRE}
\DeclareMathOperator{\DHRE}{DHRE}
\DeclareMathOperator{\HIHRE}{HIHRE}
\DeclareMathOperator{\HDHRE}{HDHRE}
\DeclareMathOperator{\NBU}{NBU}
\DeclareMathOperator{\NWU}{NWU}
\DeclareMathOperator{\NBUE}{NBUE}
\DeclareMathOperator{\NWUE}{NWUE}
\DeclareMathOperator{\HNBUE}{HNBUE}
\DeclareMathOperator{\HNWUE}{HNWUE}
\DeclareMathOperator{\Var}{Var}

\usepackage[capitalize]{cleveref}
\crefname{equation}{}{}
\newtheorem{theorem}{Theorem}[section]
\newtheorem{proposition}{Proposition}[section]
\newtheorem{lemma}{Lemma}[section]
\newtheorem{corollary}{Corollary}[section]
\newtheorem{remark}{Remark}[section]
\newtheorem{definition}{Definition}[section]
\numberwithin{equation}{section}
\title{IID Prophet Inequality with Random Horizon:\\ Going Beyond Increasing Hazard Rates}
\author{Giordano Giambartolomei\setcounter{footnote}{1}\thanks{King's College London.}, Frederik Mallmann-Trenn\footnotemark[2] and Raimundo Saona\setcounter{footnote}{7}\thanks{Institute of Science and Technology Austria.}
}
\date{}

\begin{document}
        \maketitle
        \thispagestyle{empty}
        \begin{abstract}
        	Prophet inequalities are a central object of study in optimal stopping theory. In the \textit{iid} model, a gambler sees values in an online fashion, sampled independently from a given distribution. 
        	Upon observing each value, the gambler either accepts it as a reward or irrevocably rejects it and proceeds to observe the next value. The goal of the gambler, who cannot see the future, is maximising the expected value of the reward while competing against the expectation of a prophet (the offline maximum). 
        	In other words, one seeks to maximise the gambler-to-prophet ratio of the expectations. 
        	
        	This model has been studied with infinite, finite and unknown number of values. 
        	When the gambler faces a random number of values, the model is said to have random horizon. 
        	We consider the model in which the gambler is given \textit{a priori} knowledge of the horizon's distribution. 
        	Alijani~et~al.~(2020) designed a single-threshold algorithm achieving a ratio of $\sfrac{1}{2}$ when the random horizon has an increasing hazard rate and is independent of the values. 
        	We prove that with a single threshold, a ratio of $\sfrac{1}{2}$ is actually achievable for several larger classes of horizon distributions, with the largest being known as the $\PGF$ class in reliability theory. 
        	Moreover, we show that this does not extend to its dual, the $\PGFd$ class (which includes the decreasing hazard rate class), while it can be extended to low-variance horizons. 
        	Finally, we construct the first example of a family of horizons, for which multiple thresholds are necessary to achieve a nonzero ratio. We establish that the Secretary Problem optimal stopping rule provides one such algorithm, paving the way towards the study of the model beyond single-threshold algorithms.
        \end{abstract}
        \clearpage
         \thispagestyle{empty}
        \setcounter{tocdepth}{1} 
        \tableofcontents
       \clearpage
       \pagenumbering{arabic}
    \setcounter{footnote}{0}
	\section{Introduction}
	Prophet inequalities are a central object of study in optimal stopping theory. A gambler sees nonnegative values in an online fashion, sampled from an instance of independent random variables $\{X_i\}$ with known distributions $\{\mathcal{D}_i\}$, in adversarial, random or selected order, depending on the particular model. When observing each value, the gambler either accepts it as a reward, or irrevocably rejects it and proceeds with observing the next value. The goal of the gambler, who cannot see the future, is to maximise the expected value of the reward while competing against the expectation of a prophet (out of metaphor, the offline maximum or supremum, depending on whether the instance is finite or not). In other words, one seeks to maximise the gambler-to-prophet ratio of the expectations.
	
	The gambler represents any online decision maker, such as an algorithm or stopping rule. Probabilistically, we will refer to it as a \textit{stopping time} $\tau$. The term \textit{online} implies that the gambler, unable to see the future, will always stop at a time $\tau$ such that the event $\{\tau=i\}$ depends, informally speaking, only on the first $i$ values observed.
	
	Due to the online nature of prophet inequalities, some terminology from \textit{competitive analysis} is usually borrowed. Somewhat informally, it is common to refer to a worst-case gambler-to-prophet ratio of the algorithm (that is a ratio known to be achievable for any given instance by the algorithm) as \textit{competitive ratio}. Let $c\in[0,1]$ be a real constant. An algorithm is said to be $c$-competitive (or equivalently, a $\sfrac{1}{c}$-approximation) if it has a competitive ratio of $c$. 
	An upper bound on any algorithm's highest possible competitive ratio for a given instance will be called \textit{hardness} of the instance. Saying that a prophet inequality model has a hardness of $c$ (or equivalently is $c$-hard ), means that there is a $c$-hard instance for that problem. 
	A hardness for a model is said to be \textit{tight} when it is matched by the competitive ratio of an algorithm solving it. Similarly, the competitive ratio $c$ of an algorithm solving a model is tight when it is matched by a hardness of $c$ for that model (equivalently, the algorithm is a tight $\sfrac{1}{c}$-approximation).
	
	\subsection{Prophet inequality models}\label{model}
	To date, several models and extensions of prophet inequalities are present in the literature. In this section we introduce the classical model and describe two variants of it, briefly reviewing the state-of-the-art concerning their hardness and the competitiveness of known algorithms.
	
	\paragraph{Classical Prophet Inequality.}
	The very first prophet inequality model, typically referred to as classical (or adversarial) Prophet Inequality (PI), is due to Krengel and Sucheston \cite{KrenSuch77,KrenSuch78}. The given instance is composed of  countably (possibly infinitely) many integrable independent nonnegative random variables $\{X_i\}$ with known distributions $\{\mathcal{D}_i\}$ in a fixed given order, usually referred to as \emph{adversarial} order. 
	As per the general dynamics previously described, the gambler observes in an online fashion the sequence of sampled values in the fixed order. Formally, the goal is to maximise the gambler-to-prophet ratio $\EE X_\tau$ over $\EE\sup_iX_i$. When working with \textit{infinite} instances, it is customary to consider only finite stopping rules and and variables such that $\EE\sup_iX_i<\infty$.\footnote{Formally, we maximise the gambler-to-prophet ratio with respect to all stopping times $\tau$ such that $\PP(\tau<\infty)=1$.} In~\cite{KrenSuch78}, it was shown that the $\sfrac{1}{2}$-hardness of PI (shown by Garlin \cite{KrenSuch77}) is tight. Later, a single-threshold $2$-approximation was constructed in \cite{Sam84}.
	
	\paragraph{IID Prophet Inequality.}
	Denoted as IID PI, it is a specialised model of PI, where the random variables $\{X_i\}$ are further assumed independent and identically distributed (\textit{iid}) according to a given distribution $\DD$. The hardness of this problem has been shown to be $\sfrac{1}{\beta}\approx 0.7451$, where $\beta\in\mathbb{R}$ is the unique solution to $\int_0^1 \frac{dx}{x(1-\log x)+\beta-1} = 1$ in~\cite{Kertz86}. The $\beta$-approximation \textit{quantile strategy} devised in~\cite{Correa17} shows that the aforementioned hardness is tight. 
	
	\paragraph{IID Prophet Inequality with Random Horizon.}
	The IID PI with Random Horizon (RH) was first introduced in \cite{HajKleSan07}. Consider a random variable $H$, which we will call the (random) \textit{horizon}, with given discrete distribution $\HH$, that is, supported on an arbitrary subset of $\NN$. This assumption comes with no loss of generality (formally, it rules out that the horizon has mass at zero). $H$ will be assumed finite ($\PP(H<\infty)=1$) and integrable ($\EE H <\infty$), which we denote $H\in\mathcal{L}^1(\Omega)$. RH is a relaxation of the IID PI, considering integrable \textit{iid} random variables $X_1,\ldots,X_H$, with given distribution $\DD$ independent of $\HH$. This setup, supported on a probability space $(\Omega,\FF,\PP)$, models the game where the gambler, facing an unknown number of values $X_1,\ldots,X_H$, can still use \textit{a priori} knowledge of $\HH$ when maximising returns. More precisely, denote $[n]=\{1,\ldots,n\}$. The goal is to maximise the gambler-to-prophet ratio of $\EE X_\tau$ over $\EE M_H$, where $M_H\defeq \max_{i\in[H]}X_i$ and both expectations run over the randomness of the \textit{iid} copies of $X\sim\DD$ and $H\sim\HH$. The gambler must select the value in ignorance of whether it is the last or not. If the gambler fails to stop by the time the last value has been inspected, the return is zero. This ignorance, added to the usual nonanticipative constraint, yields a no easier model than the IID PI. On top of the usual disadvantage of playing against a prophet, which knows all future realisations of the values, and chooses the largest; the gambler now competes against a prophet, which can also foresee the random number of values. We leave more in depth measure-theoretic details of the model for \Cref{opt}. 
	
	\subsection{Previous bounds on the IID Prophet Inequality with Random Horizon}\label{bounds}
	No constant-approximations are possible for RH \cite[Theorem~1.6]{AliBanGolMunWan20}. 
	Distributional knowledge of the horizon is not enough to guarantee, in expectation, that the gambler's return can attain a worst-case constant multiple of the prophet's return. Therefore, we impose additional restrictions on the distribution of the horizon.  
	The largest distributional class for which a constant-approximation has been found so far is that of horizons with increasing hazard rates $\lambda(h)$. For simplicity, we will use \textit{increasing} and \textit{decreasing} instead of \textit{nondecreasing} and \textit{nonincreasing} respectively. Recall that the hazard rate $\lambda(h)$ of a horizon $H$ is defined as follows. Denote the survival function $S(h)=\PP(H\ge h)$. If $|\supp(H)|=\infty$, for every $h\in\NN$, $\lambda(h)\defeq \PP(H=h)/S(h)$. If $|\supp(H)|<\infty$, for every $h\in\NN$, $\lambda(h)$ is defined analogously for all $ h\le\sup\supp(H)$, whereas it is set to $1$ for all $h>\sup\supp(H)$. 
	\begin{definition}[$\IHR$ class]
		$H$ is \emph{Increasing Hazard Rate ($\IHR$)} if for every $h\in\NN$, $\lambda(h)\le\lambda(h+1)$.
	\end{definition}
	The \textit{dual} class is defined by reversing the inequality on the support of the horizon $\supp(H)$.
	\begin{definition}[$\DHR$ class]
		$H$ is \emph{Decreasing Hazard Rate ($\DHR$)} if for every $h\ge\inf\supp(H)$, $\lambda(h)\ge\lambda(h+1)$.
	\end{definition}
	The geometric distribution is the only discrete distribution that has constant hazard rate, and is therefore the only discrete distribution that belongs to both classes.
	
	The state-of-the-art for RH consists in the findings of \cite{AliBanGolMunWan20}. In \cite[Theorem~3.2]{AliBanGolMunWan20}, through the use of \emph{second order stochastic dominance}, it is shown that if $H\in\IHR$, $\EE X_\tau\ge  (2-\sfrac{1}{\mu})^{-1} \EE M_H$, where $\mu\defeq\EE H\ge 1$. This ensures a (uniform) $2$-approximation on the $\IHR$ class.\footnote{For brevity, outside formal statements, we will often omit the dependence on distributional parameters when referring to competitive ratios. For example, in this sentence \textit{$2$-approximation} means, more precisely, $2-\sfrac{1}{\mu}$-approximations.}  Remarkably, this is achieved by a single-threshold algorithm \cite[Theorem~3.1]{AliBanGolMunWan20}, and it is tight \cite[Theorem~3.5]{AliBanGolMunWan20}, since if $H$ is geometrically distributed, a parametric family of two-point distributions can be shown to make the prophet-to-gambler ratio approach $2-\sfrac{1}{\mu}$, as the parameters vary suitably. The emphasis on monotone hazard rate classes stems from \cite{HajKleSan07}, which is the first work that considers RH, to the best of our knowledge.
	
	\subsection{Our contributions}\label{contributions}
	Implicitly, the motivation behind the approach that considers monotone hazard rate classes stems from interpreting RH from a reliability theory point of view: an optimal stopping game under evolving risk (or ageing) that it ends by the next step. In reliability theory several classes of distribution are used to model this aspect. However, despite more than $15$ years have passed, no progress has been made on classes larger than the $\IHR$ class. In this paper we contribute to the study of RH by:
	\begin{itemize}[noitemsep]
		\item Extending existing characterisations of the optimal algorithm to unbounded horizons.
		\item Complementing and extending key ideas from \cite{AliBanGolMunWan20} to yield the existence of single-threshold $2$-approximations on important superclasses of the $\IHR$ class, culminating in the $\PGF$ class.
		\item Showing that the $\PGFd$ class (dual of the $\PGF$ class) admits hard horizons for single-threshold algorithms..
		\item Deriving single-threshold constant-approximations for sufficiently concentrated horizons with finite second moments.
		\item Showing that RH admits families of horizons for which single-threshold algorithms are not competitive, despite the Secretary Problem (SP) optimal stopping rule providing a constant-approximation (see \Cref{related} for details on SP and its variants).
	\end{itemize}
	
	In this section we give formal statements of the above results, with the exception of \cref{optalg} in \Cref{opt}, due to the overall degree of technicalities involved. Informally speaking, this theorem characterises the optimal algorithm in terms of a discounted infinite optimal stopping problem, with a special focus on unbounded horizons. As a technical tool, the equivalence of stopping rules for RH with stopping rules for the discounted problem is leveraged not only to obtain hardness results for single-threshold algorithms, but also to formalise heuristic arguments involving the optimal algorithm under geometric horizon. More specifically: we exploit the equivalence between stopping rules for RH and corresponding rules for the discounted problem to extend hardness to neighbouring instances; we show that with a geometric horizon, if the value distribution is bounded, the optimal algorithm of RH is single-threshold, pinpointing the structural equation for the threshold. 
	
	The next three results involve a stochastic order, which we now introduce. Standard facts regarding stochastic orders are often relied upon throughout this work. They will be recalled as remarks, whose proof can be found in~\cite{ShakShan07}. In the following definitions, $H$ and $G$ always refer to horizons. 
	\begin{definition}[Probability Generating Function Order]
		Given two horizons $H$, $G$, we say that $G$ is dominated by $H$ in the \emph{probability generating function (pgf) order}, and denote it as $G\prec_{pgf} H$, if for every $t\in(0,1)$, $\EE\left( t^G\right)\ge\EE\left( t^H\right)$.
	\end{definition}
	\begin{remark}
		$G\prec_{pgf} H$ if and only if for all $t\in(0,1)$, $\sum_{i=1}^\infty S_G(i)t^i\le\sum_{i=1}^\infty S_H(i)t^i$.
	\end{remark}
	The $\PGF$ class, consisting of every positive discrete distribution that dominates, in the pgf order, the ($\NN$-valued) geometric distribution with the same mean, has been introduced in \cite{Klef83a}.
	\begin{definition}[$\PGF$ class]\label{Gdef}
		The $\PGF$ class consists of every horizon that dominates, in the pgf order, the geometric distribution with the same mean: $\PGF\defeq\{H:\, G\prec_{pgf} H,\, G\sim\Geom(\sfrac{1}{\EE H})\}$.
	\end{definition}
	The dual class is obtained by reversing the pgf ordering, and is denoted as $\PGFd$. It is well known that $\IHR\subset\PGF$, with the inclusion being strict. Similarly, $\DHR\subset\PGFd$. In applications of reliability theory, the classes within $\mathcal{G}$ and its dual are a staple of modelling many aspects of ageing and more in general of an evolving risk of failure.
	\begin{theorem}\label{Ghorizon}
		For every $H\in\PGF$ with $\EE H\eqdef\mu$, a single-threshold $2-\sfrac{1}{\mu}$-approximation exists.
	\end{theorem}
	From \cite[Theorem~3.5]{AliBanGolMunWan20} it follows that any $2$-approximation on the $\PGF$ class is tight. More specifically, the $2$-approximation yielding \Cref{Ghorizon} is essentially\footnote{The argument and the algorithm provided in \cite{AliBanGolMunWan20}, although valuable and insightful, are somewhat informal and require some tweaks and filling some gaps. A more in depth discussion is given in \Cref{Gclass}, where we show that RH is amenable enough to allow for stochastic tie-breaking, a standard techniques which combines well with the aforementioned tweaks.} the same as the one used on the $\IHR$ class in \cite[Theorem~3.1]{AliBanGolMunWan20}. The key feature of the stopping rule is determining the value $p$ such that $\PP(X>p)=\sfrac{1}{\mu}$, and then accepting the first value exceeding $p$. The $\PGF$ class may be somewhat abstract, compared to its subclasses having immediate intuitive interpretations in terms of ageing concepts, such as the $\IHR$ class. Its possible interpretations are quite general and go beyond mere ageing, as can be read in \cite{Klef83a}. To gain some intuition of how much more general of a class it is, the reader is referred to \Cref{Gclass}, where we will describe some of its largest most well-known subclasses, as evidence that \Cref{Ghorizon} generalises significantly \cite[Theorem~3.1,~Theorem~3.2]{AliBanGolMunWan20}.\footnote{It may be beneficial to give a concrete example of a distribution that, although elementary, could not be handled without our extension. Let $H$ with $S(h)$ given as $1, \sfrac{1}{2}, \sfrac{1}{4}, \sfrac{1}{5}$ for $h=1,2,3,4$ respectively and $0$ for all $h>4$. The list of the corresponding hazard rates $\lambda(h)$, $\sfrac{1}{2}, \sfrac{1}{2}, \sfrac{1}{5}, 1,\ldots$, is not monotone increasing because $\lambda(3)$ drops. The $\PGF$-class condition is, however, satisfied, thus our result can ensure a $2$-approximation (to be specific, a $2-\sfrac{20}{39}\approx1.49$-approximation). It is straightforward to compute $\EE H =\sfrac{39}{20}$, so \Cref{Gdef} requires $\EE t^G\ge\EE t^H$ for all $0<t<1$, with $G\sim\Geom\left(\sfrac{20}{39}\right)$, and the two pgf's being respectively $20t/(39-19t)$ and $\sfrac{t}{2}+\sfrac{t^2}{4}+\sfrac{t^3}{20}+\sfrac{t^5}{5}$.} We conclude by mentioning that statistical tests for the $\mathcal{L}$ class (the continuous equivalent of the $\PGF$ class) have also been developed \cite{Kle83b}. The possibility of adapting them into tests for the $\PGF$ class adds to the practical significance of the result.
	
	Our next result requires introducing the stopping rule for SP with deterministic horizon $m$. This consists of rejecting the first $r-1$ values, and subsequently accepting the first value ranked better than all the previous ones. 
	The waiting time $r_m\sim\sfrac{m}{e}$ as $m\longrightarrow\infty$ yields the optimal stopping rule (see \Cref{related} for details). In the following, we provide the first example of a family of horizons, for which no single-threshold constant-approximation is possible, and yet an adaptation of the SP stopping rule to random horizons is proven to provide a constant-approximation. The family of horizons we will consider is of the form $\{H_m\}_{m\ge M}$, with pmf's parametrised by $\eps>0$, $M\in\NN$ large enough and $m\defeq \sup\supp(H_m)<\infty$, and will therefore be denoted as $\mathcal{H}_M(\eps)$. Let $\zeta_{r_m}$ be the (optimal) stopping rule for SP with deterministic horizon $m$. Taking the minimum $\tau_m\defeq\zeta_{r_m}\wedge(H_m+1)$ produces a natural adaptation of the SP stopping rule to the random horizon $H_m$.
	\begin{theorem}\label{future}
		There exists a family of horizons $\mathcal{H}_M(\eps)$, such that for all fixed $\eps$ small enough and arbitrary $M$, no single-threshold constant-approximation is possible. Nonetheless, for every fixed $\eps$, no matter how small, the random horizon SP rule $\tau_m$ is an approximate constant-approximation, as long as $M=M(\eps)$ is fixed large enough.
	\end{theorem}
	For brevity, we say that the family of horizons $\mathcal{H}_M(\eps)$ is \textit{hard} for single-threshold algorithms, but the random horizon SP rule is \textit{competitive} on it.\footnote{We omit the qualifier \textit{approximate} appearing in the statement of \Cref{future},  a technicality which we now clarify. \textit{Approximate} refers to the fact that the competitive ratio ensured, which is a positive, increasing function of $\eps$, denoted as $g(\eps)$, is approached as $M\longrightarrow\infty$. At the large $\eps$ regime, $g(\eps)\approx\sfrac{1}{e^2}$. As $\eps$ grows, the family of horizons $\mathcal{H}_M(\eps)$ admits constant-competitive single-threshold algorithms too. However, at small $\eps$ regime, these algorithms fail, whereas $\tau_m$ guarantees approximately a $g(\eps)$ fraction of $\EE M_{H_m}$. Specifically, for all $m\ge M$, a ratio of $g(\eps)-\delta$ is guaranteed, where $\delta=\delta(M)\longrightarrow 0^+$ as $M\longrightarrow\infty$. Hence the dependence $M=M(\eps)$ with $M(\eps)$ growing large as $\eps$ vanishes.} Considering the extensive number of classes of distributions for which single-threshold algorithms are competitive, this result is quite surprising, especially because the SP stopping rule only needs information regarding the relative ranks of the values, not the values themselves.
	
	Since we showed that single-threshold $2$-approximations exist for the $\PGF$ class, finding whether single-threshold constant-approximations exist for the $\PGFd$ class turns into an interesting problem. The equivalent characterisation of stopping rules for RH in terms of a discounted stopping problem, combined with a perturbation argument on $\mathcal{H}_M(\eps)$, yields a negative answer.
	\begin{theorem}\label{Gdhorizon}
		No single-threshold constant-approximation is possible on the $\PGFd$ class.
	\end{theorem}
	In \Cref{Gdclass} we will describe more in detail some of the largest most well-known subclasses of the $\PGFd$ class. They are no less important than the corresponding subclasses of the $\PGF$ class. For example, discrete $\DHR$ distributions are crucial in several areas: they have been shown to describe well the number of seasons a show is run before it gets cancelled and the number of periods until failure of a device (a system, or a component, that functions until first failure) governed by a continuous $\DHR$ life distribution in the grouped data case \cite{Langberg80}. \Cref{Gdhorizon} motivates our conjecture that single-threshold $2$-approximations are not possible on the $\DHR$ class either. Since our characterisation result of stopping rules for RH extends to unbounded horizons, it will be a crucial tool in proving this conjecture within our framework, as $\DHR$ horizons have support of the form $[a,\infty)\cap\NN$. Furthermore, suitably combining the techniques yielding \Cref{Ghorizon,Gdhorizon} with the fact that the tight geometric instance for $2$-approximations belongs to the $\PGF$ class, motivates our second conjecture that the $\PGF$ class is essentially optimal for 
	$2$-approximations on $\mathcal{L}^1$ horizons.
	
	In our last result, we go back to single-threshold algorithms, but consider only horizons $H$ satisfying th e sole condition $\EE( H^2)<\infty$, denoted $H\in\mathcal{L}^2(\Omega)$. Here no particular information is available regarding the classification of the distribution in terms of the $\PGF$. We show that it is possible to exploit concentration bounds, in order to ensure that the same single-threshold algorithm used for the $\PGF$ is still a $2$-approximation. In the following result, $W_0$ denotes the principal branch of the Lambert function.
	\begin{theorem}\label{L2horizon}
		For every $H$ having both $\mu\defeq \EE H$ and $\sigma^2\defeq\Var H$ finite, such that
	\begin{equation}\label{concentration}
        \sigma^2\le\mu^2\left[ 1-\frac{1}{\mu}+W_0\left(-\left(2-\frac{1}{\mu}\right)e^{-\left(2-\frac{1}{\mu}\right)}\right)\right],
        \end{equation}
		a single-threshold $2-\sfrac{1}{\mu}$-approximation exists.
	\end{theorem}
	The Lambert function is readily simulated in many languages, with well-studied error bounds. Therefore, one can reliably compute numerically the magnitude of concentration required by \Cref{concentration}, so that a $2$-approximation is ensured. To exemplify, consider the large-market limit, that is $\mu\longrightarrow\infty$, which is both interesting for practical applications involving horizons with large expectation, and simplifies computations. Under this hypothesis, taking square-root on both sides of \Cref{concentration} yields approximately $\text{CV}\defeq\sfrac{\sigma}{\mu}\le0.770$, where $\text{CV}$ is the \textit{coefficient of variation} of the horizon. 
	Variations over this procedure show that both weaker and stronger than $2$ constant-approximations are possible for all suitably concentrated horizons. Consider again the upper bound on $\text{CV}$ provided by square-rooting \Cref{concentration} in the large-market limit: $\text{CV}\le \sqrt{1+W_0(-2e^{-2})}$. Replacing $2$ with $C>2$ yields $\text{CV}\le \sqrt{C-1+W_0(-Ce^{-C})}$. This will ensure a large-market limit weaker $C$-approximation. As $x\ge1$ grows, $W_0(-xe^{-x})$ is negative strictly increasing (valued $-1$ at $x=1$) and vanishing. Thus the upper bound on the $\text{CV}$ ensuring a $C$-approximation for $C$ growing large scales roughly as $\sqrt{C-1}$, meaning that there will be, although progressively weaker, constant-approximations, as long as $\text{CV}<\sqrt{C-1}$.
	In the other direction, stronger $C$-approximations are ensured when taking \textit{admissible} $1<C<2$. When reducing $C$, we clearly cannot go past $e/(e-1)$, where the upper bound on the $\text{CV}$ vanishes, meaning that the horizon is required to progressively concentrate until it becomes degenerate. At this point our estimates hit the deterministic optimum for single-threshold IID PI, $1-\sfrac{1}{e}$ \cite{Esh18}. Following similar ideas, the following holds without large-market hypothesis.
	\begin{corollary}\label{largemarket}
		For every $H$ having both $\mu\defeq \EE H$ and $\sigma^2\defeq\Var H$ finite, and every constant $C\ge[1-\left(1-\sfrac{1}{\mu}\right)^{\mu}]^{-1}$, if \[\text{CV}\le \sqrt{C-1+W_0(-Ce^{-C})},\] then a single-threshold $C$-approximation exists.
	\end{corollary}
	At a closer look, it is not surprising that stronger approximations are possible within the class of sufficiently concentrated horizons. Although it is true that the geometric horizon is tight for $2$-approximations in $\mathcal{L}^2(\Omega)$ (which is the reason for which \Cref{L2horizon} is formulated in terms of them), the aforementioned properties of $W_0(-xe^{-x})$ readily imply that no geometric horizon satisfies \Cref{concentration}.\footnote{It might be beneficial to give another concrete example of a distribution that, although elementary and sufficiently concentrated for our result to ensure a $2$-approximation (to be specific, a $1.89$-approximation), could not be handled by previous results. Consider $H$ having Zipf’s distribution on $[20]$ with shape $a = 0.3$, that is with probability mass function $f(h)\defeq \sfrac{h^{-a}}{Z}$ for all $h\in[20]$, where $Z\defeq\sum_{h=1}^{20} h^{-a} = 10.9$. The list of hazard rates $\lambda(h)$ is approximately $0.0914, 0.0818, 0.0789, 0.0785, 0.0797, 0.0820, 0.0853, 0.0896, 0.0950, \ldots$, so it is not monotone. Yet the distribution satisfies \Cref{concentration}, since $\Var H\approx 34.6$ is smaller than the right-hand side of \Cref{concentration}, which evaluates to approximately $37.0$ (all values provided are numerical approximations up to machine precision).}
	
	\subsection{Our techniques}\label{techniques}
	The toolbox we assembled, can be described in three key-points.
	
	\textbf{Discounted infinite problem.} The optimal algorithm is characterised via a correspondence between stopping rules for RH and stopping rules for a discounted optimal stopping problem, with the discount factors being the survival function of the horizon. 
	More general optimal stopping problems (unbounded random horizon) do not always admit an optimal algorithm~\cite{Sam96}. We show that for RH, this is always possible. This is crucial for geometric and $\DHR$ horizons, which are unbounded. The generalisation is the product of an original measure-theoretic construction, leveraging trace $\sigma$-algebras.\footnote{The reader not familiar with measure-theoretic probability should not be discouraged, since the rest of our main proofs leverage only the intuitive meaning of this, which will be made clear.}
	
	\textbf{Stochastic orders.} The extension of the $2$-approximation from the $\IHR$ class to the $\PGF$ class and $\mathcal{L}^2$ horizons is obtained by exploiting stochastic orders, which have never before been applied to prophet inequality problems. 
	Several distributional classes, starting from the $\IHR$ class, can be equivalently defined through some stochastic-ordering with respect to the geometric distribution (as previously seen, the largest of such classes exploits the pgf order). Thus our \cref{Ghorizon} not only extends to the $\PGF$ class, but simplifies the arguments of \cite{AliBanGolMunWan20}. The proof of \Cref{L2horizon} leverages extremal properties of specific Bernoulli distributions with respect to another stochastic order, the Laplace transform order (see \Cref{L2} for details), for distributions having finite variance.
	
	\textbf{Hard instance.} Our analysis of the hard instance for single-thresholds algorithms is also novel. The literature on prophet inequalities has focused only on bounded discrete instances that are hard for every algorithm. On the other hand, the instance we provide exploits a continuous Pareto distribution, which, combined with a suitable family of horizons, is both hard for single-threshold algorithms, and competitive for the SP stopping rule, against the offline maximum. This requires a more sophisticated and analytic approach. The phenomenon, for which certain algorithms fail on an instance but others do not, is subtle. To the best of our knowledge, a competitive analysis of the SP stopping rule against the offline maximum with a random horizon is also new (see \Cref{related}). It is worth noting that this hard instance is designed to live in a close neighbourhood of the $\PGFd$-class, so that perturbing it suitably yields, at the same time, a hard family in the $\PGFd$ class.

	\subsection{Related work}\label{related}
	The literature on prophet inequalities models with deterministic horizon is vast. The techniques required for RH differ significantly, and the reader interested in the deterministic setting is referred to surveys such as \cite{Correa18}, which contains many recent developments; \cite{Lucier17,Correa19}, for an economic point of view (including applications to Posted Price Mechanisms and more broadly online auctions); \cite{HillKertz92}, for classical results concerning infinite instances. Besides the classical and IID PI, already reviewed in \Cref{model}, the other two main models with deterministic horizon we mention here are the Random Order and the Order Selection PI. In the first the random values arrive in a uniform random order (this model has been shown to be $0.688$-competitive~\cite{ChenHuangLiTang24} and $0.7235$-hard~\cite{Giam23}), in the second the order can be chosen by the gambler (this model has been shown to be $0.7258$-competitive~\cite{BubChip23} and is $\sfrac{1}{\beta}$-hard, since the IID PI is a particular case). Very recently also Oracle-augmented PI have been introduced \cite{HarHarbLiv24}. The existing literature strictly related to RH is limited to \cite{AliBanGolMunWan20, HajKleSan07,MahSab06}, which explore the topic from the point of view of online automated mechanism design (see \cite{Sand03} for a survey), with a strong emphasis on the multiple-items setting. Although the setting treated in this work is the single-item one, it is crucial, since a standard machinery has already been developed in \cite{AliBanGolMunWan20} to extend results for single-item RH to multiple items, as long as the horizons for each item are uniformly bounded. To the best of our knowledge, \cite{MahSab06} is the first paper to consider the problem of unknown market size (that is, random horizon with unknown distribution). As an example of practical applications of RH, the original motivation proposed by the authors was that in search keyword auctions such as Google, the supply (search queries that the users enter, resulting in a search results page on which ads can be displayed) is not known in advance, and arrives in an online fashion. This is a multi-unit auction of unknown size for a perishable item (the search queries, who perish within a fraction of a second), to be allocated to buyers (the ads that will be displayed on the search results page). Their analysis provides constant-approximations of both the online and offline optimal single-price algorithm. These results are limited to unknown horizons. This assumption has later been shown to be hard for RH in \cite{HajKleSan07}. In practice, distributional knowledge is available via the history of past transactions. Therefore, in \cite{HajKleSan07} the focus shifted to admitting known horizons. As to obtaining constant-approximations, however, they focus solely on generalising to multiple items the \textit{additive} prophet inequality of \cite{HillKertz83} for \textit{bounded} value distributions, 
	whereas our main focus is on multiplicative prophet inequalities, in line with \cite{AliBanGolMunWan20}, already extensively reviewed (see \Cref{bounds}). We conclude this paragraph by mentioning that, on a technical level, RH is also well-recognised as connected to other online stochastic matching problems with random number of queries, which have recently been shown to admit $2$-approximations \cite{AouMa23}. A potential ground for applications of our techniques would be improving such guarantees.
	
	In the remaining of this section we will focus on another well-known problem in theoretical computer science, the Secretary Problem (SP), extensively studied in the random horizon setting, partially inspiring our approach. For the history and variations on the problem see \cite{GilMost66,Free83}. The classical SP is as follows: let $n$ be a known number of secretaries for hire. They are interviewed in an online fashion by the employer, all $n!$ possible orderings being equally likely. The employer is able, at the end of each interview, to rank all the secretaries that have been presented so far ( \textit{relative ranks}). Each time, the employer must either hire (and stop with the interview), or irrevocably reject the secretary (and proceed with the next interview, except for the last one: the employer will have to hire the last secretary, if the previous $n-1$ were rejected). The objective is maximising the probability of choosing the best secretary, that is the one having \textit{absolute rank} $1$. 
	The problem is solved in \cite{Lind61}: the optimal stopping rule consists in rejecting all the first $r-1$ of the $n$ secretaries, and to accept the first secretary, which is better than all previous ones (relative rank $1$) thereafter. The number $r=r_n$ can be characterised such that as $n\longrightarrow\infty$, both $\sfrac{r}{n}$ and the probability of selecting the best (secretary) tend to $\sfrac{1}{e}\approx 0.3679$. This result is also obtained in \cite{Dynk63} through an alternative study of an underlying Markov chain. Consider the following process $\{x_t\}_{t\in[n]_0}$: if the employer enumerates the secretaries in the order in which they are encountered, and then lets $x_0=1$, and for $t\in[n]$, lets $x_t$ be the number of the first secretary, which is preferable to the one denoted by $x_{t-1}$; then we have that $\{x_t\}$ is a Markov chain with state space $[n]$. Potential-based optimal stopping techniques yield the aforementioned results. The Markovian approach is crucial to the analysis of random horizon SP.
	
	When the number of secretaries is random, it is denoted by $N$, and the secretaries are interviewed in uniform random order, conditionally on $N$. At each interview the employer, who has distributional knowledge of $N$, makes the decision of stopping in ignorance of whether there will be a secretary to interview next. The employer would obtain no reward, in the eventuality that the secretary just rejected was indeed the last. This model was first studied in \cite{PresSon72} through a generalisation of the Markovian approach of \cite{Dynk63}. The optimal stopping rule is no longer as simple as in the deterministic case, where we had that it is optimal to stop after a waiting time, as soon as a secretary having relative rank $1$ is interviewed. This is called one \textit{stopping island} scenario. Depending on the horizon's distribution, we can have a more complicated set $\Gamma$ composed of multiple islands, in which it is optimal to stop at a secretary with relative rank $1$, and outside which it is not. For uniform, geometric and Poisson horizons there still is only one island yielding non-vanishing asymptotics of the probability of choosing the best, but the value varies. 
	In \cite{AbdBatTrus82} it was later shown that a hard family of horizons $\{N_m\}$ (with $\supp(N_m)=[m]_0$, for every $m\in\NN$) could be constructed, meaning that as $m\longrightarrow\infty$, the probability of choosing the best vanishes. In order to recover a constant nonzero limiting value, more information is needed.
	
	SP has also been studied in its \textit{full-information} variant, that is when the common continuous distribution, from which the quality measurement of each of the $n$ secretaries is sampled independently, is known by the employer, who can proceed to directly inspect, at each interview, the secretary's true quality. A first heuristic solution to this problem was given in \cite{GilMost66} via a backward induction argument. 
	A rigorous proof of this form of the optimal stopping rule was successively derived in \cite{Boj78} via the Markovian approach. The full-information variant has also been studied with random horizon by \cite{Por87}, where the Markovian approach of \cite{Boj78} is successfully extended. To the best of our knowledge, a hard instance is not yet known for the full-information problem with random horizon. 
	
	Both no- and full-information SP have also been studied under the harder hypothesis of \textit{random freeze}, respectively in \cite{Sam95,Sam96}. In this case, given $n$ secretaries, interviews may have to stop  at an independent random time $N\le n$, with distribution known to the employer. 
 	Finally, we mention a variant of SP studied in \cite{ChoMorRobSam64}: instead of seeking to maximise the probability of choosing the best, the employer seeks to minimise the expected rank of the chosen secretary. 
 	This problem was extended to bounded random horizons $N$ in \cite{Gia79}, and studied for a parametric family of $\IHR$ horizons. 
	To the best of our knowledge this is the first work, which investigates broad distributional classes (note that they are defined in terms of hazard rate properties) rather than families of distributions of a given form. More recent developments extend SP to combinatorial structures and unknown horizons~\cite{HoefKod17,GharVon13}.
	
	\subsection{Organisation of the paper and notation}
	
	The paper is organised as follows. In \Cref{opt} we give preliminaries on the formal approach to random horizon optimal stopping problems, and prove the equivalence with discounted optimal stopping problem. In \Cref{Gclass} we prove \Cref{Ghorizon}.  In \Cref{hardsecretary} we prove \Cref{future}. In \Cref{Gdclass} we prove \Cref{Gdhorizon}. In \Cref{L2} we prove \Cref{L2horizon}. The Supplementary materials contain a detailed version of all the proofs.
	
	\section{Preliminaries}\label{opt}
	In this section we set up our notation, give finer details regarding the probabilistic model of RH and characterise the optimal algorithm by deriving a useful equivalence between stopping rules for RH and stopping rules of an associated discounted (possibly infinite) stopping problem.
	
	\subsection{Notation} 
	We denote $\NN_0\defeq\NN\cup\{0\}$, $\NNN\defeq\NN\cup\{\infty\}$, $\overline{\RR}\defeq\RR\cup\{\infty,-\infty\}$, $[n]\defeq\{1,2,\ldots,n\}\subset\NN$, $[0]=\{0\}$, $[n]_0\defeq\{0\}\cup[n]$, $a\wedge b\defeq\min(a,b)$ and $a\vee b\defeq \max(a,b)$ for all $a,b\in\RR$. Given nonnegative functions $f(x)$ and $g(x)$ (or equivalently the corresponding sequences $\{f_n\}$ and $\{g_n\}$ if we consider $x=n\in\NN$) and a point $a\in\overline{\RR}$, assume, for simplicity, that $g(x)$ is strictly positive in a punctured neighbourhood of $a$ (if $a=\infty$, this is understood to mean for all $x$ large enough). Then we denote:
	\begin{itemize}[noitemsep]
		\item $f(x)=\bigO(g(x))$ as $x\longrightarrow a$ if $\limsup_{x\longrightarrow a}\frac{f(x)}{g(x)}<\infty$.
		\item $f(x)=\smallO(g(x))$ as $x\longrightarrow a$ if $\lim_{x\longrightarrow a}\frac{f(x)}{g(x)}=0$.
		\item $f(x)=\Omega(g(x))$ as $x\longrightarrow a$ if $\liminf_{x\longrightarrow a}\frac{f(x)}{g(x)}>0$.
		\item $f(x)\sim g(x)$ as $x\longrightarrow a$ if $\lim_{x\longrightarrow a}\frac{f(x)}{g(x)}=1$.
		\item $f(x)\asymp g(x)$ as $x\longrightarrow a$ if $f(x)=\bigO(g(x))$ and $f(x)=\Omega(g(x))$.
	\end{itemize}
	Some of the asymptotic notation above, as customary, can also be extended to a function $f(x)$ possibly taking negative values. When this occurs, it is understood that we consider $|f(x)|$ instead. We will adopt the convention of omitting the dependence on all parameters not involved in the main limiting process (which will be clear from the context), regardless of whether uniformity (when not relevant to the context) holds within the parametric range.\footnote{For example, the sentence: as $x\longrightarrow\infty$, for every $0<\eps<1$, $f(x,\eps)\asymp g(x,\eps)$; means that $\asymp$ refers to $x$, for every $\eps$ fixed; we avoid $\asymp_{\eps}$, and there is no intended reference to whether, as $\eps$ varies in $(0,1)$, the asymptotic constants hold uniformly or not.} In evaluating integrals we denote $f(x)\vert_a^b\defeq f(b)-f(a)$. \textit{Absolutely continuous} distributions will be referred to as \textit{continuous} for brevity. \textit{Discontinuous} distributions are assumed to satisfy the usual regularity assumption (isolated jumps). Convergence in distribution of random variables is denoted as weak convergence: $\overset{w}{\longrightarrow}$. \textit{Almost surely} is abbreviated $\as$ When random variables $H$, $G$ are identically distributed, it is denoted by $H\sim G$. The essential supremum of a collection of random variables is denoted with the symbol $\esup$. $S(h)\defeq \PP(H\ge h)$ is taken as the \textit{survival function} of $H$ (when necessary the notation $S_H(h)$ will be used), while we denote $\overline{F}(h)\defeq S(h+1)=\PP(H> h)$. We follow the standard probabilistic functional notation for powers: for example $\PP^2(\cdot)\defeq\left[\PP(\cdot)\right]^2$ and $\EE^2(\cdot)\defeq\left[\EE(\cdot)\right]^2$.

    \subsection{Probabilistic model}
	As aforementioned in this section (and in this section alone) RH is studied from the point of view of optimal stopping, hence horizons admitting mass at zero will be also considered.\footnote{This has been shown to be preferable\cite{Sam96}.} In \Cref{model} we ruled out mass at zero, since in all later sections, concerned with competitive analysis, this will be the standard assumption. A formal convention in fact restores the equivalence with a scenario admitting mass at zero, from a competitive analysis point of view. Add a value $X_0\defeq 0$, and set $X_\tau\defeq X_0\eqdef M_{[0]}$ for all $\omega\in\{H=0\}$. Informally, when the game does not start, both the gambler and the prophet obtain no return. Adopt the convention $\sfrac{0}{0}=1$. Having set $X_0=0$ ensures that the game always starts, unless $\omega\in\{H=0\}$. Assuming absence of penalty for entering the game even though it does not start yields the equivalence: the ratio being $1$ does not affect the worst case competitive ratio. 
	
	In RH the gambler must select the value in ignorance of whether it is the last or not: at each step $i\in\NN_0$ the gambler only learns whether the inspected value $x_i$ is the last or not (that is, whether $\omega\in\{H=i\}$ or $\omega\in\{H>i\}$) \textit{after} the decision of accepting or rejecting it has been made (that is in the $i+1$st step). Probabilistically, it is standard to model this as follows. The random list $X_1,\ldots,X_H$ is constructed from the infinite underlying process $\XX\defeq\{X_i\}$ of \textit{iid} copies of $X$, whose natural filtration (intuitively, the information associated with the history of the process up to the present) is the sequence of $\sigma(X_1,\ldots,X_i)\defeq\DF_i$, for every $i\in\NN_0$, where $\DF_0\defeq\{\emptyset,\Omega\}$ (no information at the start: this is the implicit starting point of all filtration we will introduce). Let $\nu\defeq\EE X<\infty$. The nonanticipative condition and the ignorance aforementioned are modelled by requiring that $\{\tau=i\}\in\sigma(\II_{\{H=0\}},X_1,\ldots,\II_{\{H=i-1\}},X_i)\eqdef \FF_i$ for every $i\in\NN_0$. Intuitively, this second filtration represents the information associated with the history of the game, which combines both the history of the process of the values, and the state of the game in all preceding steps (that is, whether it ended or it is still running). We recover the original $X_1,\ldots,X_H$ by imposing that the reward is zero if the gambler fails to stop by time $H$ on the infinite sequence, that is defining the reward sequence $\{Y_i\}$, where $Y_i\defeq X_i\II_{\{H\ge i\}}$, to which we add $Y_0\defeq X_0$.
	
	For technical reasons, we consider the slightly less common class of all possibly infinite stopping times for the problem $\XX$($\YY$), denoted as $\TT$($\TT^*$). Formally, $\tau\in\TT(\TT^*)$, means that $\{\tau=i\}\in\DF_i(\FF_i)$ for all $i\in\NN$, admitting also $\PP(\tau=\infty)> 0$. This requires that we \textit{compactify} the problem, by prescribing a conventional, yet natural, reward value, in the event that the gambler never stops. We set $Y_\infty\defeq \limsup_{i\longrightarrow\infty} Y_i=0$, since $\PP(H<\infty)=1$, so if the gambler never stops, $\as$ the horizon will have been reached, so there can be no reward. This way infinite stopping rules do not increase the largest possible expected reward and the equivalence with the original formulation is maintained. In conclusion, RH has been reformulated as an infinite optimal stopping problem with reward sequence $\YY\defeq\{Y_i\}_{i\in\NNN_0}$, underlying process $\XX$ and filtration $\FF\defeq \{\FF_i\}_{i\in\NNN_0}$, where $\FF_\infty\defeq \sigma\left(\bigcup_{i\in\NN_0} \FF_i\right)$ (the full information about the history of the process, which is the ending point of all infinite filtration we will introduce), with respect to the class $\TT^*$. We seek to find the value of the problem $\Val(\YY)\defeq\sup_{\tau\in\TT^*}\EE Y_\tau$ and, if it exists, an optimal stopping rule $\bar{\tau}\in\TT^*$. \textit{Optimal} means that $\EE Y_{\bar{\tau}} = \Val(\YY)$. The reference to the horizon $H$ can be made explicit, whenever the necessity arises, via the notation $\YY^\ssup{H}\defeq\{Y_i^\ssup{H}\}_{i\in\NNN_0}$.
 
	\subsection{Optimal algorithm}
	Under suitable hypotheses, an optimal stopping problem with independent horizon $H$, having survival function $S(i)$, is equivalent to a finite (if $H$ is bounded) or infinite (if $H$ is unbounded) optimal stopping problem \textit{discounted} by factors $S(i)$. The discounted problem is $\DZ\defeq\{Z_i\}_{i\in\NNN_0}$, where $Z_0\defeq Y_0$, $Z_i\defeq S(i)X_i$ for every $i\in\NN$, $Z_\infty\defeq\limsup_{i\longrightarrow\infty} Z_i=0=Y_\infty$. That is, we change only the reward function, not the underlying process. Thus, for $\DZ$ we use the filtration $\DF\defeq\{\DF_i\}$, 
	having set $\DF_0$ and $\DF_\infty$ as usual. The latter is used when the problem is infinite. Thus $\Val(\DZ)\defeq\sup_{\zeta\in\TT}\EE Z_\zeta$, and the equivalence aforementioned means that $\Val(\YY)=\Val(\DZ)$. 
	
	In \cite{Sam96} the case where the horizon $H$ is bounded (finite support $\supp(H)$) is solved. Let $m\defeq\sup\: \supp(H)<\infty$, and take the discounted problem $Z_0, Z_1, \ldots, Z_m, 0,\ldots$ denoted $\DZ^\ssup{m}$. Then \textit{backward induction}, starting at $m$, yields the optimal stopping rule $\bar{\zeta}\in\TT^m$ attaining $\Val(\DZ)$, where $\TT^m$ is the class of all all stopping rules in $\TT$ that stop no later than time $m$. The stopping rule $\bar{\tau}\defeq \bar{\zeta}\wedge (H+1)$ then is shown to yield $\Val(\YY)$. The setting considered in \cite{Sam96} does not allow, in general, for the existence of an optimal stopping rule for unbounded horizons, which require an infinite $\DZ$. Only the equivalence with the discounted problem holds in general. The unbounded case is crucial for RH, as an optimal stopping rule is needed explicitly. The main contribution of our \Cref{optalg} is showing that RH with unbounded horizon $H$ admits an optimal stopping rule of the form $\bar{\tau}\defeq \bar{\zeta}\wedge (H+1)$ as well, where we will take $\bar{\zeta}$ to be the \textit{Snell stopping rule} for $\DZ$, since for an infinite problem, backward induction is not possible. Let us recall a few facts, adapted from \cite{ChoRob63,ChoRobSieg71} to our compactified framework. Let $\TT_i$ be the class of all all stopping rules in $\TT$ that do not stop before step $i$.
	\begin{definition}[Snell envelope]
		The \emph{Snell envelope} of $\DZ$ is $\{V_i\}$, where $V_i\defeq\esup_{\zeta\in\TT_i}\EE_{\DF_i}Z_\zeta$.
	\end{definition}
	Note that $\Val(\DZ)=V_0$. The intuitive interpretation of the stochastic process $\VV=\{V_n\}$ is the best expected return available at step $i$ among rules that have reached step $i$.
	\begin{remark}\label{snell}
		Since $\EE \sup_{i\in\NN} Z_i<\infty$ for every $i\in\NN_0$ the Snell envelope satisfies 
		\begin{equation}\label{dpe}
			V_i=Z_i\vee\EE_{\DF_i}V_{i+1}.
		\end{equation}
	\end{remark}
	Informally, \Cref{dpe} represents an extension of the dynamic programming principle for the infinite process $\{V_i\}$: at time $i$ it is optimal to stop whenever the currently inspected value $Z_i$ is greater than the expected highest expected future return as per the previous informal characterisation of $V_i$. We use \Cref{dpe} to compute $V_0$. Note that for finite problems, it coincides with backward induction.
	\begin{definition}[Snell stopping rule]\label{snellstoprule}
		The \emph{Snell rule} is $\bar{\zeta}\defeq\inf\{i\in\NN_0:\,Z_i\ge \EE_{\DF_i}V_{i+1}\}$.
	\end{definition}
	\begin{remark}\label{snellstop}
		If $\EE \sup_{i\in\NN} Z_i<\infty$ and $Z_\infty\defeq \limsup_{i\longrightarrow\infty} Z_i\in\RR$, then the \emph{Snell rule} $\bar{\zeta}$ is optimal.
	\end{remark}
	We are now ready to prove the main result. The details of the proof are provided in \Cref{suppopt}.
	\begin{theorem}\label{optalg}
		Let $H$ be a finite horizon with $\mu\defeq\EE H$ and possibly mass at zero, and let $m\defeq\sup \supp(H)$, $\YY$ and $\DZ$ as previously defined. It holds that:
		\begin{itemize}[noitemsep]
			\item$\Val(\YY) = \Val(\DZ)$; 
			\item there exists $\bar{\zeta}\in\TT$ such that $\EE Z_{\bar{\zeta}} = \Val (\DZ)$;
			\item $\bar{\tau}\defeq \bar{\zeta}\wedge (H+1)$ is such that $\EE Y_{\bar{\tau}} = \Val(\YY)$.
		\end{itemize}
		Furthermore, $\bar{\zeta}$ is characterised as:
		\begin{itemize}[noitemsep]
			\item The \emph{backward induction stopping rule} for $\DZ^\ssup{m}$, if $m<\infty$;
			\item The \emph{Snell rule} for $\DZ$, if $m=\infty$.
		\end{itemize}
	\end{theorem}
        	\begin{proof}[Idea of the proof]
		The case $m<\infty$ corresponds to \cite[Theorem~2.1]{Sam96}. The most crucial step in the proof of the case $m=\infty$ is showing that for every stopping rule $\sigma\in \TT^*$, we can construct a stopping rule $\zeta\in\TT$, such that $\EE Y_\sigma=\EE Y_\tau$ (that is $\sigma$ and $\tau$ are \textit{equivalent}), where $\tau\defeq\zeta\wedge(H+1)\in\TT^*$. This requires an inductive measure-theoretic construction. 
		 For every $i\in\NN$, let $E_i\defeq\{\sigma = i\}\cap\{H\ge i\}$. We have that $E_i\in\FF_i\cap\{H\ge i\}=\DF_i\cap\{H\ge i\}$, where the intersections denote, as standard, the corresponding \textit{trace $\sigma$-algebras}. Next we construct inductively mutually disjoint events $A_i\in\DF_i$, such that $E_i=A_i\cap \{H\ge i\}$. These events $\{A_i\}$ are then relied upon to define $\zeta$ as, informally speaking, a stopping rule that stops when $\sigma$ successfully stops by the time the horizon has realised, and does not stop otherwise. Formally, denote $C\defeq\left(\bigcup_{i\in\NN}A_i\right)^c\in\DF_\infty$, and define for every $i\in\NN$,
		\begin{equation}\label{equivstop}
			\zeta(\omega)\defeq i,\:\omega\in A_i, 
		\end{equation}
		while setting $\zeta$ to $\infty$ otherwise. 
		Following the standard convention
		$\EE Y_{\sigma}\defeq \EE\sum_{i\in\NN}Y_i\II_{\{\sigma=i\}}+\EE Y_\infty$,
		we verify by direct computation 
		that $\EE Y_\sigma=\EE Y_\tau$ is ensured by $\EE Y_\infty=0$.
		
		In order to show that $\Val(\YY)=\Val(\DZ)$, 
		note that since $Y_\infty=Z_\infty=0$, we can show firstly, that for every $\sigma\in\TT^*$, $\EE Y_{\sigma}=\EE Z_{\zeta}$, with $\zeta$ as defined in \Cref{equivstop}, and secondly, that for every $\zeta\in\TT$, 
    	\begin{equation}\label{optequiv}
        	\EE Y_\zeta=\EE Y_\tau,
    	\end{equation} 
        with $\tau\defeq \zeta\wedge(H+1)$. These two facts imply, by contradiction, that $\sup_{\sigma\in\TT^*}\EE Y_\sigma = \sup_{\zeta\in\TT}\EE Z_\zeta$.
		
		Given an optimal stopping rule $\bar{\zeta}\in\TT$, we have that $\bar{\tau}\defeq\bar{\zeta}\wedge(H+1)\in\TT^*$ provides an optimal rule by an argument by contradiction based on \Cref{optequiv}. Furthermore, our assumptions ensure that $\EE \sup_{i\in\NN} Z_i<\infty$ and $\limsup_{i\longrightarrow\infty} Z_i=0$. Hence the Snell envelope $\VV$ of $\DZ$ satisfies \Cref{dpe} by \Cref{snell}, and yields the optimal stopping rule in \Cref{snellstoprule} by \Cref{snellstop}.
	\end{proof}
	
	The following result gives the intuitive interpretation of $\EE_{\DF_i}V_{i+1}$ as the best expected return available at step $i$ among rules that will not stop \textit{by or at} step $i$. We believe that it is known within the optimal stopping community, but we could not find a reference nor a proof. The proof is provided in \Cref{suppopt} for ease of the reader. As usual, assume that the process $\DZ$, adapted to a filtration $\DF$, satisfies $\EE\sup_{i\in\NN_0}Z_i<\infty$ and $\limsup_{i\longrightarrow\infty}Z_i=Z_\infty=0=Z_0$.
	\begin{lemma}\label{snellfuture}
		Let $\VV$ be the Snell envelope of $\DZ$. Then
		$\EE_{\DF_i}V_{i+1}=\esup_{\zeta\in\TT_{i+1}}\EE_{\DF_i}Z_\zeta \as$
	\end{lemma}
	
	\section{Single-threshold $2$-approximation on the $\PGF$ class}\label{Gclass}
	We build familiarity with the model by proving \Cref{Ghorizon}, which derives a $2$-approximation on the $\PGF$ class, extending the results of \cite[Section~3]{AliBanGolMunWan20}. Recall that $X\sim\DD$ is the distribution of the values and $H\sim\HH$ of the horizon (with no mass at zero). The expected return of a single-threshold algorithm is readily computed. The details are included in \Cref{suppGclassalgp} for ease of the reader.
	\begin{lemma}\label{algp}
		Let $\pi>0$ and consider the stopping rule $\tau_\pi\defeq\inf\{i\in\NN:\,Y_i\ge \pi\}\in\TT^*$, where $Y_i\defeq X_i\II_{\{H\ge i\}}$. Then $\EE X_{\tau_\pi}= c_\pi\EE(X|X\ge \pi)$, where $c_\pi=c_\pi(\HH,\DD)\defeq 1-\EE\left[\PP^H(X<\pi)\right]$.
	\end{lemma}
	
	The quantity $c_\pi$ yields the competitive ratio. With this in mind, we need to upper-bound $\EE M_H$. This is done by \textit{ex-ante} relaxation of the prophet, combined with a straightforward variational argument.\footnote{\Cref{algp} and a slightly incomplete, informal version of \Cref{max} leveraging mathematical programming already appeared in \cite[Theorem~3.2]{AliBanGolMunWan20}. The contribution is valuable and insightful. The alternative argument we provide however is not only simpler, but formalises the claim and fills some gaps, which will demand our attention later.} Further details and the complete proof of \Cref{max} are provided in \Cref{suppGclassmax}.
	\begin{lemma}\label{max}
		Let $X$ be continuous and $p>0$ such that $\PP(X\geq p)=\sfrac{1}{\mu}$. Then $\EE M_H\le\EE(X|X\ge p).$
	\end{lemma}
        \begin{proof}[Idea of the proof]
		Let $f(x)$ be the probability density function (pdf) of $M_H$ and $v(x)$ the pdf of $X$. In order to upper-bound $\EE M_H=\int_0^\infty x f(x)dx\eqdef I(f)$, we solve the variational problem of maximising the functional $I(f)$ with constraints:
		$\int_0^\infty f(x)dx=1$; $0\le f(x)\le\mu v(x)$ for all $x\in[0,\infty)$; $f$ is nonnegative Lebesgue integrable on $[0,\infty)$, 
		the constraint follows from a union bound. The maximal solution is $\bar{f}(x)=\mu v(x)$ for all $x\ge p$, and zero otherwise, with $p$ as in the claim. Then $I(\bar{f})=\EE(X|X\ge p)$.
	\end{proof}
	
	To prove \Cref{Ghorizon} we extend \Cref{max} to discontinuous $X$ via \textit{stochastic tie-breaking}. The technique is standard with deterministic horizons, but with random horizons it is not yet clear that the assumptions imposed on RH are robust enough. We will show that an underlying property of the prophet, Uniform Integrability (UI), is the key. First we derive the equivalent of \Cref{algp} for a randomised algorithm (randomisation does not increase the optimal value of optimal stopping problems). Attach to every $X_i$ a biased coin flip, that is \textit{iid} Bernoulli random variables $B_i\sim B\sim\Ber(q)$ independent of $X$ and $H$, $q$ denoting the probability of $1$ (heads).\footnote{This extended version of RH, can be treated through the equivalent process $Y_i\defeq X_iB_i\II_{\{H\ge i\}}$, whose history will generate the extended filtration $\BF$ defined via $\BF_i\defeq\sigma(\II_{\{H=0\}},X_1,B_1,\ldots,\II_{\{H=i-1\}},X_i,B_i)$. Formally, randomised stopping rules are those adapted to $\BF$.} We denote the class of adapted, and thus randomised, stopping rules as $\BT$.
	\begin{lemma}\label{algran}
		Let $\pi>0$ and consider the stopping rule $\tau_{\pi,q}\defeq\inf\{i\in\NN:\,Y_i\ge \pi\}\in\BT$.
  		Then $\EE X_{\tau_{\pi,q}}= c_{\pi,q}\EE(X|X\ge \pi)$, where $c_{\pi,q}=c_{\pi,q}(\HH,\DD)\defeq 1-\EE\left\lbrace[1-q\PP(X\ge \pi)]^H\right\rbrace$.
	\end{lemma}
	Rewriting $c_\pi=1-\EE\left\lbrace [1-\PP(X\ge\pi)]^H\right\rbrace$ in \Cref{algp} helps seeing the statement of \Cref{algran} as its natural generalisation. The proof of \Cref{algran} is combinatorially tedious, but straightforward, and is provided in \Cref{suppGclassalgran}. We are now ready to show the main result of this section, \Cref{Ghorizon}. The details of the proof are provided in \Cref{suppGclassGhorizon}.
    \begin{proof}[Idea of the proof of \Cref{Ghorizon}]
	If $X$ is continuous, by \Cref{algp,max},
	\[\frac{\EE X_{\tau_p}}{\EE M_H}\ge c_p\defeq1-\EE\left[\PP^H(X<p)\right].\]
	For any $H\in\PGF$, $G\prec_{pgf}H$ with $G\sim\Geom(\sfrac{1}{\mu})$, $\mu=\EE H$. Let $t\defeq \PP(X<p)=1-\sfrac{1}{\mu}$, then \[\EE t^H\le \EE t^G=\frac{1}{\mu}\sum_{h\in\NN}t^h\left(1-\sfrac{1}{\mu}\right)^{h-1}=\frac{1-\sfrac{1}{\mu}}{2-\sfrac{1}{\mu}}.\] This implies that 
	$c_p\ge (2-\sfrac{1}{\mu})^{-1}$ and the result follows. 
	
	Consider now $X$ discontinuous. Let $V(x)$ be the cumulative distribution function (cdf) of $X$. Assume that for some $p$, $\lim_{\eps\longrightarrow0^+}V(p-\eps)<1-\sfrac{1}{\mu}<V(p)$, that is $1-\sfrac{1}{\mu}$ is in correspondence of a jump of the cdf at $p$.\footnote{This is the most difficult case to handle, when adapting the previous part of the argument. To present the key ideas of how \Cref{max} can be ensured for a discontinuous value distribution, we focus solely on this scenario. The extension to easier scenarios is trivial.} 
  	$V(x)$ has at most countably many jump discontinuities at values $D\defeq\{j_k\}$ (increasingly enumerated), one of which is $p$. Let $\{\eps(l)\}$ be a small enough positive monotonically vanishing sequence. 
  	Linearly interpolate $V(x)$ on disjointed intervals of length $\eps(l)$ (thus eliminating all jump discontinuities) and keep it unaltered otherwise. This yields an approximating sequence of continuous and piecewise differentiable cdf's $V_l(x)$ of continuous random variables $X^\ssup{l}$. Upon showing that the family $\{X^\ssup{l}\}$ is Uniformly Integrable (UI), define $M_H^\ssup{l}\defeq \max\{X_1^\ssup{l},\ldots,X_H^\ssup{l}\}$. Then $\{M_H^\ssup{l}\}$ is UI. Furthermore, since $X^\ssup{l}\overset{w}{\longrightarrow}X$, $M_H^\ssup{l}\overset{w}{\longrightarrow}M_H$. By \Cref{max}, for every $l$ there is $p_l$ such that $\PP(X^\ssup{l}\ge p_l)=\sfrac{1}{\mu}$ and $\EE M_H^\ssup{l}\le\EE(X^\ssup{l}|X^\ssup{l}\ge p_l)$. By construction, $p_l\longrightarrow p$ as $l\longrightarrow\infty$, so $\EE(X^\ssup{l}|X^\ssup{l}\ge p_l)\longrightarrow \EE(X|X\ge p)$.\footnote{A reference for convergence in expectation under UI and weak convergence hypotheses is \cite{Bill87}.} 
  	These facts altogether ensure that $\EE M_H\le \EE(X|X\ge p)$. 
  	
  	To conclude, for any $H\in\PGF$, since $G\prec_{pgf}H$, by \Cref{algran} the single-threshold randomised algorithm $\tau_{p,\bar{q}}$, with $\bar{q}\defeq[\mu\PP(X\ge p)]^{-1}$, is such that 
	\[\frac{\EE X_{\tau_{p,\bar{q}}}}{\EE M_H}\ge c_{p,\bar{q}}\defeq1-\EE\left\lbrace\left[1-\bar{q}\PP(X\ge p)\right]^H\right\rbrace=1-\EE\left[\left(1-\frac{1}{\mu}\right)^H\right]\ge\frac{1}{2-\frac{1}{\mu}}.\]
	\end{proof}
	
	\begin{remark}\label{algorithm}
		For every horizon $H\in\PGF$ the following hold.
		\begin{itemize}[noitemsep]
			\item If $X$ is continuous, the $2-\sfrac{1}{\mu}$-approximation is a single-threshold algorithm, where the threshold is the value $p$ such that $\PP(X\ge p)=\sfrac{1}{\mu}$.
			\item If $X$ is discontinuous with $\sfrac{1}{\mu}$ in correspondence of a discontinuity jump of the survival function of $X$, the $2-\sfrac{1}{\mu}$-approximation is a single-threshold randomised algorithm, where the threshold is the value $p$ at which the discontinuity jump occurs, and the randomisation parameter is $[\mu\PP(X\ge p)]^{-1}$. 
			\item As a result if $X$ is discontinuous but $\sfrac{1}{\mu}$ does not occur in correspondence of a discontinuity jump, the randomisation parameter is $1$, meaning that the nonrandomised algorithm for the continuous case is enough to yield a $2-\sfrac{1}{\mu}$-approximation.
		\end{itemize}
	\end{remark}
	Some of the most well-known subclasses of the $\PGF$ class, which earned considerable interest in applications, are: $\IHRE$ (Increasing Hazard Rate in Expectation); $\HIHRE$ (Harmonically Increasing Hazard Rate in Expectation); $\NBU$ (New Better than Used); $\NBUE$ (New Better than Used in Expectation). The following tower of inclusions is well-established \cite{BraqRoyXie01,Rol75,Klef82} and known to be strict: $\IHR\subset\IHRE\subset\HIHRE\subset\NBU\subset\NBUE\subset\HNBUE\subset\PGF$. Further details on these subclasses are given in \Cref{suppGclassdetails}.
		
  	The tight instance for the $2-\sfrac{1}{\mu}$-approximation involves a geometric horizon, and was studied in \cite[Theorem~3.5]{AliBanGolMunWan20}. Their basic assumptions, concerning a two-point value distribution, are intuitive facts: that an optimal algorithm with geometric horizon should exist, and that it has a single-threshold. We show mathematically that they both hold, with the second extending, more in general, to any bounded value distribution. As a bonus we obtain an equation for the optimal threshold. The details of the proof of \Cref{geomprice} are provided in \Cref{suppGclassgeomprice}. Recall that $\Geom(1-q)$ denotes the geometric distribution with failure probability $0<q<1$.
	\begin{lemma}\label{geomprice}
		Let $H\sim\Geom(1-q)$, and $X$ such that $\supp(X)=[a,b]$, $0\le a\le b$. Then there exists a threshold $V_0$ such that the corresponding single-threshold algorithm is optimal.
	\end{lemma}
	\begin{proof}[Idea of the proof]
		By \Cref{optalg}, $\Val(\XX)=\Val(\YY)=\Val(\DZ)\eqdef V_0$. By \Cref{dpe}, it holds that $V_0=\EE V_1=\EE (Z_1\vee\EE_{\DF_1}V_{2}).$ Since $Z_i\defeq q^{i-1}X_i$, by time-invariance and \Cref{snellfuture} it follows that $\EE_{\DF_1}V_{2}=q\EE V_1=qV_0$. 
		Therefore, $V_0=\EE (X\vee qV_0)$; equivalently, if $X$ has cdf $F(x)$, \[f(V_0)\defeq V_0+\int_{qV_0}^bF(x)dx=b.\] The function $f(V)$ is a strictly increasing differentiable function on $[0,b]$, and $f(0)\le b\le f(b)$. Hence the value $V_0$ exists. By \Cref{snellstop,snellstoprule} and time-invariance again, it is optimal to stop at the first value greater or equal than $V_0$ (which is also the expected return).
	\end{proof}
	The tightness then follows by the straightforward computation of \cite{AliBanGolMunWan20}[Theorem~3.5]. For self-containedness we provide the short computation in \Cref{suppGclasstight}, with small variations to better fit our framework.
	
	\section{A first step beyond single-threshold algorithms}\label{hardsecretary}
	In this section we construct a parametric family of horizons, which is hard for any single-threshold algorithm, and yet it allows for an adaptive competitive stopping rule. We also show that it can be perturbed so that it enters the $\PGFd$ class, while remaining hard for single-threshold algorithms.
	
	\subsection{The hard instance}
	In this section we construct the instance, which will provide the hardness result for single-thresholds.
	Let us start with the instance $X$ for the values, which is a Pareto (Type I) distribution with scale parameter $1$ and shape parameter $1+\eps$, where $\eps>0$ is fixed small enough. The corresponding cdf, denoted as $V(x)$, vanishes on $(-\infty,1)$, whereas on $[1,\infty)$
	\begin{equation}\label{hardval}
		V(x)= 1-x^{-(1+\eps)}.
	\end{equation}
	Since $M_n$ has cdf $F_n(x)=V^n(x)$, which is differentiable for all $x\neq 1$, we obtain that its density (defined almost everywhere) $f_n(x)$ vanishes on $(-\infty,1)$, whereas $f_n(x)=nV^{n-1}(x)V'(x)$ on $(1,\infty)$. By direct computation of $\EE M_n \defeq \int_\RR xf_n(x)dx$, exploiting the properties of the Beta function and well-known asymptotics of the Gamma function, the following result holds. The proof is provided in \Cref{supphardsecretarylem}.
	\begin{lemma}\label{asympmax}
		As $n\longrightarrow\infty$, uniformly in $\eps>0$,
		$\EE M_n \sim \Gamma\left(1-\frac{1}{1+\eps}\right)n^{\frac{1}{1+\eps}}$.
	\end{lemma}
	In conjunction with the instance $X$ for the values, consider a family of horizons $\{H_m\}$, with $m\in\NN$ large enough, such that for every fixed such $m$, $\supp(H_m)=[\ell,m]\cap\NN$, with the lower limit $1<\ell\in\NN$. For every $m$ fixed, we define the pmf of each horizon by letting, for every $\ell\le h\le m$,
	\begin{equation}\label{hardhorizon}
		p_h^\ssup{m}\defeq \PP(H_m=h)=Z_m^{-1}h^{-\frac{1}{1+2\eps}},
	\end{equation}
	having defined the normalising constant $Z_m\defeq\sum_{h=\ell}^m h^{-\frac{1}{1+2\eps}}$. Through standard integral estimates of summations, it is shown that 
	\begin{equation}\label{asympnorm}
		Z_m\sim [1+(2\eps)^{-1}]m^{\frac{2\eps}{1+2\eps}},
	\end{equation} 
	and that the following holds.
	\begin{lemma}\label{asympmax_H_m}
		As $m\longrightarrow\infty$, uniformly in $\eps>0$ small enough, $\EE M_{H_m}\sim N m^{\frac{1}{1+\eps}}$, where we defined
		$N=N(\eps)\defeq\frac{\Gamma\left(1-\frac{1}{1+\eps}\right)}{\left[1+\frac{\eps}{(1+\eps)(1+2\eps)}\right][1+(2\eps)^{-1}]}$.
	\end{lemma}
	The proof of \Cref{asympnorm} can be found in the proof of \Cref{asympmax_H_m}, which is provided in \Cref{supphardsecretarylem}. By \Cref{asympnorm} and similar integral estimates as in \Cref{asympmax_H_m} it is straightforward to compute that for every $n\in\NN$, as $m\longrightarrow\infty$,
	\begin{equation}\label{expectationn}
		\EE H_m^n\sim\frac{\sum_{h=\ell}^mh^{n+1-\frac{1}{1+2\eps}}}{[1+(2\eps)^{-1}]m^{\frac{2\eps}{1+2\eps}}}\sim\frac{m^{n-\frac{1}{1+2\eps}}}{\left(n+1-\frac{1}{1+2\eps}\right)[1+(2\eps)^{-1}]m^{\frac{2\eps}{1+2\eps}}}= \frac{2\eps m^n}{n+2(n+1)\eps}.
	\end{equation}
	Note that from the elementary properties of the Gamma function 
 it holds that as $x\longrightarrow 0^+$, $\Gamma(x)\sim x^{-1}$. As a result, as $\eps$ vanishes, $N\sim2\eps\left(1-(1+\eps)^{-1}\right)^{-1}\longrightarrow 2$ and the following is immediate consequence.
	\begin{remark}\label{asympmaxunif_H_m}
		As $m\longrightarrow\infty$, uniformly in $\eps>0$ small enough, $\EE M_{H_m}\asymp m^{\frac{1}{1+\eps}}$.
	\end{remark}
	
	We are now ready to prove the hardness of the family of horizons $\{H_m\}$ for single-threshold algorithms. A stopping rule with single-threshold $\pi$ is denoted as $\tau_\pi$. A crucial point is that $\pi$ can be tailored to each horizon in the family, that is, the algorithm facing $H_m$ will set a threshold $\pi_m$ exploiting all distributional knowledge regarding $X$ and $H_m$. Yet, this is not enough to obtain a constant-approximation with the family $(X,\{H_m\})$ for suitably small $\eps$. The details of the proof, which involves a quite technical asymptotic analysis, is provided in \Cref{supphardsecretaryprop}. We also include some supporting numerics and heuristic consideration in \Cref{supphardsecretaryheur}, which ease assessing that the family of horizons in question has divergent CV as $\eps$ vanish (thus violating the requirement of \Cref{L2horizon}), and admits horizons outside the $\PGF$ and $\PGFd$ class (thus violating the requirements of \Cref{Ghorizon} and \Cref{Gdhorizon}).
	\begin{proposition}\label{hardsingle}
		For fixed $\eps>0$ small enough, given the instance of copies of $X$ and the family of horizons $\{H_m\}$ of \Cref{hardval,hardhorizon} respectively, we have that for every sequence of thresholds $\{\pi_m\}$,  $\liminf_{m\longrightarrow\infty}\frac{\EE X_{\tau_{\pi_m}}}{\EE M_{H_m}}=0$.
	\end{proposition}

    \begin{proof}[Idea of the proof.]
        Our strategy will be to characterise the sequence of optimal single thresholds $\{\bar{\pi}_m\}$ (corresponding to horizons $\{H_m\}$ for $m$ large enough and $\eps$ small enough), where $\bar{\pi}_m$ maximizes the value obtained $\EE X_{\tau_{\pi_m}} \defeq \EE Y_{\tau_{\pi_m}}^\ssup{H_m}$ over all single threshold strategies $\pi_m$.
        Denote with $r_m$ the gambler-to-prophet ratio of $\EE Y_{\tau_{\bar{\pi}_{m}}}^\ssup{H_{m}}$ and $\EE M_{H_{m}}$. As a consequence of the characterisation obtained, it will follow that there is always a subsequence $\{m_j\}$, such that as $j\longrightarrow\infty$, $r_{m_j}$ vanishes. This yields the claim, since $r_{m_j}$ is an upper bound on any gambler-to-prophet ratio for arbitrary single-thresholds. 
    
		Since $X$ is continuous, by \Cref{algp} we compute $\EE Y_{\tau_\pi}^\ssup{H_m}$, which is nontrivial only if $\pi>1$. This is expressed as the product of \[\EE(X|X\ge\pi)=\left(1+\sfrac{1}{\eps}\right)\pi\]
		and \[c_\pi=1-\EE\left[ V^{H_m}(\pi)\right]=1-\EE\left[\left(1-\sfrac{1}{\pi^{1+\eps}}\right)^{H_m}\right].\] Next, we show analytically that the optimal sequence of thresholds $\{\bar{\pi}_m\}$ satisfies the \textit{optimality equation} 
		\begin{equation*}
			g_m(\bar{\pi}_m)\defeq\EE\left[\left(1-\frac{1}{\bar{\pi}_m^{1+\eps}}\right)^{H_m}\right]+\frac{1+\eps}{\bar{\pi}_m^{1+\eps}}\EE\left[ H_m\left(1-\frac{1}{\bar{\pi}_m^{1+\eps}}\right)^{H_m-1}\right]=1.
		\end{equation*}
		
		To conclude, by these facts and \Cref{asympmax_H_m,asympmaxunif_H_m} we have that, as $m\longrightarrow\infty$, uniformly in $\eps>0$ small enough,
		\begin{equation*}
			r_m=c_{\bar{\pi}_m} \frac{\EE(X|X\ge\bar{\pi}_m)}{\EE M_{H_m}}\sim\left(1+\frac{1}{\eps}\right)\frac{c_{\bar{\pi}_m} \bar{\pi}_m}{Nm^{\frac{1}{1+\eps}}}\asymp \frac{c_{\bar{\pi}_m} \bar{\pi}_m}{\eps m^{\frac{1}{1+\eps}}}.
		\end{equation*}
		Consider that $\{\bar{\pi}_m\}$ is either bounded or unbounded.
		\begin{itemize}[noitemsep]
			\item In the bounded case, since $c_{\bar{\pi}_{m}}<1$, $r_m$ vanishes as $m\longrightarrow\infty$, for any fixed $\eps$ small enough. Note that this case covers also the trivial case for which there exists a subsequence $\{\bar{\pi}_{m_j}\}$ such that $\bar{\pi}_{m_j}\le1$, for which the numerator of $r_{m_j}$ becomes upper-bounded by $\nu\defeq\EE X$.
			\item Before moving on to the general unbounded case, we consider the subcase of divergence to infinity with $m^{\frac{1}{1+\eps}}=\smallO\left(\bar{\pi}_m\right)$. Note that by Bernoulli's inequality 
		  we have that \[c_{\bar{\pi}_{m}}\defeq1-\EE\left[\left(1-\sfrac{1}{\bar{\pi}_{m}^{1+\eps}}\right)^{H_{m}}\right]\le\frac{\EE H_{m}}{\bar{\pi}_{m}^{1+\eps}}.\] By \Cref{expectationn} with $n=1$ it holds that as $m\longrightarrow\infty$, for every fixed $\eps$ small enough $r_m$ vanishes, since
			\[c_{\bar{\pi}_{m}} \bar{\pi}_{m} m^{-\frac{1}{1+\eps}}\le \EE H_m m^{-\frac{1}{1+\eps}}\bar{\pi}_m^{-\eps}\asymp \varepsilon\left(m^{\frac{1}{1+\eps}}\bar{\pi}_m^{-1}\right)^{\eps}\longrightarrow 0.\]
			\item More in general, we can rely on the previous point to show that the unbounded case overall only allows for $r_m$ to vanish. Consider that since $0\le r_m\le 1$, if all of its convergent subsequences $\{r_{m_j}\}$ vanish for all $\eps$ small enough, then $\{r_{m}\}$ vanishes for all $\eps$ small enough. We show the sufficient condition by contradiction: assume that as $j\longrightarrow\infty$, $r_{m_j}\longrightarrow\rho=\rho(\eps)\in(0,1]$, with $\rho(\eps)$ being bounded away from $0$ no matter how small $\eps$ is taken. Then we have that as $j\longrightarrow\infty$, for $\eps>0$ fixed small enough, $\mu_j\defeq \sfrac{m_j}{\bar{\pi}_{m_j}^{1+\eps}}$
			is asymptotically equivalent to a bounded sequence, so it must be bounded too. Thus, we can consider a convergent subsequence of $\{\mu_j\}$, denoted $\{\mu_{j_k}\}$, such that as $k\longrightarrow\infty$, $\mu_{j_k}\longrightarrow\alpha=\alpha(\eps)\ge0$. A sharp asymptotic analysis of $c_{\bar{\pi}_{m_{j_k}}}$ reveals that the \textit{optimality equation} satisfied by $\bar{\pi}_{m_{j_k}}$ allows only for $\alpha$ being bounded away from zero, unless it is identically zero. In what follows, it is useful to rewrite the \textit{optimality equation} as \[g_{m_{j_k}}(\bar{\pi}_{m_{j_k}})\defeq\EE\left[\phi(\bar{\pi}_{m_{j_k}},H_{m_{j_k}})\right]=1,\] having defined \[\phi(\pi,H)\defeq \left(1-\pi^{-(1+\eps)}\right)^{H-1}\left[1-\pi^{-(1+\eps)}+(1+\eps)H\pi^{-(1+\eps)}\right].\] By estimating asymptotically the series expansion of $\phi(\bar{\pi}_{m_{j_k}},H_{m_{j_k}})$ with respect to the first argument, uniformly in $\alpha>0$, the following bound is shown to hold as $k\longrightarrow\infty$: \[g_{m_{j_k}}(\bar{\pi}_{m_{j_k}})\le1+\eps\left[(2+\alpha)e^{-\alpha}-2\right]+\bigO(\eps^2)+ \bigO \left(\frac{1}{m_{j_k}}\right).\]
			It follows that for $\eps>0$ small enough, as $k\longrightarrow\infty$, $g_{m_{j_k}}(\bar{\pi}_{m_{j_k}})<1$ if the strictly decreasing function $h(\alpha)\defeq (2+\alpha)e^{-\alpha}-2$ is bounded away from zero. Note that $h(\alpha)$ is negative and vanishes as $\alpha\ge 0$ vanishes. Since we have previously shown that if $\alpha>0$ exists, it must be bounded away from zero as $\eps$ vanishes (and thus $h(\alpha$ is bounded away from zero too)), in this case the optimality equation cannot be satisfied, as long as $\eps$ is taken sufficiently small. Thus only $\alpha =0$ would not yield a contradiction at this point, meaning that for all $\eps$ small enough, $\mu_{j_k} \longrightarrow 0$ as $k\longrightarrow\infty$. By boundedness, also $\mu_j\longrightarrow 0$ as $j\longrightarrow\infty$. This corresponds to the scenario in the previous point (replace, in that argument, $m$ with $m_j$) and therefore $r_{m_j}$ vanishes. This is in contradiction with the assumption, which ensures $r_{m_j}\longrightarrow \rho>0$, and therefore we must have $\rho=0$.
		\end{itemize}
		The sequence $\{r_m\}$ is thus shown to always admit a vanishing subsequence.
	\end{proof}
	
	\subsection{Hardness of the $\PGFd$ class}\label{Gdclass}
	In principle, there is no obvious reason why the main result of this section, \Cref{Gdhorizon}, should hold, merely based on duality. However, a hint towards the hardness of the $\PGFd$ class could be that the duality turns the estimates of \Cref{Gclass} the wrong way around. It is straightforward to see that the approach adopted on the $\PGF$ class is mirrored by the duality, meaning that the lower-bound, which provides the competitive ratio ensured, turns into an upper-bound, informally speaking. This mirroring has significance, because the ex-ante relaxation of the prophet is tight on geometric horizons, and is essentially maxed-out by the very definition of the $\PGF$ class. A similar phenomenon occurs with the argument of \cite{AliBanGolMunWan20}, if we restrict to the $\IHR$ and $\DHR$ class. The most notable subclasses of the $\PGFd$ class are the dual of those introduced at the end of the previous section: $\DHR\subset\DHRE\subset\HDHRE\subset\NWU\subset\NWUE\subset\HNWUE\subset\PGFd$.\footnote{Since these classes are obtained by reversing the inequalities in the corresponding definitions, the corresponding acronyms are obtained by replacing \textit{increasing} with \textit{decreasing} and \textit{better} with \textit{worse}.} These facts, combined with \Cref{Gdhorizon}, motivate our conjecture that no constant-approximation is possible on the $\DHR$ class either. The details of the following proof are in \Cref{supphardsecretaryGdhorizon}.
	
	\begin{proof}[Idea of the proof of \Cref{Gdhorizon}]
		We show that a perturbation of the hard instance from the previous section prevents any single-threshold constant-approximation on the $\PGFd$ class.
		\paragraph{Step 1.} We start by constructing a family $\tilde{\mathcal{H}}_M(\eps)\subset\PGFd$, which is a sequence of horizons $\tilde{H}_m$ arising as a perturbation of horizons $H_m$ in the hard family $\mathcal{H}_M(\eps)$ of \Cref{hardhorizon}, with fixed $\ell=2$, as $m$ grows, for all fixed $0<\eps<\sfrac{1}{4}$ small enough. This is done by adding to the pmf of said $H_m$, a (to be suitably rescaled) mass at one, 
		\begin{equation*}
			\delta_m\defeq Cm^{\frac{2\eps}{1+2\eps}},\quad C=C(\eps)\defeq\frac{3+10\eps}{2\eps}.
		\end{equation*} 
		We denote the perturbed horizon $\tilde{H}_m$ for every $m,\eps$ considered, and observe that the corresponding normalizing constants for the pmf's $\tilde{Z}_m$ and $Z_m$; and expectations $\tilde{\mu}_m\defeq\EE\tilde{H}_m$ and $\mu_m\defeq\EE H_m$; satisfy, as $m$ grows, the following:
		\begin{align*}
			\tilde{Z}_m&\sim \tilde{C}Z_m,\quad \tilde{C}\defeq C+1+\frac{1}{2\eps}\\ 
			\tilde{\mu}_m&\sim\frac{1}{\tilde{C}}\left(C+\frac{1+2\eps}{1+4\eps}m\right).
		\end{align*} 
		Trivially, these follow from \Cref{asympnorm,expectationn} with $n=1$ and the definitions, respectively $\tilde{Z}_m\defeq Z_m+\delta_m$ and $\tilde{\mu}_m = \sfrac{\delta_m}{\tilde{Z}_m}+\mu_m\sfrac{Z_m}{\tilde{Z}_m}$.
		To show that for all $m$ large enough, $\tilde{H}_m\in\PGFd$, we have to establish that eventually (short for all $m$ large enough from now on), for all $t\in(0,1)$ (the implicit range from now on),
		\[\frac{1}{\tilde{Z}_m}\left[\delta_mt+\sum_{h=2}^m t^h h^{-\frac{1}{1+2\eps}}\right]\ge\frac{1}{\tilde{\mu}_m}\sum_{h=1}^\infty\left(1-\sfrac{1}{\tilde{\mu}_m}\right)^{h-1}t^h,\]
		which we recast as
		\begin{equation}\label{targetmain}
			\sum_{h=2}^m t^h h^{-\frac{1}{1+2\eps}}\ge -\delta_mt+\frac{\eps_mt}{1-t(1-\sfrac{1}{\tilde{\mu}_m})},
		\end{equation}
		having set $\eps_m\defeq \sfrac{\tilde{Z}_m}{\tilde{\mu}_m}$. Let 
		\begin{align*}
			\eta_m\defeq \frac{1-\sfrac{\eps_m}{\delta_m}}{1-\sfrac{1}{\tilde{\mu}_m}}.
		\end{align*}
		We establish that eventually for all $t\in(0,\eta_m]$
		\begin{equation}\label{linear1main}
			-\delta_mt+\frac{\eps_mt}{1-t(1-\sfrac{1}{\tilde{\mu}_m})}\le0,
		\end{equation}
		and that eventually for all $t\in[\eta_m,1)$
		\begin{equation}\label{concave1main}
			\sum_{h=2}^m h(h-1)t^{h-2} h^{-\frac{1}{1+2\eps}}<\frac{2\eps_m(1-\sfrac{1}{\tilde{\mu}_m})}{[1-t(1-\sfrac{1}{\tilde{\mu}_m})]^3}.
		\end{equation}
		This implies \Cref{targetmain} for all $t\in (0,1)$, since eventually on $(0,\eta_m]$ \Cref{targetmain} holds with strict inequality directly by \Cref{linear1main}, while for all $t\in[\eta_m,1)$ we have that the difference of the left-hand and right-hand side of \Cref{targetmain}, denoted $f_m(t)$, is strictly concave (\Cref{concave1main} states precisely that $f''_m(t)<0$). Since $f_m(t)>0$ for all $t\in(0,\eta_m]$ and since it vanishes at both ends of the unit interval, it cannot vanish on $(\eta_m,1)$, meaning that there cannot be any crossings of the two sides of \Cref{targetmain} on $(\eta_m,1)$, which therefore holds everywhere. 
		
		\paragraph{Step 2.} Next, we show that with $X$ being set to be the hard instance of values $X$ of \Cref{hardval}, with $\eps>0$ fixed small enough as in \emph{Step 1}, $\EE M_{\tilde{H}_m}=\Omega (\EE M_{H_m})$ and $\EE X_{\tau_{\pi_m}}\defeq \EE Y_{\tau_{\pi_m}}^\ssup{\tilde{H}_m}\le \EE Y_{\tau_{\pi_m}}^\ssup{H_m}$ as $m$ grows, for every sequence of thresholds $\{\pi_m\}$. The first fact follows from \Cref{asympmaxunif_H_m} and observing that the law of total expectation yields
		\[\EE M_{\tilde{H}_m}=\frac{\delta_m}{\tilde{Z}_m} \EE X +\frac{Z_m}{\tilde{Z}_m}\EE M_{H_m}>\frac{Z_m}{\tilde{Z}_m}\EE M_{H_m}\sim\frac{1}{\tilde{C}}\EE M_{H_m}.\] The second fact follows from considering the survival functions of $\tilde{H}_m$ and $H_m$, respectively $\tilde{S}_m(i)$ and $S_m(i)$. Note that for all $i\ge 3$, $\tilde{S}_m(i)=S_m(i)\sfrac{Z_m}{\tilde{Z}_m}$, whereas $\tilde{S}_m(2)=S_m(2)-\sfrac{\delta_m}{\tilde{Z}_m}$. Consider next a single-threshold algorithm $\tau_\pi$, where $\pi=\pi_m$ on horizon $\tilde{H}_m$. By \Cref{optalg} (note that the notations $\tau$ and $\sigma$ used there are now swapped, as our focus here is on the original stopping rule) for any horizon $\tilde{H}_m$, there is a stopping rule equivalent to $\tau_\pi$ for the problem $\YY^\ssup{\tilde{H}_m}$, denoted as $\sigma_\pi\defeq\zeta_\pi\wedge (\tilde{H}_m+1)$, where $\zeta_\pi$ is, as per \Cref{equivstop}, an algorithm for the discounted problem $\DZ=\{\tilde{S}_m(i)X_i\}$, which stops only when $\tau_\pi$ successfully stops by the time the horizon has realised. Intuitively, we can construct $\zeta_\pi$ as the stopping rule with single threshold $\pi$ on the events $\{\tau_\pi=i\}\cap\{\tilde{H}_m\ge i\}$. Otherwise it stops at $m$ on $\{\tau_\pi>\tilde{H}_m\}$. The proof of \Cref{optalg} shows that this can be done, so that $\zeta_\pi$ is adapted to $\XX$, and as a result, as $m$ grows, 
		\[\EE Y_{\tau_\pi}^\ssup{\tilde{H}_m}=\EE\sum_{i=1}^m \tilde{S}_m(i)X_i\II_{\{\zeta_\pi=i\}}\le\EE\sum_{i=1}^m S_m(i)X_i\II_{\{\zeta_\pi=i\}}=\EE Y_{\tau_\pi}^\ssup{H_m}.\]
		
		\paragraph{Step 3.} Finally, suppose by contradiction that a constant-approximation exists on the $\PGFd$ class. Then there exist a sequence of thresholds $\{\pi_m\}$ and $c>0$ such that $\tau_{\pi_m}$ is a $\sfrac{1}{c}$-approximation on $\tilde{\mathcal{H}}_M(\eps)\subset\PGFd$ by \emph{Step 1}. By \Cref{hardsingle} for all $\eps>0$ fixed small enough, there exists a subsequence $\{\pi_{m_j}\}$ such that 
		$r_{m_j}$, the gambler-to-prophet ratio of $\EE Y_{\tau_{\pi_{m_j}}}^\ssup{H_{m_j}}$ over $\EE M_{H_{m_j}}$, vanishes. Combining this with \emph{Step 2}'s asymptotics yields the following contradiction as $j$ grows:
		\begin{equation*}
			c\le\frac{ \EE Y_{\tau_{\pi_{m_j}}}^\ssup{\tilde{H}_{m_j}}}{\EE M_{\tilde{H}_{m_j}}}=\bigO( r_{m_j})\longrightarrow0.
		\end{equation*}
	\end{proof}
	
	\subsection{The competitive Secretary Problem rule}
	In this section we show that a straightforward adaptation of the SP stopping rule with deterministic horizon $m$ is competitive on the horizons $H_m$ defined in \Cref{hardhorizon}. Consider any instance of the values $X$ which is nonnegative, integrable and continuous random variable.\footnote{The hypothesis of absolute continuity is to ensure that there are no ties in the ranking process involved in the SP rule. It is possible to also allow discontinuous random variables, by exploiting the following tie-breaking procedure: consider $(U_1,\ldots,U_m)$, where $U_i$ are \textit{iid} copies of $U\sim\Unif(0,1)$. Then define $\XX=(X_1,U_1,\ldots,X_m,U_m)$, and instead of the usual order relationship that gives rise to the usual ranking for SP, adopt the following ranking order: if $X_i<X_j$, $(X_i,U_i)\preccurlyeq(X_j,U_j)$, whereas if $X_i=X_j$ and $U_i<U_j$, then by breaking ties thanks to the uniform random variables we say that, again, $(X_i,U_i)\preccurlyeq(X_j,U_j)$. Since the uniform random variable are continuous, there are almost surely no ties by ranking via $\preccurlyeq$ instead of $<$. For simplicity, we present the argument only for the continuous case and by using the standard $\le$ as a ranking order.} Consider the process $\XX\defeq(X_1,\ldots,X_m)$, where $X_i$ are \textit{iid} copies of $X$. We equivalently characterise it in the context of SP via a uniform random permutation $\pi:\Omega\longrightarrow S_m$, where $S_m$ is the set of all permutations of $[m]$ and $\Omega$ the supporting probability space. In online notation, $\pi=(\pi_1,\ldots,\pi_m)$. Consider $\XX^\pi\defeq(X_{\pi_1},\ldots,X_{\pi_m})$. Since $\XX$ is exchangeable, $\XX^\pi\sim\XX$. 
	Consider now the waiting time $r_m-1$ for the SP stopping rule, denoted as $\zeta_{r_m}$, and recall that
	\begin{equation}\label{waitingtime}
		r_m\sim \frac{m}{e}.
	\end{equation}
	Consider $\XX^\pi$ conditionally on (the $\sigma$-algebra generated by) $\XX$, denoted as $(\XX^\pi|\XX)$. This is the process of realised values arriving in uniform random order, and constitutes therefore the usual setup for SP. In this framework, the classic result of \cite{Lind61} can be restated as 
	\begin{equation}\label{1/e}
		\PP\left(X_{\pi_{\zeta_{r_m}}}=M_m|\XX\right)\ge \frac{1}{e}.
	\end{equation} 
	Given the rule $\zeta_{r_m}$, denote the event of winning (that is, choosing the best) as $W\defeq\{X_{\pi_{\zeta_{r_m}}}=M_m\}$ and the event of winning at time $i$ as $W_i\defeq \{i=\inf\{j\ge r_m:\:X_{\pi_j}=M_m\}\}$. Note that by partitioning $W$ conditionally on $\XX$, 
	we have that 
	\begin{equation}\label{W}
		\PP(W|\XX)=\sum_{i=r_m}^m \PP(W_i|\XX).
	\end{equation}
 	The proof of \Cref{1/e} relies on the following fact, which we will also exploit:
	\begin{equation}\label{W_i}
		\PP\left(i=\inf\{j\ge r_m:\:X_{\pi_j}=M_m\}|\XX\right)=\frac{r_m-1}{m}\frac{1}{i-1}.
	\end{equation}   
When considering RH, instead of working with $\XX$ as usual, we start from $\XX^\pi$, by defining accordingly $Y_i^\ssup{H_m}\defeq X_{\pi_i}\II_{\{H_m\ge i\}}$. This is equivalent, since for every stopping rule $\tau\in\TT^*$, $\EE X_\tau=\EE X_\tau^\pi=\EE Y_\tau^{\ssup{H_m}}$. Yet it will be formally advantageous in the following proof of the main result of the section. The details of the argument are provided in \Cref{supphardsecretaryfuture}.
	
	\begin{proof}[Idea of the proof of \Cref{future}]
		Consider the stopping rule $\tau_m\defeq\zeta_{r_m}\wedge(H_m+1)\in\TT^*$. By \Cref{optequiv} in \Cref{optalg} firstly, and the law of total expectation with respect to $\XX$ secondly, we have that
		\[\EE Y_{\tau_m}^\ssup{H_m}=\EE \left[S_m(\zeta_{r_m})X_{\pi_{\zeta_{r_m}}}\right]=\EE\left\lbrace\EE\left[S_m(\zeta_{r_m})X_{\pi_{\zeta_{r_m}}}|\XX\right]\right\rbrace.\] Since at any given step (greater than $r_m-1$) stopping at the overall maximum is a subevent of stopping at the relative maximum, \[\EE\left[S_m(i)X_{\pi_i}\II_{\{\zeta_{r_m}=i\}}|\XX\right]\ge \EE\left[S_m(i)X_{\pi_i}\II_{W_i}|\XX\right]=S_m(i)M_m\PP(W_i|\XX).\] Therefore, \[\EE\left[S_m(\zeta_{r_m})X_{\pi_{\zeta_{r_m}}}|\XX\right]\ge M_m\sum_{i=r_m}^m S_m(i)\PP(W_i|\XX).\] By \Cref{asympnorm} and comparison of the summation with the corresponding integral as in \Cref{asympmax_H_m}, we estimate the survival function of $H_m$ as $m\longrightarrow\infty$: recalling \Cref{waitingtime}, for all $i>r_m\sim\sfrac{m}{e}$ we have that \[S_m(i)>1- \left(\frac{i-1}{m+1}\right)^{\frac{2\eps}{1+2\eps}}.\] Define the sequence of functions \[f_m(\eps)\defeq(m+1)^{-\frac{2\eps}{1+2\eps}}\frac{r_m}{m}\int_{r_m-2}^{m-1} x^{-\frac{1}{1+2\eps}}dx\] and the function \[g(\eps)\defeq \frac{1}{e}\left\lbrace1-\left[1+\frac{1}{2\eps}\right]\left(1-e^{-\frac{2\eps}{1+2\eps}}\right)\right\rbrace.\] For every fixed $\eps>0$, \[\sum_{i=r_m}^m S_m(i)\PP(W_i|\XX)\ge \sum_{i=r_m}^m \PP(W_i|\XX)-\sum_{i=r_m}^m \left(\frac{i-1}{m+1}\right)^{\frac{2\eps}{1+2\eps}}\PP(W_i|\XX)\ge\frac{1}{e}-f_m(\eps)\longrightarrow g(\eps)\] as $m\longrightarrow\infty$, where we used \Cref{W,W_i,1/e} and the usual integral estimate in the second inequality, whereas the asymptotics follows by \Cref{waitingtime}, the fact that 
		\[f_m(\eps)\longrightarrow f(\eps)\defeq \frac{1}{e}\left[1+\frac{1}{2\eps}\right]\left(1-e^{-\frac{2\eps}{1+2\eps}}\right),\] 
		and $g(\eps)=\sfrac{1}{e}-f(\eps)$. It is crucial that $g(\eps)$ is positive and strictly increasing on $(0,\infty)$: even though as $\eps\longrightarrow0$, $f(\eps)=\sfrac{1}{e}+\bigO\left(\eps\right)$, yielding that $g(\eps)=\bigO(\eps)$, no matter how small, it still provides a constant approximation for any $\eps$ fixed.
		
		Putting together all the above yields, as $m\longrightarrow\infty$, $\EE[S_m(\zeta_{r_m})X_{\pi_{\zeta_{r_m}}}|\XX]\ge (g(\eps)-\delta)M_m$, where $\delta\longrightarrow 0^+$ as $m\longrightarrow\infty$. This implies the final estimate \[\EE X_{\tau_m}\defeq\EE Y_{\tau_m}^\ssup{H_m}\ge (g(\eps)-\delta)\EE M_m.\] Fix $\eps>0$ small, and consider horizons $\{H_m\}_{m\ge M}$ as defined in \Cref{hardhorizon}, where $M=M(\eps)\in\NN$ is large enough. By \Cref{hardsingle}, the RH problem under this family does not admit single-threshold constant-approximations for the fixed $\eps$, assuming it is small enough (the hard instance for the values being the Pareto distribution defined in \Cref{hardval}). 
		On the other hand, by our final estimate, $\tau_m$ provides an approximate $g(\eps)$-approximation: for every (small) $\eps>0$ fixed, setting $M$ large enough, using $H_m\le m$ yields that for every instance $X$,
		\[\frac{\EE X_{\tau_m}}{\EE M_{H_m}}\ge g(\eps)-\delta,\]
		where $\delta=\delta(M)\longrightarrow 0$ as $M\longrightarrow\infty$.
	\end{proof}
	
	\section{A $2$-approximation for concentrated $\mathcal{L}^2$-horizons}\label{L2}
	In this final section we prove \Cref{L2horizon}, thus extending the $2$-approximation for the $\PGF$ class to sufficiently concentrated horizons. Besides its practical value, this will reveal interesting connections with deterministic prophet inequalities. We will make use of the following stochastic order.
	\begin{definition}[Laplace transform order]
		Given nonnegative random variables $H$ and $G$, $G$ is dominated by $H$ in the \emph{Laplace transform (Lt) order}, denoted as $G\prec_{Lt} H$, if $\forall\:s>0$, $\EE e^{-sG}\ge\EE e^{-sH}$.
	\end{definition}
	\begin{remark}\label{bernoulli}
		Consider a horizon $H$, having finite expectation $\mu$ and variance $\sigma^2$, and random variables $B\sim\Ber\left(\frac{\mu^2}{\mu^2+\sigma^2}\right)$ and $G\defeq\frac{\mu^2+\sigma^2}{\mu}B$. Then $G\prec_{Lt}H$.
	\end{remark}
	We are now ready to prove \Cref{L2horizon}. The details of the proof are provided in \Cref{suppL2}.
	\begin{proof}[Idea of the proof of \Cref{L2horizon}]
		Recall that the single-threshold $2$-approximation of \Cref{Ghorizon} obtains at least competitive ratio $c_{p,\bar{q}}=1-\EE\left[(1-\sfrac{1}{\mu})^H\right]=1-\EE \left( e^{-\bar{s}H}\right)$ with $\bar{s}\defeq-\log(1-\sfrac{1}{\mu})$. Since $G\prec_{Lt}H$, where $G$ is the two-point distribution of \Cref{bernoulli}, 
		it follows that
		\begin{equation}\label{deterministic}
			c_{p,\bar{q}}\ge1-\EE\left( e^{-\bar{s}G}\right)=\frac{\mu^2}{\mu^2+\sigma^2}\left[1-\left(1-\frac{1}{\mu}\right)^\frac{\mu^2+\sigma^2}{\mu}\right].
		\end{equation}
		As a result, an inequality that yields a single-threshold $2-\sfrac{1}{\mu}$-approximation is
		\begin{equation}\label{replaceC}
			\frac{\mu^2}{\mu^2+\sigma^2}\left[1-\left(1-\frac{1}{\mu}\right)^\frac{\mu^2+\sigma^2}{\mu}\right]\ge\frac{1}{2-\frac{1}{\mu}},
		\end{equation}
		and straightforward estimates show that this is ensured by $e^{-x}\le1-\frac{x}{2-\sfrac{1}{\mu}}$, having defined $x\defeq 1+\sfrac{\sigma^2}{\mu^2}$. Defining constraints $y\defeq x-(2-\sfrac{1}{\mu})>-1$ and $-2<\bar{y}\defeq-(2-\sfrac{1}{\mu})<-1$, we can rewrite the condition above as $ye^y\le\bar{y}e^{\bar{y}}$. Under the constraints given, this is more explicitly rewritten as $y\le W_0(\bar{y}e^{\bar{y}})$, where $W_0$ denotes the principal branch of the Lambert function \footnote{Recall that the Lambert function is the solution to the equation $ye^y=z$ over the complex numbers, and thus a function with branches. Over $\RR$ the Lambert function is described through a principal branch, denoted as $W_0(z)$, and a secondary branch, denoted as $W_{-1}(z)$, since, as evident from the graph of the function $ye^y$ on $(-\infty,0]$, $y=-1$ is a branching point for the solution. $W_{-1}(z)$ takes care of the branch corresponding to $y\le-1$, $W_{0}(z)$ of the one corresponding to $y\ge-1$.} . Recalling the definition of $y$ and $x$, this yields \Cref{concentration}.
	\end{proof}
	On a practical note, the Lambert function is readily simulated in various languages.\footnote{For example in Python, with standard error bounds of $10^{-8}$, via scipy.special.lambertw. Higher precision alternatives also exist (for example through the bisection method).} Thus it is possible to reliably implement the substitution of suitable values of $C>1$ for $2-\sfrac{1}{\mu}$ in \Cref{concentration}, even without large-market hypothesis (the latter was used to heuristically exemplify this fact in \Cref{contributions}). The only accuracy issue would arise near the branching point of the Lambert function, that is, when $\bar{y}\approx-1$, which is an uninteresting range for RH, since it only concerns horizons with $2-\sfrac{1}{\mu}\approx 1$, that is $\mu\approx 1$, which approaches a degenerate case at which the gambler and the prophet perform equally, since there is only one item to allocate. With this caveat, we conclude by discussing how to derive \Cref{largemarket}. 
	
	As direct consequence of substituting \textit{admissible} values of $C>1$ for $2-\sfrac{1}{\mu}$ in \Cref{replaceC} we obtain the upper-bound \[\text{CV}\le\sqrt{1-\sfrac{1}{\mu}+W_0(-Ce^{-C})}=\sqrt{C-1+W_0(-Ce^{-C})},\] which ensures a $C$-approximation. 
	This highlights interesting degenerate cases, and establishes an important connection with the existing literature. Note that the deterministic case ($\text{CV}=0$) can be trivially recovered directly from \Cref{deterministic}, which can be reformulated as a guarantee that the competitive ratio (given by $1/C$ here) is always greater than \[\frac{1-\left(1-\mu^{-1}\right)^{\mu(1+\text{CV}^2)}}{1+\text{CV}^2}\ge\frac{1-e^{-(1+\text{CV}^2)}}{1+\text{CV}^2}.\]
	Therefore if $\text{CV}=0$, the optimal single-threshold algorithm exploited ensures a competitive ratio of $1-\sfrac{1}{e}$. This value is the global maximum of the function on the right-hand side, and coincides with the optimal competitive ratio for single-threshold algorithms for (deterministic) IID PI \cite{Esh18}[Theorem~21]. Note that the left-hand side tends to the right-hand side in the large-market limit, at which regime the admissible values of $C$ aforementioned are easily computed to be constrained to be larger than $e/(e-1)$, and therefore a competitive ratio of $1-\sfrac{1}{e}$ cannot be exceeded even in a best-case analysis (with respect to the horizon). More in general, dropping the large-market limit, the left-hand side also achieves global maximum $1-\left(1-\sfrac{1}{\mu}\right)^{\mu}$ at $\text{CV}=0$ (this is not surprising, since the deterministic horizon is no harder than the random horizon for the gambler). Hence $C$ will always be constrained to be greater than $[1-\left(1-\sfrac{1}{\mu}\right)^{\mu}]^{-1}$. This proves \Cref{largemarket} and justifies the admissibility condition.

    \section{Conclusion}
    The contributions of this work can be divided into two main groups. The first group of results (in \Cref{Gclass,L2}) extend the distributional classes to which single-threshold algorithms can be applied, which impacts areas such as pricing and mechanism design of online auctions for the following reasons. The classes of distributions in $\PGF$ are well established across several areas, in particular in reliability theory. Furthermore, horizons classified in terms of concentration measures such as the $\text{CV}$ can be of wide use in applications, since their classification only requires the first two moments. Single-threshold algorithms are simple to implement, and thus proving their competitiveness and optimality for as many distributional classes as possible is advantageous. Although our main focus is on results for single item RH, they can be generalised to any number of items via the mechanism designed in \cite{AliBanGolMunWan20}. Furthermore, it is likely that similar results related to optimal stopping problems with $\IHR$ random horizons may benefit from a similar extension to the $\PGF$ class
    . The key to this extension is clear from the fact that our proof, in essence, relies only on certain stochastic ordering properties of the $\IHR$ class, which is popular not only in reliability theory, but in economics and computer science. RH has several connections to interesting online stochastic matching problems \cite{AouMa23}. A potential implementation of the techniques we established in this work would be extending the use of stochastic orders to these models and see if they improve known guarantees.
        
    The second group of results (in \Cref{opt,hardsecretary}) provide a more theoretical contribution. On one end, extending the characterisation of the optimal algorithm to unbounded horizons is of help in formalising hardness results. On the other hand, it has also helped significantly in showing that there are instances which are hard for single-thresholds, but for which a multiple-thresholds algorithm such as the SP rule is competitive. This last fact, combined with the fact that single-threshold algorithms already fail on the $\PGFd$ class, suggests that the study of RH under (most likely adaptive) multiple-threshold algorithms is a new and promising area of research.
	\section*{Acknowledgements.}This research was partially supported by the EPSRC grants EP/W005573/1 and EP/X021696/1, the ERC CoG 863818 (ForM-SMArt) grant, the ANID Chile grant ACT210005 and the French Agence Nationale de la Recherche (ANR) under reference ANR-21-CE40-0020 (CONVERGENCE project). We would like to thank Jos\'{e} Correa for his precious advice, Bruno Ziliotto and Vasilis Livanos for early conversations.    
	\bibliographystyle{abbrv}
	\bibliography{RandomHorizon}
	\clearpage
        \appendix
        \renewcommand{\appendixpagename}{Supplementary materials}
	\appendixpage
        \section{Proofs omitted from Section \ref{opt}}\label{suppopt}
        
        \subsection{Proof of \Cref{optalg}}
        \begin{proof}[Proof of \Cref{optalg}]
		If $m<\infty$, without loss of generality we can trivially reduce the discounted problem $\DZ^\ssup{m}$, that is $Z_0, Z_1, \ldots, Z_m, 0,\ldots$, to the finite sequence $\DZ^m$, that is $Z_0, Z_1, \ldots, Z_m$ with filtration $\DF^m\defeq \{\DF_0,\DF_1,\ldots,\DF_m\}$. Thus by \cite[Theorem~2.1]{Sam96}, which uses backward induction \cite[Theorem~3.2]{ChoRobSieg71} to construct $\bar{\zeta}$, it follows that the natural candidate $\bar{\tau}\defeq \bar{\zeta}\wedge (H+1)$ yields $\Val(\YY)$.
		
		Assume $m=\infty$.	We start by constructing, for every stopping rule $\sigma$ for the problem relative to $\YY$, a stopping rule $\zeta$ for the problem relative to $\DZ$, such that $\EE Y_\sigma=\EE Y_\tau$, where $\tau\defeq\zeta\wedge(H+1)$. It will be useful to rewrite $\FF_i=\sigma(\DF_i,\HH_{i})$, where $\HH_i\defeq\sigma(\{H=0\},\ldots,\{H=i-1\})$. The construction will require a suitable decomposition of the probability space $\Omega$. For every $i\in\NN$, let $E_i\defeq\{\sigma = i\}\cap\{H\ge i\}$. Since $\{\sigma = i\}\in\FF_i$, we have that $E_i\in\FF_i\cap\{H\ge i\}=\DF_i\cap\{H\ge i\}$, where the intersections denote the corresponding \textit{trace $\sigma$-algebras}. The equality is seen as follows: since  $\FF_i=\sigma(\DF_i,\HH_{i})\defeq\sigma(A\cap B,\, A\in\DF_i,\, B\in\HH_i)$, we have that $\FF_i\cap\{H\ge i\}\defeq\sigma(A\cap B\cap\{H\ge i\},\, A\in\DF_i,\, B\in\HH_i)$. By definition, for every $B\in\HH_i$, $B\cap\{H\ge i\}$ is either $\emptyset$ or $\{H\ge i\}$. Thus the only nonempty intersections are $A\cap B\cap\{H\ge i\}=A\cap\{H\ge i\}$. 
		
		Since $E_i\in\DF_i\cap\{H\ge i\}$, by definition there exists $A_i\in\DF_i$ such that $E_i=A_i\cap \{H\ge i\}$. Note that the events $\{\sigma=i\}$ constitute a partitioning of $\Omega$, while $\{H\ge i\}\supseteq\{H\ge i+1\}$. Furthermore, without loss of generality, the events $\{A_i\}$ can be assumed mutually disjoint. This can be seen by induction: suppose we have $E_i=A_i\cap \{H\ge i\}$ and $E_{i+1}=A'_{i+1}\cap \{H\ge i+1\}$, with $A'_{i+1}\cap A_i\neq\emptyset$. Since $E_i\subseteq\{\sigma=i\}$ and $E_{i+1}\subseteq\{\sigma=i+1\}$, $E_i\cap E_{i+1}=\emptyset$. Thus $\emptyset=(A_i\cap\{H\ge i\})\cap(A'_{i+1}\cap\{H\ge i+1\})\supseteq(A_i\cap\{H\ge i+1\})\cap(A'_{i+1}\cap\{H\ge i+1\})$. This means that the parts of $A'_{i+1}$ and $A_i$, which intersect, lie outside $\{H\ge i+1\}$. Define $A_{i+1}\defeq A'_{i+1}\setminus A_i$. Since $A_{i+1}$ is only missing a part of $A'_{i+1}$ that lies outside $\{H\ge i+1\}$, we also have $E_{i+1}=A_{i+1}\cap\{H\ge i+1\}$. Furthermore, $A_{i+1}=A'_{i+1}\cap A_i^c\in\DF_{i+1}$, because $A'_{i+1}\in\DF_{i+1}$ and $A_i\in\DF_i\subseteq\DF_{i+1}$. Thus we replace $A'_{i+1}$ with $A_{i+1}$ and the induction step is complete. Starting with $i=1$, and performing inductively this substitution whenever necessary, yields the mutually disjoint $\{A_i\}$.
		
		We can finally construct $\zeta$ as, informally speaking, a stopping rule that stops when $\sigma$ successfully stops by the time the horizon has realised, and does not stop otherwise. The precise measure theoretic construction is as follows: given the mutually disjoint $\{A_i\}$ as constructed above, recall that since $A_i\in\DF_i$, we can denote $C\defeq\bigcap_{i\in\NN}A_i^c=\left(\bigcup_{i\in\NN}A_i\right)^c\in\DF_\infty$, therefore we have that
		\begin{equation}\label{equivstopapp}
			\zeta(\omega)\defeq\begin{cases}i,&\omega\in A_i,\,\forall i\in\NN\\
			\infty,&\omega\in C\end{cases}
		\end{equation}
		is a valid stopping rule in $\TT$. Moreover, $\tau\defeq \zeta\wedge(H+1)$ is a valid finite stopping rule in $\TT^*$, since for every $i\in\NN$, both $\{\zeta=i\}$ and $\{H+1=i\}=\{H=i-1\}$ are adapted to $\FF_i=\sigma(\DF_i,\HH_{i})$, that is both $\zeta$ and $H+1$ are valid stopping rules for the problem, and the minimum of stopping rules is a stopping rule.
		
		The stopping rules $\sigma$ and $\tau$ are equivalent, that is $\EE Y_\sigma=\EE Y_\tau$. In fact, for every $i\in\NN$, $\EE Y_i\II_{\{\sigma=i,\,H<i\}}=0=\EE Y_i\II_{\{\tau=i,\,H<i\}}$ by definition of $Y_i$, which is null on $\{H<i\}$. We also have that $\EE Y_\infty=0$. Following the standard convention that  \[\EE Y_{\sigma}\defeq \EE\sum_{i\in\NN}Y_i\II_{\{\sigma=i\}}+\EE Y_\infty,\] the equivalence follows from showing that for every $i\in\NN$, $\EE Y_i\II_{\{\sigma=i,\,H\ge i\}}=\EE Y_i\II_{\{\tau=i,H\ge i\}}$, since the expectation can be exchanged with the summation in our setting (as it can be seen in the next paragraph, it is enough to use \textit{Fubini's theorem}). By construction of the events $\{A_i\}$ and the stopping rules $\zeta$ and $\tau$, we have that, indeed, \[\EE Y_i\II_{\{\sigma=i,\,H\ge i\}}=\EE Y_i\II_{A_i\cap \{H\ge i\}}=\EE Y_i\II_{ \{\zeta=i,\,H\ge i\}}=\EE Y_i\II_{\{\tau=i,\,H\ge i\}}.\]
		
		We show first that $\Val(\YY)=\Val(\DZ)$. First note that, by the construction of the stopping rule $\zeta\in\TT$ in \Cref{equivstopapp} in terms of the stopping rule $\sigma\in\TT^*$, we have that, since $Y_\infty=0$,
		\[Y_\sigma=\sum_{i\in\NN}X_i\II_{\{H\ge i\}}\II_{\{\sigma=i\}}=\sum_{i\in\NN}X_i\II_{E_i}=\sum_{i\in\NN}X_i\II_{A_i\cap\{H\ge i\}}=\sum_{i\in\NN}X_i\II_{\{H\ge i\}}\II_{\{\zeta=i\}},\]
		thus yielding
		\begin{align}\label{optequivapp'}
			\EE Y_\sigma&=\sum_{i\in\NN}\EE X_i\II_{\{H\ge i\}}\II_{\{\zeta=i\}}=\sum_{i\in\NN}\EE X_i\II_{\{\zeta=i\}}\EE_{\DF_i}\II_{\{H\ge i\}}\notag\\&=\sum_{i\in\NN}S(i)\EE X_i\II_{\{\zeta=i\}}=\EE\sum_{i\in\NN}S(i) X_i\II_{\{\zeta=i\}}=\EE Z_\zeta.
		\end{align}
		Here the expectation can be exchanged with the summation since \[\sum_{i\in\NN}S(i)\EE X_i\II_{\{\zeta=i\}}\leq\nu \sum_{i\in\NN}S(i)=\nu\mu<\infty.\] Secondly we note that, by a similar computation, for any stopping rule $\zeta\in\TT$, $\tau\defeq\zeta\wedge(H+1)$ is such that
		\begin{align}\label{optequivapp}
			\EE Z_{\zeta}&=\EE\sum_{i\in\NN}S(i) X_i\II_{\{\zeta=i\}}=\sum_{i\in\NN}\EE X_i\II_{\{\zeta=i\}}\EE_{\DF_i}\II_{\{H\ge i\}}=\EE\sum_{i\in\NN} X_i\II_{\{\zeta=i\}}\II_{\{H\ge i\}}\notag\\&=\EE\sum_{i\in\NN} X_i\II_{\{\tau=i\}}\II_{\{H\ge i\}}=\EE\sum_{i\in\NN} Y_i\II_{\{\tau=i\}}=\EE Y_{\tau}.
		\end{align}
		These two facts together imply that it is not possible to have $\sup_{\sigma\in\TT^*}\EE Y_\sigma\neq \sup_{\zeta\in\TT}\EE Z_\zeta$. Indeed, we cannot have that $\sup_{\sigma\in\TT^*}\EE Y_\sigma> \sup_{\zeta\in\TT}\EE Z_\zeta$, because otherwise there would be some $\sigma\in\TT^*$ such that $\EE Y_\sigma>\EE Z_\zeta$ for all $\zeta\in\TT$, but this is not possible since by \Cref{optequivapp'} for every $\sigma\in\TT^*$ there is a $\zeta\in\TT$ such that $\EE Y_\sigma = \EE Z_\zeta$.  We cannot have that $\sup_{\sigma\in\TT^*}\EE Y_\sigma< \sup_{\zeta\in\TT}\EE Z_\zeta$ either, because otherwise there would be some $\zeta\in\TT$ such that $\EE Y_\sigma<\EE Z_\zeta$ for all $\sigma\in\TT^*$, but this is not possible since by \Cref{optequivapp} for every $\zeta\in\TT$,  $\zeta\wedge(H+1)\defeq\tau\in\TT^*$ is such that $\EE Y_\tau = \EE Z_\zeta$. Therefore $\Val(\YY)=\Val(\DZ)$ follows by definition.
		
		Given an optimal stopping rule $\bar{\zeta}\in\TT$, we have that $\bar{\tau}\defeq\bar{\zeta}\wedge(H+1)$ provides an optimal stopping rule in $\TT^*$ by \Cref{optequivapp} with $\zeta=\bar{\zeta}$ and $\tau=\bar{\tau}$. Since \[\EE \sup_{i\in\NN} Z_i\leq\EE\sum_{i\in\NN} Z_i = \sum_{i\in\NN}S(i)\EE X_i = \mu \nu<\infty\] and $\limsup_{i\longrightarrow\infty} Z_i=Z_\infty=0$, the Snell envelope of $\DZ$, denoted $\VV\defeq\{V_i\}_{i\in\NN_0}$, satisfies \Cref{dpe} by \Cref{snell} and yields the optimal stopping rule \[\bar{\zeta}=\inf\{i\in\NN_0:\,Z_i\ge \EE_{\DF_i}V_{i+1}\}\]
		by \Cref{snellstop}.
	\end{proof}
	
	\subsection{Proof of \Cref{snellfuture}}
	\begin{proof}[Proof of \Cref{snellfuture}]
		Note that $\esup_{\zeta\in\TT_{i+1}}\EE_{\DF_i}Z_\zeta\le \EE_{\DF_i}V_{i+1}$ is straightforward. By definition, $\EE_{\DF_{i+1}}Z_\zeta\le V_{i+1} \as$ for every $\zeta\in\TT_{i+1}$. Thus, taking expectations, $\EE_{\DF_{i}}Z_\zeta\le \EE_{\DF_{i}}V_{i+1} \as$ for every $\zeta\in\TT_{i+1}$.
		
		The argument consists in finding a lower bound to $\esup_{\zeta\in\TT_{i+1}}\EE_{\DF_i}Z_\zeta$ which we then show to converge to $\EE_{\DF_i}V_{i+1}$, yielding the claim by the \textit{squeeze theorem}. This is done by exploiting \textit{regular stopping rules}. We briefly recall a few standard results regarding these, which will be used. Let $i\in\NN_0$.
		\begin{enumerate}[i), noitemsep]
			\item A stopping rule $\zeta\in\TT_i$ is \textit{regular from $i$ on}, if for every $k\ge i$, $\EE_{\DF_k}Z_\zeta>Z_k \as$ on $\{\zeta>k\}$.
			\item If $\EE\sup_{i\in\NN_0}Z_i<\infty$, for any stopping rule $\zeta\in\TT_i$ there exists a stopping rule $\zeta'\in\TT_i$ regular from $i$ on, such that $\EE_{\DF_i}Z_{\zeta'}\ge \EE_{\DF_i}Z_{\zeta} \as$
			\item If $\EE\sup_{i\in\NN_0}Z_i<\infty$ and both stopping rules $\zeta\in\TT_i$ and $\zeta'\in\TT_i$ are regular from $i$ on, then so is $\zeta''\defeq \zeta\vee\zeta'$ and $\EE_{\DF_i}Z_{\zeta''}\ge \EE_{\DF_i}Z_{\zeta}\vee\EE_{\DF_i}Z_{\zeta'} \as$
		\end{enumerate}
		By a standard property of the essential supremum, there is a sequence of stopping rules in $\TT_{i+1}$, $\{\zeta_k\}_{k\in\NN}$, such that 
		\begin{equation}\label{mon}
			V_{i+1}=\sup_{k\in\NN}\EE_{\DF_{i+1}}Z_{\zeta_k}\uparrow \max_{i\le j\le k}\EE_{\DF_{i+1}}Z_{\zeta_j}
		\end{equation} 
		as $k$ grows. By (ii), to every stopping rule $\zeta_j$ it corresponds a $\zeta'_j$ regular from $i+1$ on, such that $\EE_{\DF_i}Z_{\zeta'}\ge \EE_{\DF_i}Z_{\zeta}$. Let $\zeta''_k\defeq \max_{j\in [k]}\zeta'_j\in\TT_{i+1}$. By (iii), $\EE_{\DF_{i+1}}Z_{\zeta''_k}\ge \max_{j\in [k]}\EE_{\DF_{i+1}}Z_{\zeta'_j}$. Then 
		\[\esup_{\zeta\in\TT_{i+1}}\EE_{\DF_i}Z_\zeta\ge \EE_{\DF_i}Z_{\zeta''_k}= \EE_{\DF_i}\EE_{\DF_{i+1}}Z_{\zeta''_k}\ge \EE_{\DF_i}\max_{j\in [k]}\EE_{\DF_{i+1}}Z_{\zeta'_j}\ge\EE_{\DF_i}\max_{j\in [k]}\EE_{\DF_{i+1}}Z_{\zeta_j}\longrightarrow\EE_{\DF_i}V_{i+1} \as\]
		as $k$ grows, where the limit follows by \Cref{mon} and the \textit{monotone convergence theorem}.
	\end{proof}
 
	\section{Proofs omitted from Section \ref{Gclass}}\label{suppGclass}
	
	\subsection{Proof of \Cref{algp}}\label{suppGclassalgp}
	\begin{proof}[Proof of \Cref{algp}]
		Consider the following decomposition. For every $h\in\NN$, let
		\begin{align*}
			E_\pi^h&\defeq \{H=h\}\cap\bigcup_{i\in[h]}\{i=\inf\{j\in[h]:\,X_j\ge \pi\}\},\\
			E_0^h&\defeq \{H=h\}\cap\bigcap_{i\in[h]}\{X_i<\pi\}.
		\end{align*}
		Then $\{H=h\}= E_\pi^h\cup E_0^h$. Intuitively, $E_\pi^h$ is the part of the event $\{H=h\}$ on which the algorithm $\tau_\pi$ stops, achieving at least $\pi$, while $E_0^h$ is the part on which it fails to stop, achieving zero. By partitioning the probability space $\Omega=\bigcup_{h\in\NN}(E_\pi^h\cup E_0^h)$ we can rewrite, exploiting that on $E_0^h$, $Y_{\tau_\pi}=Y_\infty =0$,  
		\[\EE Y_{\tau_\pi}=\sum_{h\in\NN}\EE\left( Y_{\tau_\pi}\II_{E_\pi^h}\right).\] 
		The exchange of expectation and summation holds \textit{a posteriori}, since we will obtain a finite quantity for the right-hand side. By the independence of $H$ and $X$,
		\begin{align*}
			\EE \left(Y_{\tau_\pi}\II_{E_\pi^h}\right)&=\EE\left( Y_{\tau_\pi}\II_{\{H=h\}}\sum_{i\in[h]}\II_{\{i=\inf\{j\in[h]:\,X_j\ge \pi\}\}}\right)=\sum_{i\in[h]}\EE\left( X_i\II_{\{H=h\}}\II_{\{i=\inf\{j\in[h]:\,X_j\ge \pi\}\}}\right)\\&=\PP(H=h)\sum_{i\in[h]}\EE \left(X_i\II_{\{X_i\ge \pi\}}\prod_{j\in[i-1]}\II_{\{X_j< \pi\}}\right)\\&=\PP(H=h)\sum_{i\in[h]}\PP^{i-1}(X<\pi)\EE \left(X_i\II_{\{X_i\ge \pi\}}\right).
		\end{align*}
		Since $\EE \left(X_i\II_{\{X_i\ge \pi\}}\right)=\EE(X|X\ge\pi)\PP(X\ge\pi)$ for all $i\in\NN$ and \[\sum_{i\in[h]}\PP^{i-1}(X<\pi)=\frac{1-\PP^h(X<\pi)}{\PP(X\ge\pi)},\]
		it follows that $\EE \left(X_{\tau_\pi}\II_{E_\pi^h}\right)=\EE(X|X\ge\pi)[1-\PP^h(X<\pi)]\PP(H=h)$.
		Hence \[\EE X_{\tau_\pi}=\EE Y_{\tau_\pi}=\EE(X|X\ge\pi)\sum_{h\in\NN}[1-\PP^h(X<\pi)]\PP(H=h)=\EE(X|X\ge\pi)\EE[1-\PP^H(X<\pi)],\] which yields the result.
	\end{proof}
	
	\subsection{Proof of \Cref{max} and discussion on amendments to \cite{AliBanGolMunWan20}}\label{suppGclassmax}
	Before giving the details of the proof of \Cref{max}, we further discuss the need for a variational argument to upper-bound $\EE M_H$ for any continuous $X$. The key idea at the base of this argument is similar to that of \cite[Theorem~3.2]{AliBanGolMunWan20}, but they do not state what restrictions their argument imposes on $X$. As a matter of fact, they achieve the upper bound only for $X$ with discrete\footnote{For the value distribution, unlike for horizons, we use \textit{discrete} in the standard, more general, sense.} distribution and admitting a value $p$ such that $\PP(X\geq p)=\sfrac{1}{\mu}$, where $\mu=\EE H$. If their claim was true as stated, the resulting algorithm would not have a consistently defined threshold for more general distributions. It is not clear how these results should be extended from this special class to a more general one, nor it is mentioned by the authors. Perhaps, they have in mind discretised continuous distributions. Yet, this would mean that by some approximation or bounding argument, their result only applies to continuous distributions. We could not find a straightforward way, which might suggest this type of argument as trivial. Examples from the literature on other stochastic optimisation problems, such as the bilateral trade problem, adopt a similar approach, of discretising and then bounding, but the proofs of the bounding step are given and usually nontrivial. See for example \cite{CaiWu23}.
    \begin{proof}[Proof of \Cref{max}]
		Let $f(x)$ be the pdf and $F(x)$ the cdf of $M_H$. The pdf exists by absolute continuity of $X$ and $f=F'$. Let $v(x)$ be the pdf and $V(x)$ the cdf of the value $X$. We are seeking a upper bound on $\EE M_H=\int_0^\infty x f(x)dx\eqdef I(f)$, which will therefore be the functional to be maximised over $f\in\text{L}_+^1([0,\infty))$, the set of nonnegative Lebesgue-integrable functions on $[0,\infty)$. Since $f$ is a density supported on the nonnegative reals, we also have a first constraint $\int_0^\infty f(x)dx=1$ for the maximisation problem. Furthermore the law of total probability holds. Thus for every $x\in[0,\infty)$, as $\eps\longrightarrow 0$,
		\begin{align*}
			f(x)&\longleftarrow\frac{\PP(x< M_H\le x+\eps)}{\eps}=\sum_{h\in\NN}\frac{\PP(x< M_h\le x+\eps)}{\eps}\PP(H=h)\\&\le\sum_{h\in\NN}h\frac{\PP(x< X\le x+\eps)}{\eps}\PP(H=h)\longrightarrow\mu v(x),
		\end{align*}
		which implies that a second constraint holds: $0\le f(x)\le\mu v(x)$ for all $x\in[0,\infty)$. The inequality above can be shown as follows. Denote by $F_h(x)$ the conditional cdf of $M_H$, given the event $\{H=h\}$, for all $h\in\NN$ such that $\PP(H=h)>0$. Then for all $x\in(0,\infty)$,
		\begin{align*}
			\PP(x< M_h\le x+\eps)&=F_h(x+\eps)-F_h(x)=V^h(x+\eps)-V^h(x)=[V(x+\eps)-V(x)]\cdot\\&\sum_{i=0}^{h-1}V^{h-i}(x+\eps)V^i(x)\le h[V(x+\eps)-V(x)]=h\PP(x<X<x+\eps).
		\end{align*}
		For $x=0$, it follows trivially from $\PP(0<M_h\le \eps)=F_h(\eps)=V^h(\eps)\le hV(\eps)=h\PP(0<X<\eps)$.
		
		In conclusion we seek to solve the constrained maximisation variational problem
		\begin{align}
			I(f)&\defeq \int_0^\infty x f(x)dx\notag\\&\int_0^\infty f(x)dx=1\notag\\&0\le f(x)\le\mu v(x),\:\forall\,x\in[0,\infty)\notag\\&f\in\text{L}_+^1([0,\infty)).\label{var}
		\end{align}
		The functional $I(f)$ is trivially maximised by choosing a function $\bar{f}$ that is as large as possible on as many large values as possible, since they contribute the most to the integral, that is
		\[\bar{f}(x)=\begin{cases}
			\mu v(x),&x\geq p\\
			0,&0\le x <p,
		\end{cases}\]
		where $p=\inf \left\lbrace y\in[0,\infty):\,\int_y^\infty \mu v(x)dx= 1\right\rbrace$. Note that the condition is equivalent to \[p=\inf \left\lbrace y\in[0,\infty):\,\PP(X\ge y)= \sfrac{1}{\mu}\right\rbrace,\] and since the survival function of $X$ is assumed continuous, there is only one value $y$ yielding equality, meaning that $p$ is the only value such that $\PP(X\ge p)= \sfrac{1}{\mu}$. The maximal value of the variational problem is therefore \[I(\bar{f})=\int_p^\infty x \mu v(x)dx=\int_p^\infty x \frac{v(x)}{\PP(X\ge p)}dx=\frac{\EE \left(X\II_{\{X\ge p\}}\right)}{\PP(X\ge p)}=\EE(X|X\ge p).\]
	\end{proof}

	\subsection{Proof of \Cref{algran}}\label{suppGclassalgran}
	\begin{proof}[Proof of \Cref{algran}]
		Consider the following decomposition. For every $h\ge 1$, let
		\begin{align*}
			E_\pi^h&\defeq \{H=h\}\cap\bigcup_{i\in[h]}\{i=\inf\{j\in[h]:\,X_j\ge \pi,\,B_j=1\}\},\\
			E_0^h&\defeq \{H=h\}\cap\bigcap_{i\in[h]}\{i=\inf\{j\in[h]:\,X_j\ge \pi,\,B_j=1\}\}^c.
		\end{align*}
		Then $\{H=h\}= E_\pi^h\cup E_0^h$. Intuitively, $E_\pi^h$ is the part of the event $\{H=h\}$ on which the randomised algorithm $\tau_\pi$ stops, achieving at least $\pi$, while $E_0^h$ is the part on which it fails to stop, achieving zero. By partitioning the probability space $\Omega=\bigcup_{h\in\NN}(E_\pi^h\cup E_0^h)$ we can rewrite, exploiting that on $E_0^h$, $Y_{\tau_\pi}=Y_\infty =0$,  
		\[\EE Y_{\tau_\pi}=\sum_{h\in\NN}\EE Y_{\tau_\pi}\II_{E_\pi^h}.\] 
		The exchange of expectation and summation holds \textit{a posteriori}, since we will obtain a finite quantity for the right-hand side. 
		
		From a combinatorial point of view, note that for every $i\in\NN$,
		\[\{i=\inf\{j\in[h]:\: X_j\ge \pi,\: B_j=1\}\}=\bigcup_{k=1}^{2^{i-1}}E_k^i\]
		having defined, informally, 
		\begin{align*}
			E_1^i&\defeq\{X_1<\pi,\:\ldots,\: X_{i-1}<\pi,\: X_i\ge\pi,\: B_i=1\}\\
			E_2^i&\defeq\{X_1\ge\pi,\: B_1=0,\: X_2<\pi\:\ldots,\: X_{i-1}<\pi,\: X_i\ge\pi,\: B_i=1\}\\
			\vdots&\\
			E_{2^{i-1}}^i&\defeq\{X_1\ge\pi,\: B_1=0,\: \ldots,\: X_{i-1}\ge\pi,\: B_{i-1}=0,\: X_i\ge\pi,\: B_i=1\},
		\end{align*}
		where the $2^{i-1}$ events run through all possible combinations of the first $i-1$ steps not allowing the randomised stopping rule to stop. For example, assuming $i\ge 3$ and denoting $^nC_k\defeq\binom{n}{k}$,
		\begin{itemize}[noitemsep]
			\item $E_1^i=E_{^{i-1}C_0}^i$ is the event where none of the first $i-1$ steps has a random variable surpassing the threshold $\pi$;  
			\item $E_{^{i-1}C_0+1}^i,\: \ldots,\: E_{^{i-1}C_0+^{i-1}C_1}^i$ are all those events having only one of the first $i-1$ steps surpassing the threshold $\pi$ and yet the corresponding coin flip represented by the Bernoulli variable not allowing  the randomised rule to stop at the value;
			\item$E_{^{i-1}C_0+^{i-1}C_1+1}^i,\:\ldots,\: E_{^{i-1}C_0+^{i-1}C_1+^{i-1}C_2}^i$ are all those events having only two of the first $i-1$ steps surpassing the threshold $\pi$ and yet the corresponding coin flips not allowing the randomised rule to stop at the values and so on.
		\end{itemize} 
		This yields a total of $\sum_{j=0}^{i-1}\,^{i-1}C_j=(1+1)^{i-1}$ disjoint events. Denote $E_k^i\defeq\overline{E}_k^i\cap\{X_i\ge\pi,\: B_i=1\}$, where the overline denotes the complement. By the independence of $X$, $B$ and $H$,
		\begin{align*}
			\EE \left(Y_{\tau_{\pi,q}}\II_{E_\pi^h}\right)&=\EE\left( Y_{\tau_{\pi,q}}\II_{\{H=h\}}\sum_{i\in[h]}\II_{\{i=\inf\{j\in[h]:\,X_j\ge \pi,\,B_j=1\}\}}\right)\\&=\sum_{i\in[h]}\EE \left(X_iB_i\II_{\{H=h\}}\II_{\{i=\inf\{j\in[h]:\,X_j\ge \pi\}\}}\right)\\&=\PP(H=h)\sum_{i\in[h]}\EE\left( X_iB_i\sum_{k=1}^{2^{i-1}}\II_{E_k^i}\right)=\PP(H=h)q\EE\left[ X\II_{\{X\ge\pi\}}\sum_{i\in[h]}\sum_{k=1}^{2^{i-1}}\PP(\overline{E}_k^i)\right].
		\end{align*}
		Since $\EE\left( X\II_{\{X\ge \pi\}}\right)=\EE(X|X\ge\pi)\PP(X\ge\pi)$ and \[\sum_{k=1}^{2^{i-1}}\PP(\overline{E}_k^i)=\sum_{l=0}^{i-1}\,^{i-1}C_l\PP^{i-1-l}(X<\pi)[(1-q)\PP(X\ge\pi)]^l=[1-q\PP(X\ge\pi)]^{i-1},\]
		it follows that $\EE\left( X_{\tau_\pi}\II_{E_\pi^h}\right)=\EE(X|X\ge\pi)\{1-[1-q\PP(X\ge\pi)]^h\}\PP(H=h)$.
		Hence 
		\begin{align*}
			\EE X_{\tau_{\pi,q}}&=\EE Y_{\tau_{\pi,q}}=\EE(X|X\ge\pi)\sum_{h\in\NN}\{1-[1-q\PP(X\ge\pi)]^h\}\PP(H=h)\\&
			=\EE(X|X\ge\pi)\EE\{1-[1-q\PP(X\ge\pi)]^H\},
		\end{align*}
		which yields the result.
	\end{proof}
	
	\subsection{Proof of \Cref{Ghorizon}}\label{suppGclassGhorizon}
	\begin{proof}[Proof of \Cref{Ghorizon}]
		For $X$ continuous, by \Cref{algp,max},
		\[\frac{\EE X_{\tau_p}}{\EE M_H}\ge c_p\defeq1-\EE\PP^H(X<p).\]
		For every $H\in\PGF$, it holds that $G\prec_{pgf}H$, with $G\sim\Geom(\sfrac{1}{\mu})$, which denotes the geometric distribution with mean $\mu=\EE H$. Thus, denoting $t\defeq \PP(X<p)=1-\sfrac{1}{\mu}$, 
		\[\EE \left(t^H\right)\le \EE\left( t^G\right)=\frac{1}{\mu}\sum_{h\in\NN}t^h\left(1-\frac{1}{\mu}\right)^{h-1}=\frac{t}{\mu}\frac{1}{1-t\left(1-\frac{1}{\mu}\right)}=\frac{1-\frac{1}{\mu}}{2-\frac{1}{\mu}}.\]
		Thus the result follows since this implies that 
		\[c_p\ge 1-\frac{1-\frac{1}{\mu}}{2-\frac{1}{\mu}}=\frac{1}{2-\frac{1}{\mu}}.\]
		
		For $X$ discontinuous, let $V(x)$ be the cdf of $X$. The case that is most difficult to handle by adapting the previous part of the argument, which assumes continuity, is when for some $p$, \[\PP(X\ge p)>\frac{1}{\mu}>\lim_{\eps\longrightarrow 0^+}\PP(X\ge p+\eps),\] or equivalently, \[\PP(X<p)<1-\frac{1}{\mu}<\lim_{\eps\longrightarrow0^+}\PP(X<p+\eps).\] Using the continuity from above of the cdf, this is equivalent to \[\lim_{\eps\longrightarrow0^+}V(p-\eps)<1-\frac{1}{\mu}<V(p),\] that is to having $1-\sfrac{1}{\mu}$ in correspondence of a jump of the cdf at $p$. We will therefore explicitly treat this scenario.\footnote{If $1-\sfrac{1}{\mu}$ is not in correspondence of a jump discontinuity but the distribution is nonetheless discontinuous, similar ideas can be put in place, without the issue, which will lead to the necessity of randomisation, being there.} By monotonicity, $V(x)$ has at most countably many jump discontinuities at values $\{j_k\}$, where $k$ ranges either in $\NN$ (infinitely many jumps) or in $[n]$ for some $n\in\NN$ (finitely many jumps). In the following we adopt the terminology corresponding to the infinite case, replacing \textit{sequence} with \textit{list} yields the corresponding statements for the finite case; furthermore, $\{j_k\}$ is understood as an increasing enumeration of all the discontinuity points, and denoted as $D$. As customary in the literature concerning the analysis of discontinuous distributions, we will assume that we are not in the pathological case of $D$ dense or having accumulation points. Finally, recall that under our assumptions, there is a value $\bar{k}$ such that $p=j_{\bar{k}}$. Let $\eps(l)$ be a small enough positive monotonically vanishing sequence. Since all jumps are isolated, for every $l\in\NN$ there exists an increasing subsequence of the jump discontinuities, denoted, with slight abuse of notation, $\{j_{k_{\ell}(l)}\}\subseteq D$, such that \[D\subseteq\bigcup_{\ell}[j_{k_{\ell}(l)}-\eps(l),j_{k_{\ell}(l)}]\eqdef U_l,\]
		where the union is disjoint. Note that, for some $l\in\NN$, it could happen that $j_{k_{1}(l)}-\eps(l)<0$. We now linearly interpolate $V(x)$ on the intervals $[j_{k_{\ell}(l)}-\eps(l),j_{k_{\ell}(l)}]$, and keep it unaltered otherwise, so as to construct an approximating sequence of cdf's $V_l(x)$. We define, for every fixed $l\in\NN$,
		\[V_l(x)\defeq\begin{cases}
			V(j_{k_{\ell}(l)}-\eps(l))+(x-j_{k_{\ell}(l)}+\eps(l))\frac{V(j_{k_{\ell}(l)})-V(j_{k_{\ell}(l)}-\eps(l))}{\eps(l)}\\\forall x\in[j_{k_{\ell}(l)}-\eps(l),j_{k_{\ell}(l)}],\:\forall\ell:\: j_{k_{\ell}(l)}-\eps(l)\ge0,\\
			x\frac{V(j_{k_{1}(l)})}{j_{k_{1}(l)}}\\\forall x\in[0,j_{k_{1}(l)}],\: \text{if}\: j_{k_{1}(l)}-\eps(l)<0,\\
			V(x)\\\text{otherwise}.
		\end{cases}\]
		Denote $\DD_l$ the distribution having cdf $V_l(x)$ and denote the corresponding continuous (since $V_l(x)$ is continuous and piecewise differentiable) random variables $X^\ssup{l}\sim\DD_l$. The family $\{X^\ssup{l}\}$ is dominated by $\eps(1)+X\in\mathcal{L}^1(\Omega)$ by construction, since 
		\[\EE X^\ssup{l}\le\int_0^\infty1-V(x)dx+\eps(l)\sum_{\ell}V(j_{k_\ell(l)})-V(j_{k_\ell(l)}-\eps(l))\le \EE X +\eps(1),\]
		having used \[1-V(j_{k_\ell(l)}-\eps(l))-[1-V(j_{k_\ell(l)})]=V(j_{k_\ell(l)})-V(j_{k_\ell(l)}-\eps(l))\le\PP(j_{k_\ell(l)}-\eps(l)\le X\le j_{k_\ell(l)}),\] $\PP(X\in U_l)\le 1$ and $\eps(l)\le\eps(1)$.
		Therefore, the family $\{X^\ssup{l}\}$ is Uniformly Integrable (UI). Note also that by construction $X^\ssup{l}\overset{w}{\longrightarrow}X$. Therefore, defining $M_H^\ssup{l}\defeq \max\{X_1^\ssup{l},\ldots,X_H^\ssup{l}\}$, where $X_i^\ssup{l}\sim \DD_l$ are \textit{iid}, we have that $\{M_H^\ssup{l}\}$ is UI and $M_H^\ssup{l}\overset{w}{\longrightarrow}M_H$. Using \cite[Corollary~5]{Bill87} it follows that $\EE M_H^\ssup{l}\longrightarrow\EE M_H$ as $l\longrightarrow\infty$. Furthermore, by \Cref{max}, for every $l$ there is $p_l$ such that $\PP(X^\ssup{l}\ge p_l)=\sfrac{1}{\mu}$ and $\EE M_H^\ssup{l}\le\EE(X^\ssup{l}|X^\ssup{l}\ge p_l)$. Since by construction $p_l\longrightarrow p$ as $l\longrightarrow\infty$, the fact that $\{X^\ssup{l}\}$ is UI and $X^\ssup{l}\overset{w}{\longrightarrow}X$ implies also that $\EE(X^\ssup{l}|X^\ssup{l}\ge p_l)\longrightarrow \EE(X|X\ge p)$ as $l\longrightarrow\infty$. These facts altogether imply that also in this case $\EE M_H\le \EE(X|X\ge p)$.
		
		To conclude, by \Cref{algran} the single-threshold randomised algorithm $\tau_{p,\bar{q}}$, where \[\bar{q}\defeq\frac{1}{\mu\PP(X\ge p)},\]
		is such that 
		\[\frac{\EE X_{\tau_{p,\bar{q}}}}{\EE M_H}\ge c_{p,\bar{q}}\defeq1-\EE\left\lbrace\left[1-\bar{q}\PP(X\ge p)\right]^H\right\rbrace=1-\EE\left[\left(1-\frac{1}{\mu}\right)^H\right]\ge\frac{1}{2-\frac{1}{\mu}},\]
		where the last inequality follows as in the continuous case, by exploiting that for any $H\in\PGF$, $G\prec_{pgf}H$.
	\end{proof}
	
	\subsection{Further details on the $\PGF$ class}\label{suppGclassdetails}
    For a better intuition of how much more general the somewhat abstract $\PGF$ class is, compared with the $\IHR$ class, we mention some of its most well-known subclasses, which earned considerable interest in applications. $H$ always denotes a horizon.
	\begin{definition}
		$H$ is \emph{New Better than Used ($\NBU$)} if for every $h,k\ge2$, $S(h+k)\le S(h)S(k)$.
	\end{definition}
	\begin{definition}
		For every $h\in\NN$ the \emph{mean residual life} of $H$ is defined as
		\[m(h)\defeq\begin{cases}
			\EE(H-h|H\ge h),&S(h)>0\\
			0,&S(h)=0.\end{cases}\]
	\end{definition}
	Note that $m(1)=\mu\defeq \EE H$.
	\begin{definition}\label{HNBUEdef}
		$H$ is \emph{Harmonically New Better than Used in Expectation ($\HNBUE$)} if for all $n\in\NN$, \[\frac{n}{\sum_{h=1}^n\frac{1}{m(h)}}\le \mu.\]
	\end{definition}
	Furthermore, consider all the other classes in the tower of inclusions introduced after \Cref{algorithm} in \Cref{Gclass}. For simplicity we omit introducing further definitions, as the denominations are suggestive enough to convey the intuitive meaning, and can be easily found in the literature (see, for example, \cite{BraqRoyXie01} for definitions and some of the inclusion properties we will recall). The $\HNBUE$ class, introduced in \cite{Rol75}, is roughly speaking the largest known, which is a subclass of the $\PGF$ class \cite{Klef82}, with intuitive meaning in terms of ageing properties.
	
	\subsection{Proof of \Cref{geomprice}}\label{suppGclassgeomprice}
        \begin{proof}[Proof of \Cref{geomprice}]
		Consider the RH for $X_i\sim X$ for all $i\in\NN$, to which we add $X_0=0$ as per the framework of \Cref{opt}. By \Cref{optalg}, and keeping the same notations and constructions, $\Val(\XX)=\Val(\YY)=\Val(\DZ)=V_0$. By \Cref{dpe}, $V_0=\EE V_1=\EE (Z_1\vee\EE_{\DF_1}V_{2}).$ Note that $Z_i\defeq q^{i-1}X_i$, thus $\EE_{\DF_1}V_{2}=q\EE V_1=qV_0$. Indeed, \[\EE_{\DF_1}V_{2}=\esup_{\zeta\in\TT_{2}}\EE_{\DF_1}Z_\zeta=\esup_{\zeta\in\TT_{2}}\EE Z_\zeta=q\sup_{\zeta\in\TT_{2}}\EE \left(q^{\zeta-2 }X_\zeta\right)=q\esup_{\zeta\in\TT_{1}}\EE \left(q^{\zeta-1 }X_{\zeta}\right)=q\EE V_1,\] where we used:
		\begin{itemize} [noitemsep]
			\item \Cref{snellfuture} in the first and last equality; 
			\item that any $\zeta\in\TT_2$ is independent of $\DF_1$ and $\{X_i\}$ are \textit{iid}, implying that $X_\zeta$ is also independent of $\DF_1$, in the second equality;
			\item that $\{X_i\}$ are \textit{iid} in the second to last equality, implying that any stopping rule that has not stopped by step $2$ and undergoes discount factors shifted backwards by one, has the same expected future reward as a stopping rule that has not stopped by step $1$ and undergoes unshifted discount factors. We call this property of geometric discounts \textit{time-invariance}.
		\end{itemize} 
		We have obtained the equation 
		\begin{equation}\label{geom}
			V_0=\EE (X\vee qV_0).
		\end{equation} 
		Denoting the cdf of $X$ as $F(x)$, it is straightforward to compute that
		\[\EE (X\vee qV_0)=qV_0F(qV_0)+\int_{qV_0}^bxdF(x)=b-\int_{qV_0}^bF(x)dx,\]
		where we used integration by parts in the second equality. By plugging this in \Cref{geom} we obtain the equation
		\begin{equation}\label{geom2}
			f(V_0)\defeq V_0+\int_{qV_0}^bF(x)dx=b.
		\end{equation}
		Note that, since $V_0$ is the expected return, $0\le V_0\le b$, and thus $f(V)$ is a strictly increasing differentiable function on $[0,b]$, since by the \textit{Fundamental Theorem of Calculus} $f'(V)=1-F(qV)>0$ as $qV<b$. Observe that, by using $F(x)\le1$, we have $f(0)\le b$, and by using the increasing monotonicity of $F(x)$, we have that $f(b)\ge b+F(qb)b(1-q)\ge b$. In conclusion by the \textit{Intermediate Value Theorem} a value $V_0$ satisfying \Cref{geom2} exists. By \Cref{snellstop,snellstoprule} and the time-invariance aforementioned, it is optimal to stop at the first value greater or equal than $V_0$.
	\end{proof}
	
	\subsection{Tightness of $2$-approximations}\label{suppGclasstight}
	The following argument appears in \cite[Theorem~3.5]{AliBanGolMunWan20}. We include it for ease of reference, and because of a slight difference in the more rigorous setup we established with \Cref{geomprice}.
        \begin{theorem}[{\cite{AliBanGolMunWan20}[Theorem~3.5]}]
		RH is $2-\sfrac{1}{\mu}$-hard.
	\end{theorem}
	\begin{proof}
		Fix a geometric horizon having failure probability $0<q<1$, denoted as $H\sim\Geom(1-q)$, and an instance for the values $X\sim\DD_p$, where $\DD_p$ denotes a two-point distribution defined as $p\defeq\PP(X=x_2)=1-\PP(X=x_1)$, with $x_2>x_1>0$ and, as we vary $p\in(0,1)$,
		\begin{equation}\label{relation}
			x_1\le\frac{qpx_2}{1-q(1-p)}.
		\end{equation}
		Then by \Cref{geomprice} it is optimal to adopt a single-threshold
		\begin{equation}\label{optprice}
			V_0\le\frac{px_2}{1-q(1-p)},
		\end{equation}
		coinciding with the expected return. The upper bound follows from the fact that since $x_1<V_0<x_2$, we need to explore only two cases:
		\begin{itemize}[noitemsep] 
			\item If $x_1<qV_0$, $\EE (X\vee qV_0)=px_2+(1-p)qV_0$. Plugging this into \Cref{geom} yields $V_0=px_2+(1-p)qV_0$. Solving for $V_0$, we obtain \[V_0 = \frac{px_2}{1-q(1-p)}.\]
			\item If $qV_0\leq x_1$, $\EE (X\vee qV_0)=\EE X=(1-p)x_1+px_2$. Plugging this into \Cref{geom} yields $V_0=(1-p)x_1+px_2$.
		\end{itemize}
		Note that by \Cref{relation}
		\[px_2+(1-p)x_1\le\left[p+\frac{(1-p)pq}{1-q(1-p)}\right]x_2= \frac{px_2}{1-q(1-p)},\]
		which implies \Cref{optprice}. 
		
		As we vary $p\in(0,1)$, for simplicity, require \Cref{relation} with equality, that is 
		\begin{equation}\label{relationeq}
			x_1=\frac{qpx_2}{1-q(1-p)}.
		\end{equation}
		By the law of total expectation we can compute straightforwardly that
		\begin{align*}
			\EE M_H&=(1-q)\sum_{i\in\NN}\EE M_i q^{i-1}=(1-q)\left[x_1\frac{1-p}{1-(1-p)q}+x_2\left(\frac{1}{1-q}-\frac{1-p}{1-(1-p)q}\right)\right]\\&=x_1\frac{(1-p)(1-q)}{1-(1-p)q}+x_2\frac{p}{1-(1-p)q}=x_1\left(\frac{(1-p)(1-q)}{1-(1-p)q}+\frac{1}{q}\right),
		\end{align*}
		since $\EE M_i = x_1(1-p)^i+x_2[1-(1-p)^i]$ and having used \Cref{relationeq} in the last equality.
		
		Finally, we consider the gambler-to-prophet ratio. For any $\tau\in\TT^*$, by \Cref{relationeq} again, we obtain
		\[\frac{\EE X_\tau}{\EE M_H}\leq\frac{x_2}{x_1}\frac{\frac{p}{1-q(1-p)}}{\frac{(1-p)(1-q)}{1-(1-p)q}+\frac{1}{q}}=\frac{1-q(1-p)}{qp}\frac{\frac{p}{1-q(1-p)}}{\frac{(1-p)(1-q)}{1-(1-p)q}+\frac{1}{q}}=\frac{1}{\frac{(1-p)(1-q)q}{1-(1-p)q}+1}\longrightarrow \frac{1}{q+1}=\frac{1}{2-\frac{1}{\mu}}\]
		as $p\longrightarrow 0$, since $1-q=\sfrac{1}{\mu}$. Thus no online algorithm can be more than $2-\sfrac{1}{\mu}$-competitive.
	\end{proof}
 
        \section{Proofs omitted from Section \ref{hardsecretary}}\label{supphardsecretary}
        
        \subsection{Proof of \Cref{asympmax,asympmax_H_m}}\label{supphardsecretarylem}
        \begin{proof}[Proofs of \Cref{asympmax}]
        Consider
        \begin{equation}\label{Emax}
		\EE M_n = \int_\RR xf_n(x)dx=\int_1^\infty xnV^{n-1}(x)V'(x)dx.
	\end{equation}
		Performing the substitution $u=V(x)$ into \Cref{Emax} and plugging in the expression $V^{-1}(u)\defeq (1-u)^{-\frac{1}{1+\eps}}$ yields
		\[\EE M_n =\int_0^{1}V^{-1}(u)nu^{n-1}du=n\Beta\left(n,1-\frac{1}{1+\eps}\right),\]
		where \[\Beta(\alpha,\beta)\defeq\int_0^1u^{\alpha-1}(1-u)^{\beta-1}du\] is the Beta function. The claim follows from the fact that 
		\begin{equation}\label{asympgamma}
			\Beta\left(n,1-\frac{1}{1+\eps}\right)=\frac{\Gamma(n)\Gamma\left(1-\frac{1}{1+\eps}\right)}{\Gamma\left(n+1-\frac{1}{1+\eps}\right)}\sim\Gamma\left(1-\frac{1}{1+\eps}\right) n^{-\frac{\eps}{1+\eps}}.
		\end{equation}
		This is shown by using the fact that the Gamma function \[\Gamma(x)\defeq\int_0^\infty t^{x-1}e^{-t}dt\sim(x-1)^{x-1}e^{-(x-1)}\sqrt{2\pi(x-1)}\] as $x\longrightarrow\infty$, by a well-known asymptotics, derived through the Laplace method. In particular, one applies the result for $x=n$ and $x=n+1-\frac{1}{1+\eps}$, and then simplifies the ratio of Gamma functions' asymptotics as $n\longrightarrow\infty$:
		\begin{align*}
			\frac{\Gamma(n)}{\Gamma\left(n+1-\frac{1}{1+\eps}\right)}&\sim\frac{(n-1)^{n-1}\sqrt{n-1}e^{n-\frac{1}{1+\eps}}}{\left(n-\frac{1}{1+\eps}\right)^{n-\frac{1}{1+\eps}}\sqrt{n-\frac{1}{1+\eps}}e^{n-1}}\sim e^{\frac{\eps}{1+\eps}} \frac{(n-1)^{n-1}}{\left(n-\frac{1}{1+\eps}\right)^{n-\frac{1}{1+\eps}}}\\&=\left(1-\frac{1-\frac{1}{1+\eps}}{n-\frac{1}{1+\eps}}\right)^{n-\frac{1}{1+\eps}}\frac{e^{\frac{\eps}{1+\eps}}}{(n-1)^{1-\frac{1}{1+\eps}}}\sim\left(1-\frac{\frac{\eps}{1+\eps}}{n-\frac{1}{1+\eps}}\right)^{n-\frac{1}{1+\eps}}\frac{e^{\frac{\eps}{1+\eps}}}{n^{\frac{\eps}{1+\eps}}}\sim n^{-\frac{\eps}{1+\eps}}.
		\end{align*}
	\end{proof}
 
    \begin{proof}[Proof of \Cref{asympmax_H_m}]
		Since the sequence $\left\lbrace h^{-\frac{1}{1+2\eps}}\right\rbrace_h$ is monotone decreasing, the summation defining $Z_m$ is asymptotically equivalent to the corresponding integral through the standard estimate \[\int_{\ell}^{m}x^{-\frac{1}{1+2\eps}}dx\le\sum_{h=\ell}^m h^{-\frac{1}{1+2\eps}}\le (\ell-1)^{-\frac{1}{1+2\eps}}+\int_{\ell}^{m}x^{-\frac{1}{1+2\eps}}dx.\] This yields that as $m\longrightarrow\infty$, \Cref{asympnorm} holds the extra term outside the integral is constant, while the following asymptotics of the integral holds:
		\[\int_{\ell}^mx^{-\frac{1}{1+2\eps}}dx=\frac{x^{1-\frac{1}{1+2\eps}}}{1-\frac{1}{1+2\eps}}\bigg\vert_{\ell}^m\sim\frac{1+2\eps}{2\eps} m^{\frac{2\eps}{1+2\eps}}.\]
		Consider that by the law of total expectation, \Cref{asympmax,asympnorm} and the increasing monotonicity of $\EE M_n$, as $m\longrightarrow\infty$, 
		\begin{align}\label{asympH_m}
			\EE M_{H_m}&=\sum_{n=\ell}^m\EE M_n p_n^\ssup{m}\sim\sum_{h=\ell}^{\lfloor \sqrt{m}\rfloor-1}\EE M_n p_n^\ssup{m}+ \frac{\Gamma\left(1-\frac{1}{1+\eps}\right)}{Z_m}\sum_{n=\lfloor \sqrt{m}\rfloor}^mn^{\frac{1}{1+\eps}-\frac{1}{1+2\eps}}\notag\\&\sim\frac{\Gamma\left(1-\frac{1}{1+\eps}\right)}{\frac{1+2\eps}{2\eps}} \frac{\sum_{n=\lfloor \sqrt{m}\rfloor}^mn^{\frac{\eps}{(1+\eps)(1+2\eps)}}}{ m^{\frac{2\eps}{1+2\eps}}}.
		\end{align}
		where the last asymptotic relation holds since we will now show that as $m\longrightarrow \infty$ the last series diverges with a higher order than \[\sum_{h=\ell}^{\lfloor \sqrt{m}\rfloor-1}\EE M_n p_n^\ssup{m}\le\EE M_{\lfloor \sqrt{m}\rfloor}\sim\sqrt{m}^{\frac{1}{1+\eps}}.\] 
		Indeed, since the sequence $\left\lbrace n^{\frac{\eps}{(1+\eps)(1+2\eps)}}\right\rbrace_n$ is monotone increasing, the last series of \Cref{asympH_m} is asymptotically equivalent to the corresponding integral through the standard estimate 
		\begin{align*}
			(\lfloor \sqrt{m}\rfloor-1)^{\frac{\eps}{(1+\eps)(1+2\eps)}}+\int_{\lfloor \sqrt{m}\rfloor}^{m}x^{\frac{\eps}{(1+\eps)(1+2\eps)}}dx&\le\sum_{n=\lfloor \sqrt{m}\rfloor}^mn^{\frac{\eps}{(1+\eps)(1+2\eps)}}\\&\le (m+1)^{\frac{\eps}{(1+\eps)(1+2\eps)}}+ \int_{\lfloor \sqrt{m}\rfloor}^{m}x^{\frac{\eps}{(1+\eps)(1+2\eps)}}dx.
		\end{align*} As $m\longrightarrow\infty$,
		\begin{equation}
			\int_{\lfloor \sqrt{m}\rfloor}^m x^{\frac{\eps}{(1+\eps)(1+2\eps)}}dx=\frac{x^{1+\frac{\eps}{(1+\eps)(1+2\eps)}}}{1+\frac{\eps}{(1+\eps)(1+2\eps)}}\bigg\vert_{\lfloor \sqrt{m}\rfloor}^m\sim\frac{m^{1+\frac{\eps}{(1+\eps)(1+2\eps)}}}{1+\frac{\eps}{(1+\eps)(1+2\eps)}} ,
		\end{equation}
		and the other terms in $\lfloor \sqrt{m}\rfloor-1$, $\lfloor \sqrt{m}\rfloor^{\frac{1}{1+\eps}}$ and $m+1$ are negligible, being of lower order. Plugging this asymptotics for the summation in the last term of \Cref{asympH_m} yields the result sought, since \[1+\frac{\eps}{(1+\eps)(1+2\eps)}-\frac{2\eps}{1+2\eps}=\frac{1}{1+\eps}.\]
	\end{proof}
	
	\subsection{Proof  of \Cref{hardsingle}}\label{supphardsecretaryprop}
    \begin{proof}[Proof of \Cref{hardsingle}]
		Our strategy will be to show the boundedness of the sequence of optimal single thresholds $\{\bar{\pi}_m\}$ (corresponding to horizons $\{H_m\}$ for $m$ large enough and $\eps$ small enough), which maximize the sequence $\EE X_{\tau_{\pi_m}}\defeq\EE Y_{\tau_{\pi_m}}^\ssup{H_m}$, where $\{\pi_m\}$ is an arbitrary sequence of single-thresholds and $Y_{i}^\ssup{H_m}\defeq X_i\II_{\{H_m\ge i\}}$. 
        Then, we show that there is a subsequence $\{m_j\}$, such that as $j\longrightarrow\infty$,
		\begin{equation}\label{subsequence}
			r_{m_j}\defeq\frac{\EE Y_{\tau_{\bar{\pi}_{m_j}}}^\ssup{H_{m_j}}}{\EE M_{H_{m_j}}}\longrightarrow 0,
		\end{equation} 
		thus yielding the claim, since this ratio is an upper bound on any ratio defined for arbitrary single-thresholds. 
		
		Since $X$ is a continuous random variable, we can apply \Cref{algp} to compute $\EE Y_{\tau_\pi}^\ssup{H_m}$, which is nontrivial only if $\pi>1$. Consider first that by direct computation,
		\begin{equation}\label{conditional}
			\EE(X|X\ge\pi)=\pi+	\EE(X-\pi|X\ge\pi)=\pi+\frac{1}{\overline{V}(\pi)}\int_{\pi}^\infty\overline{V}(x)dx=\pi+\pi^{1+\eps}\int_{\pi}^\infty\frac{1}{x^{1+\eps}}dx=\left(1+\frac{1}{\eps}\right)\pi
		\end{equation}
		and
		\begin{equation}\label{c_pi}
			c_\pi=1-\EE V^{H_m}(\pi)=1-\EE\left(1-\frac{1}{\pi^{1+\eps}}\right)^{H_m}.
		\end{equation}
		For every fixed $m$ large enough, consider that by \Cref{conditional,c_pi} the function \[f_m(\pi)\defeq \EE Y_{\tau_\pi}^{\ssup{H_m}}=c_\pi \EE(X|X\ge\pi)= \left(1+\frac{1}{\eps}\right)\pi\left[1-\EE\left(1-\frac{1}{\pi^{1+\eps}}\right)^{H_m}\right],\] and compute
		\[f'_m(\pi)=\left(1+\frac{1}{\eps}\right)\left[1-\EE\left(1-\frac{1}{\pi^{1+\eps}}\right)^{H_m}-\frac{1+\eps}{\pi^{1+\eps}}\EE H_m\left(1-\frac{1}{\pi^{1+\eps}}\right)^{H_m-1}\right].\]
		Note that for any fixed $m$ and $\eps>0$, assuming that $\ell>1$, we have the following: 
		\begin{itemize}[noitemsep]
    		\item $f_m(\pi)$ increases when $\pi\approx1$, since $f_m'(1)=1+\sfrac{1}{\eps}>0$.
    		\item $f_m(\pi)$ decreases when $\pi$ is large enough, since, by Bernoulli's inequality \[\left(1-\frac{1}{\pi^{1+\eps}}\right)^h\ge 1-\frac{h}{\pi^{1+\eps}},\] we have that, as $\pi\longrightarrow\infty$,
    		\begin{align*}
    			f'_m(\pi)&\le\left(1+\frac{1}{\eps}\right)\left\lbrace\frac{\EE H_m}{\pi^{1+\eps}}-\frac{1+\eps}{\pi^{1+\eps}}\left(\EE H_m-\frac{\EE H_m(H_m-1)}{\pi^{1+\eps}}\right)\right\rbrace\\&=\frac{1}{\pi^{1+\eps}}\left(1+\frac{1}{\eps}\right)\left(-\eps\EE H_m+\frac{1+\eps}{\pi^{1+\eps}}\EE H_m(H_m-1)\right),
    		\end{align*}
    		with the last term being eventually negative as $\pi$ grows large. 
            Thus, there exists $\pi^*=\pi_m^*(\eps)$ such that $f_m'(\pi)<0$ for every $\pi>\pi^*$.
		\end{itemize}
		The fact that the function $f_m(\pi)$ is strictly increasing at the left end of the interval $[1,\infty)$ and strictly decreasing on the right end, that is on $(\pi^*,\infty)$, means that, by continuity, that is by the \textit{Intermediate Value Theorem} applied to the compact interval $[1,\pi^*]$, a global maximum exists and is attained at $\bar{\pi}_m=\bar{\pi}_m(\eps)\in(1,\pi^*]$. Furthermore, it is characterised by $f_m'(\bar{\pi}_m)=0$, which we rewrite as
		\begin{equation}\label{maxeq0}
			g_m(\bar{\pi}_m)\defeq\EE\left(1-\frac{1}{\bar{\pi}_m^{1+\eps}}\right)^{H_m}+\frac{1+\eps}{\bar{\pi}_m^{1+\eps}}\EE H_m\left(1-\frac{1}{\bar{\pi}_m^{1+\eps}}\right)^{H_m-1}=1.
		\end{equation}
		To conclude, consider that by \Cref{subsequence,conditional,c_pi,asympmax_H_m,asympmaxunif_H_m} we have that, as $m\longrightarrow\infty$, uniformly in $\eps>0$ small enough,
		\begin{equation}\label{ratio}
			r_m=c_{\bar{\pi}_m} \frac{\EE(X|X\ge\bar{\pi}_m)}{\EE M_{H_m}}\sim\left(1+\frac{1}{\eps}\right)\frac{c_{\bar{\pi}_m} \bar{\pi}_m}{Nm^{\frac{1}{1+\eps}}}\asymp \frac{c_{\bar{\pi}_m} \bar{\pi}_m}{\eps m^{\frac{1}{1+\eps}}}.
		\end{equation}
		
        Consider that $\{\bar{\pi}_m\}$ is either bounded or unbounded. Then by considering the following cases, we can always conclude that the optimal sequence of the gambler-to-prophet ratios, for fixed threshold strategies, admits a vanishing $\{r_{m_j}\}$, proving \Cref{subsequence}.
		\paragraph{Bounded and slowly increasing.} Assume that $\bar{\pi}_m = \smallO\left(m^{\frac{1}{1+\eps}}\right)$. 
            Since $c_{\bar{\pi}_{m}}<1$, by \Cref{ratio}, we have that $r_m$ trivially vanishes as $m\longrightarrow\infty$, for any fixed $\eps$ small enough. Note that this case covers also the trivial case for which there exists a subsequence $\{\bar{\pi}_{m_j}\}$ such that $\bar{\pi}_{m_j}\le1$, for which as $j\longrightarrow\infty$,
            \[
                r_{m_j}=\frac{\EE X}{\EE M_{H_{m_j}}}\asymp \frac{1}{\eps m_j^{\frac{1}{1+\eps}}} \longrightarrow 0.
            \]
          \paragraph{Unbounded and fast increasing.}
            Assume that $m^{\frac{1}{1+\eps}}=\smallO(\bar{\pi}_m)$. Note that, by Bernoulli's inequality, 
            \[
                \EE\left(1-\frac{1}{\bar{\pi}_{m}^{1+\eps}}\right)^{H_{m}}\ge1-\frac{\EE H_{m}}{\bar{\pi}_{m}^{1+\eps}},
            \]
    		yielding that \[c_{\bar{\pi}_{m}}\defeq1-\EE\left(1-\frac{1}{\bar{\pi}_{m}^{1+\eps}}\right)^{H_{m}}\le\frac{\EE H_{m}}{\bar{\pi}_{m}^{1+\eps}}.\]
    		In conclusion, by \Cref{expectationn} with $n=1$ it holds that as $m\longrightarrow\infty$, uniformly in $\eps$ small enough, as $m\longrightarrow\infty$,
    		\[
                \frac{c_{\bar{\pi}_{m}} \bar{\pi}_{m}}{ m^{\frac{1}{1+\eps}}}\le \frac{\EE H_m}{ m^{\frac{1}{1+\eps}}\bar{\pi}_m^\eps}
                    \asymp \varepsilon\left(\frac{m^{\frac{1}{1+\eps}}}{\bar{\pi}_m}\right)^{\eps}
                    \longrightarrow 0.
            \]
    		Therefore, by \Cref{ratio}, we have that $r_m$ vanishes.
      
    		\paragraph{Unbounded (and moderately increasing).} 
            By relying on the result above for unbounded and fast increasing behaviour, we show that if $\{\bar{\pi}_m\}$ is assumed unbounded, with no additional properties, $r_m$ vanishes. This is an argument by contradiction and will essentially boil down to analysing, hypothetically, the moderately increasing regime, that is the one in which $\bar{\pi}_m \asymp m^{\frac{1}{1+\eps}}$. We will show that it is not possible, and draw our conclusion from there.
            
            Recall that $\{ r_m \} \subseteq [0, 1]$ is a bounded sequence. Therefore, to show that $\{r_m\}$ vanishes as long as we choose $\eps$ small enough, it is sufficient to show that, under the same assumptions, all of its convergent subsequences $\{ r_{m_j} \}$ vanish. 
            By contradiction, assume that there is a convergent subsequence $\{ r_{m_j} \}$ that does not vanish, no matter how small $\eps > 0$ is taken, that is, as $j \longrightarrow \infty$, we have that $r_{m_j}\longrightarrow\rho=\rho(\eps) \in (0,1]$, with $\rho(\eps)$ being bounded away from $0$ no matter how small $\eps > 0$ is taken. 
            Then, by \Cref{ratio}, we have that as $j \longrightarrow \infty$, for $\eps > 0$ fixed small enough,
    		\begin{equation}\label{alphabound}
    			\mu_j\defeq \frac{m_j}{\bar{\pi}_{m_j}^{1+\eps}}\sim\left[\left(1+\frac{1}{\eps}\right)\frac{ c_{\bar{\pi}_{m_j}}}{\rho N}\right]^{1+\eps}\le\left(\frac{1+\eps }{\eps \rho N}\right)^{1+\eps}.
    		\end{equation}
    		Since $\{\mu_j\}$ is asymptotically equivalent to a bounded sequence, it must be bounded too. 
            By boundedness, we can consider a convergent subsequence such that for some fixed $\alpha=\alpha(\eps)\ge0$ satisfying the bound in \Cref{alphabound}, as $k\longrightarrow\infty$, $\mu_{j_k}\longrightarrow\alpha. $

            Consider the function $\eps \mapsto \alpha(\eps)$. 
            We show that, unless $\alpha$ is eventually (that is, for all $\eps$ small enough) identically $0$, $\liminf_{\eps \to 0} \alpha(\eps) > 0$, that is, it is bounded away from zero uniformly for all $\eps$ small enough. 
            Indeed, by contradiction, assume that $\alpha$ is not identically $0$ for all $\eps$ small enough, and that $\liminf_{\eps \to 0} \alpha(\eps) = 0$. 
            Recall that, as $k\longrightarrow\infty$,
    		\begin{equation}\label{asymp_pi}
    			\frac{1}{\bar{\pi}_{m_{j_k}}^{1+\eps}}\sim\frac{\alpha}{m_{j_k}}\longrightarrow 0.
    		\end{equation} 
    		By \Cref{maxeq0,asymp_pi} and \Cref{expectationn} with $n=1$ 
    		\begin{align*}
    			c_{\bar{\pi}_{m_{j_k}}}&=\frac{1+\eps}{\bar{\pi}_{m_{j_k}}}\EE H_{m_{j_k}}\left(1-\frac{1}{\bar{\pi}_{m_{j_k}}^{1+\eps}}\right)^{H_{m_{j_k}}-1}\ge \frac{1+\eps}{\bar{\pi}_{m_{j_k}}}\EE H_{m_{j_k}}\left(1-\frac{1}{\bar{\pi}_{m_{j_k}}^{1+\eps}}\right)^{m_{j_k}}\\&\sim (1+\eps) \alpha\frac{\EE H_{m_{j_k}}}{m_{j_k}}\left(1-\frac{\alpha}{m_{j_k}}\right)^{m_{j_k}}\sim\frac{1+\eps}{1+4\eps}2\eps\alpha e^{-\alpha}
    		\end{align*}
    		On the other hand, by definition, Jensen's inequality \Cref{asymp_pi} and \Cref{expectationn} with $n=1$
    		\begin{align*}
    			c_{\bar{\pi}_{m_{j_k}}}&\defeq 1-\EE \left(1-\frac{1}{\bar{\pi}_{m_{j_k}}^{1+\eps}}\right)^{H_{m_{j_k}}}\le 1- \left(1-\frac{1}{\bar{\pi}_{m_{j_k}}^{1+\eps}}\right)^{\EE H_{m_{j_k}}}\\&\sim 1-\left(1-\frac{\alpha}{m_{j_k}}\right)^{\frac{2\eps m_{j_k}}{1+4\eps}}\sim1-e^{-\frac{2\alpha\eps}{1+4\eps}}.
    		\end{align*}
            Then, combining both estimates, we have that, taking the limit as $k\longrightarrow\infty$,
    		\begin{equation}\label{boundedalpha}
    			1-e^{-\frac{2\alpha\eps}{1+4\eps}}
                    \ge\frac{1+\eps}{1+4\eps}2\eps\alpha e^{-\alpha}.
    		\end{equation}
    		Since $\liminf_{\eps \to 0} \alpha(\eps) = 0$, we have that, for infinitely many $\eps$ arbitrarily small,\footnote{In what follows all asymptotic notation is understood as $\eps$ vanishes along a subsequence. We refrain from using a sequence $\eps_l$ for mere simplicity of notation, to avoid working with multiple indices ($k$ and $l$), which might result confusing.}
    		\[
                1-e^{-\frac{2\alpha\eps}{1+4\eps}}
                    =\frac{2\alpha\eps}{1+4\eps}+\bigO(\alpha^2\eps^2),
            \]
    		and
    		\[\frac{1+\eps}{1+4\eps}2\eps\alpha e^{-\alpha}=\frac{1+\eps}{1+4\eps}2\eps\alpha +\bigO(\eps\alpha^2).\]
            Replacing back in \Cref{boundedalpha}, we obtain
    		\[
                \frac{2\alpha\eps}{1+4\eps}+\bigO(\alpha^2\eps^2)\ge\frac{1+\eps}{1+4\eps}2\eps\alpha +\bigO(\eps\alpha^2),
            \]
    		which is equivalent to 
    		\[\frac{1}{1+4\eps}+\bigO(\alpha\eps)\ge\frac{1+\eps}{1+4\eps}+\bigO(\alpha),\]
    		which is trivially not possible as $\eps$ vanishes, regardless of the rate at which $\alpha$ vanishes.
            We conclude that $\liminf_{\eps \to 0} \alpha(\eps) > 0$.
    		
            Recall also that the optimality condition \Cref{maxeq0} must be satisfied, which we recast as
    		\begin{equation}\label{maxeq0'}
    			g_{m_{j_k}}(\bar{\pi}_{m_{j_k}})
                    \defeq \EE \left[ \phi(\bar{\pi}_{m_{j_k}},H_{m_{j_k}})\right]
                    = 1 \,,
    		\end{equation}
    		where \[\phi(\pi,H)\defeq \left(1-\frac{1}{\pi^{1+\eps}}\right)^{H-1}\left[1-\frac{1}{\pi^{1+\eps}}+(1+\eps)\frac{H}{\pi^{1+\eps}}\right].\]
    		To reach a contradiction, we derive the following asymptotic upper bound, holding uniformly in $\alpha>0$ and as $k\longrightarrow\infty$,
    		\begin{equation}\label{ebound}
    			g_{m_{j_k}}(\bar{\pi}_{m_{j_k}})
                    \le 1 + \eps \left[(2+\alpha) e^{-\alpha}-2\right] + \bigO(\eps^2)+ \bigO \left(\frac{1}{m_{j_k}}\right).
    		\end{equation}
    		We first clinch the argument and then show how to derive this asymptotic upper bound. 
    		\begin{itemize}[noitemsep]
    		\item Assume that it is the case that $\liminf_{\eps \to 0} \alpha(\eps) > 0$. Then for all $\eps>0$ small enough, \Cref{ebound} yields that $g_{m_{j_k}}(\bar{\pi}_{m_{j_k}})<1$, in contradiction with \Cref{maxeq0'}. Indeed, under the assumption, the strictly decreasing function $h(\alpha)\defeq (2+\alpha)e^{-\alpha}-2$ is negative and bounded away from zero for all $\eps>0$ small enough, since in this range of values of $\eps$, $h(\alpha)$ vanishes only as $\alpha$ vanishes (see \Cref{h}). 
    		\begin{figure}\centering
    			\begin{tikzpicture}[scale = .7]
    				\begin{axis}[axis lines = middle, xlabel={$\alpha$}, x label style={anchor=west}]
    					\addplot[domain=0:1, samples=100, color=blue]{(2+x)*exp(-x)-2};
    				\end{axis}
    			\end{tikzpicture}
    			\caption{The function $h(\alpha)$.}
    			\label{h}
    		\end{figure}
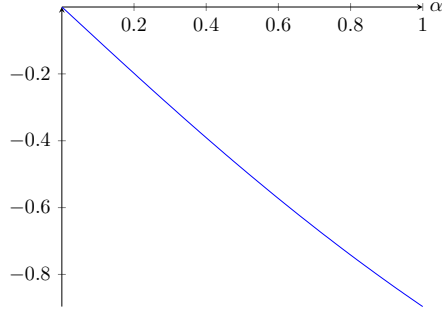
    		\item As a result of the contradiction reached, $\alpha$ must be identically $0$ for all $\eps>0$ small enough, meaning that for all such $\eps$, $\mu_{j_k} \longrightarrow 0$ as $k\longrightarrow\infty$. By boundedness, this ensures that $\mu_j\longrightarrow0$ as $j\longrightarrow\infty$. That is, we are, subsequentially, in the fast increasing scenario of the previous point, which shows (upon replacing $m$ with $m_j$)that $r_{m_j}$ vanishes. This is in contradiction with the initial assumption, that $r_{m_j}\longrightarrow \rho>0$, and therefore we have shown that $\rho=0$.
    		\end{itemize}

            \paragraph{Proof of \Cref{ebound}.} We are only left with proving the bound \Cref{ebound}. To show \Cref{ebound}, we compute the series expansion of $\phi(\bar{\pi}_{m_{j_k}},H_{m_{j_k}})$ with respect to the first argument. 
            By \Cref{asymp_pi} and the fact that $H_{m_{j_k}}\le m_{j_k}$, we have that, as $k\longrightarrow\infty$,
    		\begin{align*}
    			\left(1-\frac{1}{\bar{\pi}_{m_{j_k}}^{1+\eps}}\right)^{H_{m_{j_k}}-1}  &=  \exp \left( \left( H_{m_{j_k}} - 1 \right) \, \log \left( 1 - \frac{1}{\bar{\pi}_{m_{j_k}}^{1+\eps}}\right) \right)
                     \sim \exp \left( -\frac{H_{m_{j_k}}}{\bar{\pi}_{m_{j_k}}^{1+\eps}}\left(1+\bigO\left(\frac{1}{\bar{\pi}_{m_{j_k}}^{1+\eps}} \right) \right) \right) \\
                    & \sim \; \exp \left( - \alpha \frac{H_{m_{j_k}}}{m_{j_k}} \left( 1 + \bigO\left(\frac{1}{m_{j_k}} \right)\right) \right)
                    \\& = \sum_{n\ge0}\frac{(-1)^n}{n!}\alpha^n\left(\frac{H_{m_{j_k}}}{m_{j_k}}\right)^n\left(1+\bigO\left(\frac{1}{m_{j_k}}\right)\right)^n.
    		\end{align*}
    		Thus
    		\begin{align*}
    			\phi(\bar{\pi}_{m_{j_k}},H_{m_{j_k}})&\sim \sum_{n\ge0}\frac{(-1)^n}{n!}\alpha^n\left(\frac{H_{m_{j_k}}}{m_{j_k}}\right)^n\left(1+n\bigO\left(\frac{1}{m_{j_k}}\right)\right)\left[1-\frac{\alpha}{m_{j_k}}+(1+\eps)\alpha\frac{H_{m_{j_k}}}{m_{j_k}}\right]\\&= \sum_{n\ge0}\frac{(-1)^n}{n!}\alpha^n\left(\frac{H_{m_{j_k}}}{m_{j_k}}\right)^n\left[1+(1+\eps)\alpha\frac{H_{m_{j_k}}}{m_{j_k}}\right]+\bigO\left(\frac{1}{m_{j_k}}\right)
    		\end{align*}
    		since both series 
            \[
                \sum_{n\ge0}\frac{(-1)^n}{n!}\alpha^n\left(\frac{H_{m_{j_k}}}{m_{j_k}}\right)^n
                ,\quad
                \sum_{n\ge0}\frac{(-1)^n}{n!}n\alpha^n\left(\frac{H_{m_{j_k}}}{m_{j_k}}\right)^n
            \]
    		converge uniformly on the probability space. 
            Due to the uniformity in the probability space holding for all asymptotic estimates so far performed, we can finally take expectation and obtain, by \textit{Fubini Theorem} and \Cref{expectationn}, that
    		\begin{align*}
                g_{m_{j_k}}(\bar{\pi}_{m_{j_k}})
                    &\sim \sum_{n\ge0}\frac{(-1)^n}{n!}\alpha^n\frac{\EE H_{m_{j_k}}^n}{m_{j_k}^n}+(1+\eps)\alpha\sum_{n\ge0}\frac{(-1)^n}{n!}\alpha^n\frac{\EE H_{m_{j_k}}^{n+1}}{m_{j_k}^{n+1}}+\bigO\left(\frac{1}{m_{j_k}}\right)\\
                    &\sim  \sum_{n\ge0}\frac{(-1)^n}{n!}\alpha^n\frac{2\eps}{n+2(n+1)\eps}+(1+\eps)\alpha\sum_{n\ge0}\frac{(-1)^n}{n!}\alpha^n\frac{2\eps}{n+1+2(n+2)\eps}+\bigO\left(\frac{1}{m_{j_k}}\right)\\
                    &< 1 + 2\eps\sum_{n\ge1}\frac{(-1)^n}{n!}\alpha^n+\eps(1+\eps)\alpha\sum_{n\ge0}\frac{(-1)^n}{n!}\alpha^n+\bigO\left(\frac{1}{m_{j_k}}\right)\\
                    &=1 + \eps\left[2(e^{-\alpha}-1)+(1+\eps)\alpha e^{-\alpha}\right] + \bigO\left(\frac{1}{m_{j_k}}\right) \,.
    		\end{align*}
    		Recall that globally on the nonnegative reals, $\alpha e^{-\alpha}\le\sfrac{1}{e}$, and therefore as $\eps\longrightarrow 0^+$,
    		\[
                \eps\left[2(e^{-\alpha}-1)+(1+\eps)\alpha e^{-\alpha}\right] 
                    \le \eps\left[(2+\alpha)e^{-\alpha}-2\right]+\bigO(\eps^2) 
            \]
    		which implies \Cref{ebound}.
    \end{proof}
    
    \subsection{Heuristics for the horizons in the hard instance of \Cref{hardsingle}}\label{supphardsecretaryheur}
    We now comment heuristically regarding why single-threshold approximations should be expected to fail for this instance, perhaps even in higher capacity than what proven in \Cref{hardsingle} (which restricts the instance to the value distribution $X$, in order to show hardness). Note that by exploiting \Cref{expectationn} with $n=1,2$ we have that as $m\longrightarrow\infty$,
	\[\frac{\Var H_m}{\EE^2 H_m}=\frac{\EE H_m^2}{\EE^2 H_m}-1\longrightarrow\frac{(1+4\eps)^2}{4\eps(1+3\eps)}-1.\] Since the ratio in the limit diverges as $\eps$ vanishes, this provides an intuitive reason why the single-threshold $2$-approximation of \Cref{L2} should fail. In this sense, the hypothesis that $\eps$ is small is essential. It is straightforward to compute the derivative of the ratio in the limit and show that it is negative for any $\eps$. Thus this limit decreases as $\eps$ grows and tends to $\sfrac{1}{3}<0.594\approx1+W_0(-2e^{-2})$. This means that for $\eps$ large enough the single-threshold $2$-approximation of \Cref{L2} is still valid, as the distribution $H_m$ has enough concentration in this case. The horizon $H_m$ in \Cref{pgfzoom} instead has $\text{CV}\approx1.074$ due to $\eps$ being sufficiently small and $m$ large enough. In this case not only the single-threshold $2$-approximation of \Cref{L2} is not valid, but no analogous single-threshold $C$-approximation is possible, as $m\longrightarrow\infty$, since no matter how large $C$ is taken, since $1+W_0(-Ce^{-C})\longrightarrow 1<1.074$ from below, as $C\longrightarrow\infty$. Furthermore, plotting the pgf of $H_m$, for $\eps$ small and $m$ large, against the pgf of a geometric horizon with the same mean as in \Cref{pgfzoom}, shows that there is no ordering of the two (due to the crossing point of the graphs), and therefore this family of horizons is outside the $\PGF$ (and $PGFd$) class, which explains heuristically why the $2$-approximation of \Cref{Gclass} fails too.
	\begin{figure}\centering
		\includegraphics[scale=0.6]{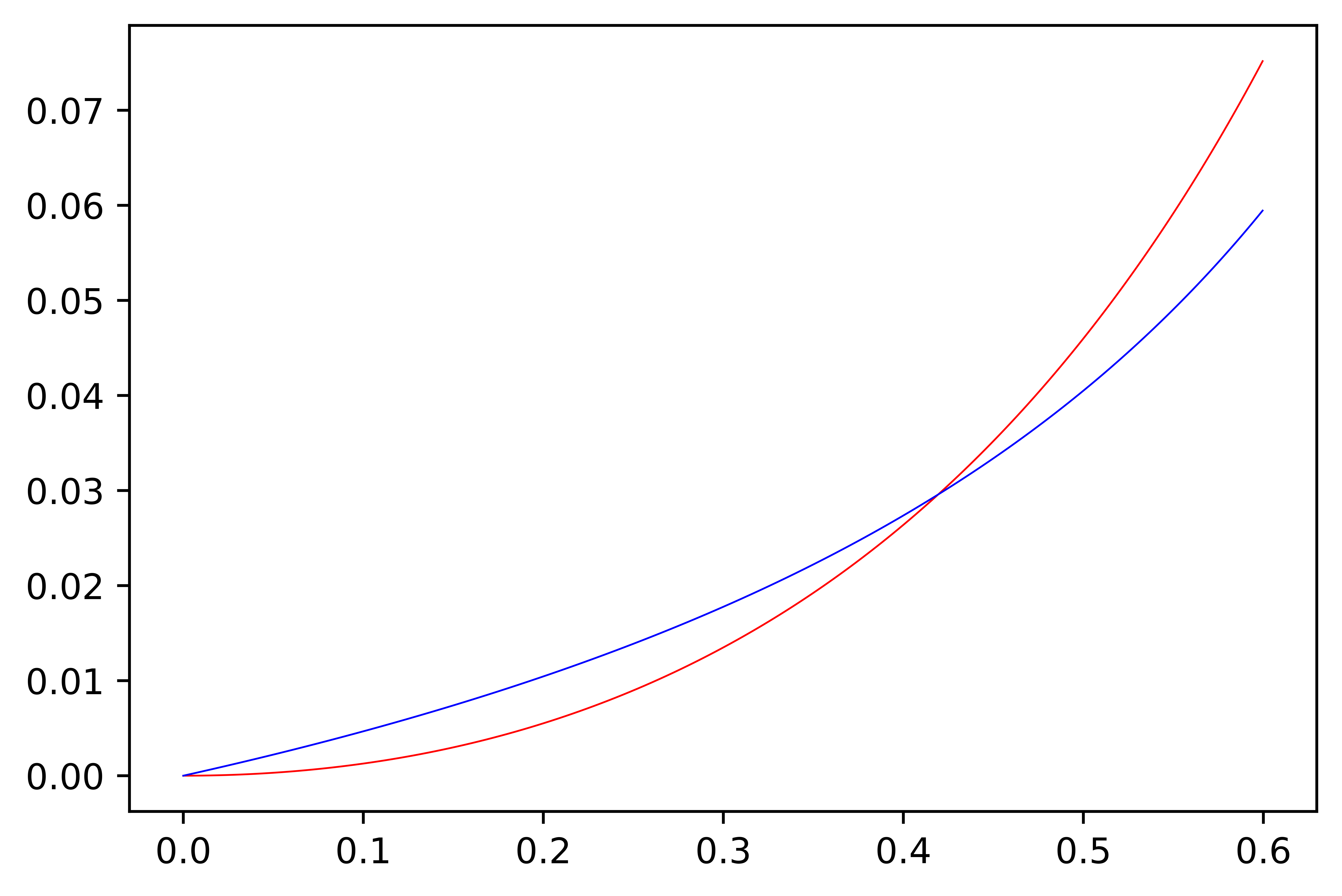}
		\caption{Zoom on $(0,0.6)$ of the pgf of $H_{100}$, with $\ell=2$, $\eps=0.001$ (red) against the pgf of a geometric horizon with the same mean (blue).}
		\label{pgfzoom}
	\end{figure}

	\subsection{Proof of \Cref{Gdhorizon}}\label{supphardsecretaryGdhorizon}
	\begin{proof}[Proof of \Cref{Gdhorizon}]
		~\paragraph{Step 1.} We start by constructing a family $\tilde{\mathcal{H}}_M(\eps)\subset\PGFd$, which is a sequence of horizons $\tilde{H}_m$ arising as a perturbation of horizons $H_m$ in the hard family $\mathcal{H}_M(\eps)$ of \Cref{hardhorizon}, with fixed $\ell=2$, as $m$ grows, for all fixed $0<\eps<\sfrac{1}{4}$ small enough. This is done by adding to the pmf of said $H_m$, a (to be suitably rescaled) mass at one, 
		\begin{equation}\label{mass}
			\delta_m\defeq Cm^{\frac{2\eps}{1+2\eps}},\quad C=C(\eps)\defeq\frac{3+10\eps}{2\eps}.
		\end{equation} 
		We denote the perturbed horizon $\tilde{H}_m$ for every $m,\eps$ considered, and observe that the corresponding normalizing constants for the pmf's $\tilde{Z}_m$ and $Z_m$; and expectations $\tilde{\mu}_m\defeq\EE\tilde{H}_m$ and $\mu_m\defeq\EE H_m$; satisfy, as $m$ grows, the following:
		\begin{align}
			\tilde{Z}_m&\sim \tilde{C}Z_m,\quad \tilde{C}\defeq C+1+\frac{1}{2\eps}\label{Zratio}\\ 
			\tilde{\mu}_m&\sim\frac{1}{\tilde{C}}\left(C+\frac{1+2\eps}{1+4\eps}m\right).\label{muratio}
		\end{align} 
		Trivially, \Cref{Zratio,muratio} follow from \Cref{asympnorm,expectationn} and the definitions, respectively $\tilde{Z}_m\defeq Z_m+\delta_m$ and $\tilde{\mu}_m = \sfrac{\delta_m}{\tilde{Z}_m}+\mu_m\sfrac{Z_m}{\tilde{Z}_m}$.
		To show that for all $m$ large enough, $\tilde{H}_m\in\PGFd$, we have to establish that eventually (short for all $m$ large enough from now on), for all $t\in(0,1)$ (the implicit range from now on),
		\[\frac{1}{\tilde{Z}_m}\left[\delta_mt+\sum_{h=2}^m t^h h^{-\frac{1}{1+2\eps}}\right]\ge\frac{1}{\tilde{\mu}_m}\sum_{h=1}^\infty\left(1-\sfrac{1}{\tilde{\mu}_m}\right)^{h-1}t^h,\]
		which we recast as
		\begin{equation}\label{target}
			\sum_{h=2}^m t^h h^{-\frac{1}{1+2\eps}}\ge -\delta_mt+\frac{\eps_mt}{1-t(1-\sfrac{1}{\tilde{\mu}_m})},
		\end{equation}
		having set $\eps_m\defeq \sfrac{\tilde{Z}_m}{\tilde{\mu}_m}$. Let 
		\begin{align*}\label{eta}
			\eta_m\defeq \frac{1-\sfrac{\eps_m}{\delta_m}}{1-\sfrac{1}{\tilde{\mu}_m}}.
		\end{align*}
		We establish that eventually for all $t\in(0,\eta_m]$
		\begin{equation}\label{linear1}
			-\delta_mt+\frac{\eps_mt}{1-t(1-\sfrac{1}{\tilde{\mu}_m})}\le0,
		\end{equation}
		and that eventually for all $t\in[\eta_m,1)$
		\begin{equation}\label{concave1}
			\sum_{h=2}^m h(h-1)t^{h-2} h^{-\frac{1}{1+2\eps}}<\frac{2\eps_m(1-\sfrac{1}{\tilde{\mu}_m})}{[1-t(1-\sfrac{1}{\tilde{\mu}_m})]^3}.
		\end{equation}
		This implies \Cref{target} for all $t\in (0,1)$, since eventually on $(0,\eta_m]$ \Cref{target} holds with strict inequality directly by \Cref{linear1}, while for all $t\in[\eta_m,1)$ we have that the difference of the left-hand and right-hand side of \Cref{target}, denoted $f_m(t)$, is strictly concave (\Cref{concave1} states precisely that $f''_m(t)<0$). Since $f_m(t)>0$ for all $t\in(0,\eta_m]$ and since it vanishes at both ends of the unit interval, it cannot vanish on $(\eta_m,1)$, meaning that there cannot be any crossings of the two sides of \Cref{target} on $(\eta_m,1)$, which therefore holds everywhere. The rest of this step is dedicated to show \Cref{linear1,concave1}.
		
		In order to establish \Cref{linear1}, it is enough to note that, upon simplifying the factor $t$, we can recast it as
		\begin{equation}\label{linear2}
			t\le\frac{\delta_m-\eps_m}{\delta_m(1-\sfrac{1}{\tilde{\mu}}_m)}=\eta_m.
		\end{equation}
		Note that
		\[\eps_m\sim \frac{\tilde{C}^2m^{\frac{2\eps}{1+2\eps}}}{C+\frac{1+2\eps}{1+4\eps}m},\]
		and therefore we have that
		\[
			\eta_m>1-\frac{\eps_m}{\delta_m}=1-\frac{\tilde{C}^2}{C\left(C+\frac{1+2\eps}{1+4\eps}m\right)}\longrightarrow 1.
		\]
		
		In order to establish \Cref{concave1}, note that a standard integral bound yields that
		\[\sum_{h=2}^m h(h-1)t^{h-2} h^{-\frac{1}{1+2\eps}}<\sum_{h=2}^m h^{2-\frac{1}{1+2\eps}}\le\int_2^{m+1}x^{2-\frac{1}{1+2\eps}}dx<\frac{(m+1)^{3-\frac{1}{1+2\eps}}}{3-\frac{1}{1+2\eps}}\]
		on the unit interval considered, and that on $[\eta_m,1)$
		\[\frac{2\eps_m(1-\sfrac{1}{\tilde{\mu}_m})}{[1-t(1-\sfrac{1}{\tilde{\mu}_m})]^3}\ge \frac{2\eps_m(1-\sfrac{1}{\tilde{\mu}_m})}{[1-\eta_m(1-\sfrac{1}{\tilde{\mu}_m})]^3}=2\frac{\delta_m^3}{\eps_m^2}\left(1-\frac{1}{\tilde{\mu}_m}\right)\sim\frac{2C^3\frac{1+2\eps}{1+4\eps}}{\tilde{C}^2}m^{2+\frac{2\eps}{1+2\eps}},\]
		and therefore it is sufficient to establish that eventually on $[\eta_m,1)$ we have that
		\begin{equation}\label{concave2}
			\frac{(m+1)^{3-\frac{1}{1+2\eps}}}{3-\frac{1}{1+2\eps}}< \frac{2C^3\frac{1+2\eps}{1+4\eps}}{\tilde{C}^2}m^{2+\frac{2\eps}{1+2\eps}}.
		\end{equation}
		This follows from noting that the powers of the $m$ terms are exactly the same and that the constants of the leading term of the left-hand side is strictly less than that of the right-hand side, that is
		\[\frac{1}{3-\frac{1}{1+2\eps}}<\frac{2C^3\frac{1+2\eps}{1+4\eps}}{\tilde{C}^2},\]
		by the very definition of $C$ and $\eps<\sfrac{1}{4}$.
		
		\paragraph{Step 2.} Next, we show that with $X$ being set to be the hard instance of values $X$ of \Cref{hardval}, with $\eps>0$ fixed small enough as in \emph{Step 1}, $\EE M_{\tilde{H}_m}=\Omega (\EE M_{H_m})$ and $\EE X_{\tau_{\pi_m}}\defeq \EE Y_{\tau_{\pi_m}}^\ssup{\tilde{H}_m}\le \EE Y_{\tau_{\pi_m}}^\ssup{H_m}$ as $m$ grows, for every sequence of thresholds $\{\pi_m\}$. 
		
		The first fact follows from \Cref{asympmaxunif_H_m,Zratio} and observing that the law of total expectation yields
		\[\EE M_{\tilde{H}_m}=\frac{\delta_m}{\tilde{Z}_m} \EE X +\frac{Z_m}{\tilde{Z}_m}\EE M_{H_m}>\frac{Z_m}{\tilde{Z}_m}\EE M_{H_m}\sim\frac{1}{\tilde{C}}\EE M_{H_m}.\]
		
		The second fact follows from considering the survival functions of $\tilde{H}_m$ and $H_m$, respectively $\tilde{S}_m(i)$ and $S_m(i)$ and applying \Cref{optalg}. We have that
		\begin{equation}\label{survivals}
		\begin{aligned}
			S_m(1)&=1, &&S_m(2)=1, &&S_m(i)=\frac{1}{Z_m}\sum_{h= i}^mh^{-\frac{1}{1+2\eps}}\:\forall i\ge 3\\
			\tilde{S}_m(1)&=1, &&\tilde{S}_m(2)=1-\frac{\delta_m}{\tilde{Z}_m},  &&\tilde{S}_m(i)=\frac{Z_m}{\tilde{Z}_m}S_m(i)\:\forall i\ge 3.
		\end{aligned}
		\end{equation}
		Starting with a single-threshold algorithm $\tau_\pi$, where $\pi=\pi_m$ on horizon $\tilde{H}_m$. By \Cref{optalg} (note that the notations $\tau$ and $\sigma$ used there are now swapped, as our focus here is on the original stopping rule) for any horizon $\tilde{H}_m$, there is a stopping rule equivalent to $\tau_\pi$ for the problem $\YY^\ssup{\tilde{H}_m}$, denoted as $\sigma_\pi\defeq\zeta_\pi\wedge (\tilde{H}_m+1)$, where $\zeta_\pi$ is, as per \Cref{equivstop}, an algorithm for the discounted problem $\DZ=\{\tilde{S}_m(i)X_i\}$, which stops only when $\tau_\pi$ successfully stops by the time that the horizon $\tilde{H}_m$ has realised. Intuitively, we can construct $\zeta_\pi$ as the stopping rule with single threshold $\pi$ on the events $\{\tau_\pi=i\}\cap\{\tilde{H}_m\ge i\}$. Otherwise it stops at $m$ on $\{\tau_\pi>\tilde{H}_m\}$. The proof of \Cref{optalg} shows that this can be done, so that $\zeta_\pi$ is adapted to $\XX$, and as a result, as $m$ grows, 
		\[\EE Y_{\tau_\pi}^\ssup{\tilde{H}_m}=\EE\sum_{i=1}^m \tilde{S}_m(i)X_i\II_{\{\zeta_\pi=i\}}\le\EE\sum_{i=1}^m S_m(i)X_i\II_{\{\zeta_\pi=i\}}=\EE Y_{\tau_\pi}^\ssup{H_m}.\]
		Indeed, by \Cref{survivals,Zratio} we have that
		\begin{align*}
			\EE Y_{\tau_\pi}^\ssup{\tilde{H}_m}&=-\frac{\delta_m}{\tilde{Z}_m}\EE X_2\II_{\{\zeta_\pi=2\}}+\EE X_1\II_{\{\zeta_\pi=1\}}+\EE X_2\II_{\{\zeta_\pi=2\}}+\frac{Z_m}{\tilde{Z}_m}\EE\sum_{i=3}^m S_m(i)X_i\II_{\{\zeta_\pi=i\}}\\
			&=\EE Y_{\tau_\pi}^\ssup{H_m}\left(1-\frac{\delta_m}{\tilde{Z}_m}\right)\le\EE Y_{\tau_\pi}^\ssup{H_m}.
		\end{align*}
		Clearly we also used $\sigma'_\pi\defeq\zeta_\pi\wedge (H_m+1)$ to obtain the second to last equality, similarly to how we used $\sigma_\pi$ in the first equality.
		
		\paragraph{Step 3.} Finally, suppose by contradiction that a constant-approximation exists on the $\PGFd$ class. Then there exist a sequence of thresholds $\{\pi_m\}$ and $c>0$ such that $\tau_{\pi_m}$ is a $\sfrac{1}{c}$-approximation on $\tilde{\mathcal{H}}_M(\eps)\subset\PGFd$. By \Cref{hardsingle} for all $\eps>0$ fixed small enough, there exists a subsequence $\{\pi_{m_j}\}$ such that 
		$r_{m_j}$, the gambler-to-prophet ratio of $\EE Y_{\tau_{\pi_{m_j}}}^\ssup{H_{m_j}}$ over $\EE M_{H_{m_j}}$, vanishes. Combining this with the previous step's asymptotics yields the following contradiction as $j$ grows:
		\begin{equation*}
			c\le\frac{ \EE Y_{\tau_{\pi_{m_j}}}^\ssup{\tilde{H}_{m_j}}}{\EE M_{\tilde{H}_{m_j}}}=\bigO( r_{m_j})\longrightarrow0.
		\end{equation*}
	\end{proof}

	\subsection{Proof of \Cref{future}}\label{supphardsecretaryfuture}
    	\begin{proof}[Proof of \Cref{future}]
		Consider the stopping rule $\tau_m\defeq\zeta_{r_m}\wedge(H_m+1)\in\TT^*$. By \Cref{optequiv} in \Cref{optalg} firstly, and the law of total expectation with respect to $\XX$ secondly, 
		\begin{equation}\label{expectedreturn}
			\EE Y_{\tau_m}^\ssup{H_m}=\EE\left( S_m(\zeta_{r_m})X_{\pi_{\zeta_{r_m}}}\right)=\EE\left\lbrace\EE\left[S_m(\zeta_{r_m})X_{\pi_{\zeta_{r_m}}}|\XX\right]\right\rbrace.
		\end{equation} 
		Having exploited the equivalent discounted deterministic horizon ($m$) problem spares us from conditioning on $H_m$, which is now out of the picture and only appears through the survival discount factor of $H_m$, $S_m(i)$. The only randomness left in the conditional expectation is that of $\pi$ and the stopping rule $\zeta_{r_m}$.
		
		Consider that 
		\begin{equation}\label{discountconditional}
			\EE\left[S_m(\zeta_{r_m})X_{\pi_{\zeta_{r_m}}}|\XX\right]=\sum_{i=r_m}^m \EE\left[S_m(i)X_{\pi_i}\II_{\{\zeta_{r_m}=i\}}|\XX\right].
		\end{equation}
		Furthermore, denoting $\mathcal{M}_j\defeq\max_{i\in[j]}\{X_{\pi_i}\}$ and noting that $\mathcal{M}_m=M_m$, conditionally on $\XX$ we have that
		\[\II_{\{\zeta_{r_m}=i\}}=\II_{\{i=\inf\{j\ge r_m:\:X_{\pi_j}=\mathcal{M}_j\}\}}\ge\II_{\{i=\inf\{j\ge r_m:\:X_{\pi_j}=M_m\}\}}\defeq\II_{W_i}.\]
		It follows that
		\begin{equation}\label{lowerconditional1}
			\EE\left[S_m(i)X_{\pi_i}\II_{\{\zeta_{r_m}=i\}}|\XX\right]\ge \EE\left[S_m(i)X_{\pi_i}\II_{W_i}|\XX\right]=S_m(i)M_m\PP(W_i|\XX).
		\end{equation}
		Plugging \Cref{lowerconditional1} into \Cref{discountconditional} yields
		\begin{equation}\label{lowerconditional2}
			\EE\left[S_m(\zeta_{r_m})X_{\pi_{\zeta_{r_m}}}|\XX\right]\ge M_m\sum_{i=r_m}^m S_m(i)\PP(W_i|\XX).
		\end{equation}
		By \Cref{asympnorm} and comparison of the sum with the corresponding integral as in \Cref{asympmax_H_m}, we can estimate the cdf of $H_m$ \[F_m(i-1)=\frac{\sum_{h=\ell}^{i-1} h^{-\frac{1}{1+2\eps}}}{\sum_{h=\ell}^{m} h^{-\frac{1}{1+2\eps}}}\] as  $m\longrightarrow\infty$, since for all $i>r_m\sim\sfrac{m}{e}$ (the lower bound's asymptotics follows from \Cref{waitingtime}) we have that
		\begin{equation}\label{const}
			F_m(i-1)\le \frac{\int_{\ell-1}^{i-1} x^{-\frac{1}{1+2\eps}}dx}{\int_{\ell-1}^{m+1} x^{-\frac{1}{1+2\eps}}dx}=\frac{x^{\frac{2\eps}{1+2\eps}}\Big\vert_{\ell-1}^{i-1}}{x^{\frac{2\eps}{1+2\eps}}\Big\vert_{\ell-1}^{m+1}}=\left(\frac{i-1}{m+1}\right)^{\frac{2\eps}{1+2\eps}}\frac{1-\left(\frac{\ell-1}{i-1}\right)^{\frac{2\eps}{1+2\eps}}}{1-\left(\frac{\ell-1}{m+1}\right)^{\frac{2\eps}{1+2\eps}}}<\left(\frac{i-1}{m+1}\right)^{\frac{2\eps}{1+2\eps}}.
		\end{equation}
		Thus as $m\longrightarrow\infty$, \[S_m(i)=1-F_m(i-1)>1- \left(\frac{i-1}{m+1}\right)^{\frac{2\eps}{1+2\eps}}.\]
		In conclusion, as $m\longrightarrow\infty$, as $\eps\longrightarrow0$,
		\begin{align*}
			\sum_{i=r_m}^m S_m(i)\PP(W_i|\XX)&\ge \sum_{i=r_m}^m \PP(W_i|\XX)-\sum_{i=r_m}^m \left(\frac{i-1}{m+1}\right)^{\frac{2\eps}{1+2\eps}}\PP(W_i|\XX)\\&=\PP(W|\XX)-\frac{1}{(m+1)^{\frac{2\eps}{1+2\eps}}}\frac{r_m-1}{m}\sum_{i=r_m}^m \left(i-1\right)^{-\frac{1}{1+2\eps}}\\&\ge\frac{1}{e}-\frac{1}{(m+1)^{\frac{2\eps}{1+2\eps}}}\frac{r_m-1}{m}\int_{r_m-2}^{m-1} x^{-\frac{1}{1+2\eps}}dx\longrightarrow g(\eps),
		\end{align*}
		where we used \Cref{W,W_i} in the equality and \Cref{1/e} and the usual integral estimate in the second last inequality, whereas the limiting behaviour as $m\longrightarrow\infty$ follows by \Cref{waitingtime}, having defined
		\begin{align*}
			f_m(\eps)&\defeq\frac{1}{(m+1)^{\frac{2\eps}{1+2\eps}}}\frac{r_m}{m}\int_{r_m-2}^{m-1} x^{-\frac{1}{1+2\eps}}dx=\frac{1+2\eps}{2\eps}\frac{1}{(m+1)^{\frac{2\eps}{1+2\eps}}}\frac{r_m}{m} x^{\frac{2\eps}{1+2\eps}}\Big\vert_{r_m-2}^{m-1}\longrightarrow f(\eps)\\&\defeq \frac{1+2\eps}{2\eps}\frac{1}{e}\left(1-\frac{1}{e^{\frac{2\eps}{1+2\eps}}}\right)=\begin{cases}\frac{1}{e}-\frac{1}{e^2}+\bigO\left(\frac{1}{\eps}\right),&\eps\longrightarrow\infty\\\frac{1}{e}+\bigO\left(\eps\right),&\eps\longrightarrow0,\end{cases}
		\end{align*}
		and \[g(\eps)\defeq\frac{1}{e}-\frac{1+2\eps}{2\eps}\frac{1}{e}\left(1-\frac{1}{e^{\frac{2\eps}{1+2\eps}}}\right)\defeq\frac{1}{e}-f(\eps)=\bigO(\eps)>0\] as $\eps\longrightarrow 0^+$. It is crucial to note that $g(\eps)$ it is positive and strictly increasing on $(0,\infty)$, since it can be rewritten as $g(\eps)=\sfrac{1}{e}\left(1-h(t(\eps))\right)$, having defined \[t(\eps)\defeq\frac{2\eps}{1+2\eps},\quad h(t)\defeq\frac{1-e^{-t}}{t}.\] The monotonicity then follows from the fact that $t(\eps)$ is monotone increasing and upper-bounded by $1$, which ensures that $h'(t)<0$.\footnote{As $\eps$ grows large, $g(\eps)\longrightarrow e^{-2}$. However, the regime $\eps$ large is not relevant to the hardness of single-thresholds, which is only ensured by $\eps$ small enough.} Plugging the bound obtained into \Cref{lowerconditional2} yields, as $m\longrightarrow\infty$,
		\[\EE\left[S_m(\zeta_{r_m})X_{\pi_{\zeta_{r_m}}}|\XX\right]\ge (g(\eps)-\delta)M_m,\]
		where $\delta\longrightarrow 0^+$ as $m\longrightarrow\infty$. Plugging this into \Cref{expectedreturn} yields
		\[\EE X_{\tau_m}\defeq\EE Y_{\tau_m}^\ssup{H_m}\ge (g(\eps)-\delta)\EE M_m.\]
		
		In conclusion, consider the family thus structured. For a fixed $\eps>0$ small enough, we allow any instance for the values $X$ to undergo the family of horizons $\{H_m\}_{m\ge M}$ as defined in \Cref{hardhorizon}, where $M=M(\eps)\in\NN$ is large enough. By \Cref{hardsingle}, the RH problem under this family does not admit single-threshold constant-approximations for the fixed $\eps>0$, since for any sequence of thresholds $\{\pi_m\}_{m\ge M}$, we can assume $\eps$ small enough, so that there exists a vanishing subsequence \[r_{m_j}\defeq\frac{\EE X_{\tau_{\pi_{m_j}}}}{\EE M_{H_{m_j}}},\] as long as we consider the value distribution $X$ defined in \Cref{hardval} (recall that here $\tau_{\pi_m}$ refers to the single-threshold algorithm facing horizon $H_m$ and value $X$ with threshold $\pi_m$). However, by the above, $\tau_m$ provides a constant-approximation. More specifically, regardless of how small $\eps$ is fixed, if we set $M$ large enough, $\tau_m$ constitutes an approximate $g(\eps)$-approximation, since using $H_m\le m$ we obtain that for any random variable $X$
		\[\frac{\EE X_{\tau_m}}{\EE M_{H_m}}\ge (g(\eps)-\delta)\frac{\EE M_m}{\EE M_m}=g(\eps)-\delta,\]
		where $\delta\longrightarrow 0$ as $M\longrightarrow\infty$.
	\end{proof}
	
	\section{Proofs omitted from Section \ref{L2}}\label{suppL2}
	\begin{proof}[Proof of \Cref{L2horizon}]
		In the argument of \Cref{Ghorizon} we showed that for general distributions the competitive ratio of the single-threshold algorithm used there is \[c_{p,\bar{q}}=1-\EE\left[\left(1-\frac{1}{\mu}\right)^H\right]=1-\EE\left( e^{-\bar{s}H}\right)\] with $\bar{s}\defeq-\log(1-\sfrac{1}{\mu})$. Since $G\prec_{Lt}H$, where $G$ is the two-point distribution of \Cref{bernoulli}, that is
		\[G\sim\begin{cases}0,&\text{w.p.}\:\frac{\sigma^2}{\mu^2+\sigma^2}\\
			\frac{\mu^2+\sigma^2}{\mu},&\text{w.p.}\:\frac{\mu^2}{\mu^2+\sigma^2},\end{cases}\]
		it follows that
		\[c_{p,\bar{q}}\ge1-\EE\left( e^{-\bar{s}G}\right)=\frac{\mu^2}{\mu^2+\sigma^2}\left[1-\left(1-\frac{1}{\mu}\right)^\frac{\mu^2+\sigma^2}{\mu}\right].\]
		As a result an inequality that yields a single-threshold $2-\sfrac{1}{\mu}$-approximation is
		\[\frac{\mu^2}{\mu^2+\sigma^2}\left[1-\left(1-\frac{1}{\mu}\right)^\frac{\mu^2+\sigma^2}{\mu}\right]\ge\frac{1}{2-\frac{1}{\mu}},\]
		and it is equivalent to
		\begin{equation}\label{Lamb1}
			\left(e^{\mu\log\left(1-\frac{1}{\mu}\right)}\right)^{x}\le1-\frac{x}{2-\frac{1}{\mu}},
		\end{equation}
		having set $x\defeq1+\sfrac{\sigma^2}{\mu^2}$. Since $\log(1-\sfrac{1}{\mu})\le-\sfrac{1}{\mu}$,
		\[ e^{\mu\log\left(1-\frac{1}{\mu}\right)}\le e^{-1}.\] As a result an even simpler inequality that yields a single-threshold $2-\sfrac{1}{\mu}$-approximation is 	
		\begin{equation}\label{Lamb2}
			e^{-x}\le1-\frac{x}{2-\frac{1}{\mu}},
		\end{equation}
		which is equivalent to 
		\begin{equation}\label{Lamb3}
			ye^y\le\bar{y}e^{\bar{y}},
		\end{equation}
		with constraints $y\defeq x-(2-\sfrac{1}{\mu})>-1$ and $-2<\bar{y}\defeq-(2-\sfrac{1}{\mu})<-1$. The solution to the equation $ye^y=z$ is the Lambert function. Over the reals the Lambert function is described through a principal branch, denoted as $W_0(z)$, and a secondary branch, denoted as $W_{-1}(z)$, since
		 $y=-1$ is a branching point for the solution. Since we constrained \Cref{Lamb3} on $y>-1$, $ye^y>-e^{-1}$ is monotone increasing, and for $-2<\bar{y}<-1$, $-1e^{-1}<\bar{y}e^{\bar{y}}<-2e^{-2}$, we can equivalently formulate \Cref{Lamb3} in terms of $W_0$, namely \[y\le W_0(\bar{y}e^{\bar{y}}).\] Recalling the definition of $y$ and $x$, this yields \Cref{concentration}.
	\end{proof}
\end{document}